%% file: main.tex
\documentclass[11pt]{article}

\newif\ifcomments
\commentstrue

\input{header}

\title{Connectivity Labeling Schemes for Edge and Vertex Faults\\via Expander Hierarchies}

\author{Yaowei Long\\ University of Michigan
\and 
Seth Pettie\thanks{Supported by NSF Grant CCF-2221980.}
\\ University of Michigan
\and 
Thatchaphol Saranurak\thanks{Supported by NSF Grant CCF-2238138.}
\\ University of Michigan}

\date{}

\begin{document}

\maketitle
\pagenumbering{gobble}
\input{0_abstract}

\newpage
\tableofcontents

\newpage

\pagenumbering{arabic}

\input{1_intro_new}
\input{2_edge_det_simple_new}

\input{2.5_shorter_new}

\input{vertex_failure_conn_labeling_scheme}

\input{4_edge_rand}

\input{5_vertex_rand}
\input{6_lower_bound}

\input{conclusion}

\section*{Acknowledgments}
\addcontentsline{toc}{section}{Acknowledgments}

Yaowei Long and Thatchaphol Saranurak are partially funded by the Ministry of Education and Science of Bulgaria's support for INSAIT, Sofia University ``St. Kliment Ohridski'' as part of the Bulgarian National Roadmap for Research Infrastructure.

\appendix

\input{a_low_deg_tree}

\input{b_singleton_detection}

\input{main.bbl}

\end{document}

%% file: header.tex
\usepackage{fullpage}
\usepackage[utf8]{inputenc}

\usepackage{graphicx}
\usepackage{amsmath,amsthm,amssymb}
\usepackage{bbm}
\usepackage{mathtools}
\usepackage{multirow}
\usepackage{makecell}

\usepackage{thmtools}
\usepackage{thm-restate,color,xspace}
\usepackage{comment}

\usepackage{xcolor}
\usepackage{nameref}
\usepackage{array}

\definecolor{ForestGreen}{rgb}{0.1333,0.5451,0.1333}
\definecolor{DarkRed}{rgb}{0.65,0,0}
\definecolor{Red}{rgb}{1,0,0}
\usepackage[linktocpage=true,pagebackref=true,colorlinks,linkcolor=DarkRed,citecolor=ForestGreen,bookmarks,bookmarksopen,bookmarksnumbered]{hyperref}

\usepackage[capitalise]{cleveref}

\DeclarePairedDelimiter{\ceil}{\lceil}{\rceil}
\DeclarePairedDelimiter{\floor}{\lfloor}{\rfloor}

\DeclarePairedDelimiter{\ang}{\langle}{\rangle}

\declaretheorem[numberwithin=section]{theorem}
\declaretheorem[numberlike=theorem]{lemma}
\declaretheorem[numberlike=theorem,name=Lemma]{lem}
\declaretheorem[numberlike=theorem]{fact}

\declaretheorem[numberlike=theorem]{corollary}

\declaretheorem[numberlike=theorem]{conjecture}

\declaretheorem[numberlike=theorem]{claim}

\declaretheorem[numberlike=theorem]{observation}

\declaretheorem[numberlike=theorem,style=definition]{definition}

\declaretheorem[numberlike=theorem,name=Definition,style=definition]{defn}
\declaretheorem[numberlike=theorem,name=Open Question]{question}

\newcommand{\rb}[2]{\raisebox{#1 mm}[0mm][0mm]{#2}}

\newcommand{\emphbf}[1]{\emph{\bfseries{#1}}}

\newcommand{\ind}[1]{\mathbbm{1}\left({#1}\right)}

\newcommand{\Var}{\mathbb{V}}
\newcommand{\E}{\mathbb{E}}
\newcommand{\uid}{\mathsf{uid}}
\newcommand{\singleton}{\mathsf{singleton}}
\newcommand{\rank}{\mathsf{rank}}
\newcommand{\sk}{\mathsf{sk}}
\newcommand{\lge}{\mathsf{lge}}
\newcommand{\LGE}{\mathsf{LGE}}

\newcommand{\wt}{\mathsf{wt}}
\newcommand{\bydef}{\stackrel{\operatorname{def}}{=}}
\newcommand{\Ball}{\mathsf{Ball}}

\newcommand{\dist}{\mathsf{dist}}
\newcommand{\DFS}{\mathsf{DFS}}

\newcommand{\Clique}{\mathsf{Clique}}
\newcommand{\Unite}{\mathsf{Unite}}

\newcommand{\Z}{\mathbb{Z}}

\newcommand{\F}{\mathbb{F}}

\global\long\def\wtilde{\widetilde}
\global\long\def\affected{{\rm affected}}

\global\long\def\sp{{\rm sp}}

\global\long\def\root{\mathsf{root}}
\global\long\def\Euler{\mathsf{Euler}}

\global\long\def\pos{\mathsf{pos}}

\global\long\def\larr{\leftarrow}
\global\long\def\sstart{\mathsf{start}}
\global\long\def\eend{\mathsf{end}}
\global\long\def\prefixSum{\mathsf{PrefixEnum}}
\global\long\def\exc{\rm exc}
\global\long\def\inc{\rm inc}

\global\long\def\deg{\mathsf{Deg}}
\global\long\def\aff{\rm aff}
\global\long\def\unaff{\rm unaff}
\global\long\def\valid{\rm valid}
\global\long\def\qry{\rm qry}
\global\long\def\edgenum{\mathsf{Enum}}
\global\long\def\incidentedge{\mathsf{IncidentEdge}}
\global\long\def\ArtificialE{\mathsf{ArtificialE}}
\global\long\def\giant{{\rm giant}}
\global\long\def\ext{\rm ext}
\global\long\def\arbo{\rm arbo}
\global\long\def\child{\rm child}
\global\long\def\Conn{\mathsf{Conn}}
\global\long\def\nd{\rm nd}

\global\long\def\poly{\mathrm{poly}}

\global\long\def\tmn{\mathrm{tmn}}
\global\long\def\oc{\mathrm{oc}}

\global\long\def\nb{\mathrm{nb}}
\global\long\def\ListAffectedComps{\mathsf{ListAffectedComps}}
\global\long\def\ListSubtrees{\mathsf{ListSubtrees}}
\global\long\def\ListTerminals{\mathsf{ListTerminals}}
\global\long\def\ListNeighbors{\mathsf{ListNeighbors}}
\global\long\def\IsTerminal{\mathsf{IsTerminal}}
\global\long\def\EnumFromGiant{\mathsf{EnumFromGiant}}
\global\long\def\PickTerminal{\mathsf{PickTerminal}}
\global\long\def\profile{\mathsf{profile}}
\global\long\def\id{\mathsf{id}}
\global\long\def\boundary{\mathsf{NeighborEdge}}
\global\long\def\NeighborEdge{\mathsf{NeighborEdge}}
\global\long\def\NeighborVertex{\mathsf{NeighborVertex}}

\global\long\def\InnerTerminals{\mathsf{InnerTerminals}}
\global\long\def\occurrences{\mathsf{occurrences}}
\global\long\def\SuccTerminal{\mathsf{SuccTerminal}}
\global\long\def\Degree{\mathsf{Degree}}
\global\long\def\Enum{\mathsf{Enum}}
\global\long\def\ArtificialE{\mathsf{ArtificialEdge}}
\global\long\def\ccolor{\mathsf{color}}

\newcommand{\level}{\mathsf{level}}

\newcommand{\Patrascu}{P{\v{a}}tra\c{s}cu\xspace}

\newcommand{\Boruvka}{Bor\r{u}vka\xspace}

\newcommand{\ignore}[1]{}

%% file: 0_abstract.tex
\begin{abstract}
We consider the problem of assigning short \emph{labels} to 
the vertices and edges of a graph $G$ so that given any 
query $\ang{s,t,F}$ with $|F|\leq f$, we can determine whether $s$ and $t$ are still connected in $G-F$, given only the labels of $F\cup\{s,t\}$.

This problem has been considered when $F\subset E$ (edge faults),
where correctness is guaranteed with high probability (w.h.p.) \cite{DoryP21} or deterministically~\cite{IzumiEWM23}, 
and when $F\subset V$ (vertex faults), both w.h.p.~and deterministically~\cite{ParterP22a,ParterPP24}.
Our main results are as follows.

\begin{description}
\item[Deterministic Edge Faults.] We give a new deterministic labeling scheme for edge faults that uses $\tilde{O}(\sqrt{f})$-bit labels, which can be constructed in polynomial time.  This improves on Dory and Parter's~\cite{DoryP21} \emph{existential} bound of 
$O(f\log n)$ (requiring exponential time to compute) and the efficient $\tilde{O}(f^2)$-bit scheme of 
Izumi, Emek, Wadayama, and Masuzawa~\cite{IzumiEWM23}.
Our construction uses an improved 
edge-expander hierarchy 
and a distributed coding technique based on 
Reed-Solomon codes.

\item[Deterministic Vertex Faults.] We improve Parter, Petruschka, and Pettie's~\cite{ParterPP24} deterministic $O(f^7\log^{13} n)$-bit labeling scheme for vertex faults to $O(f^4\log^{7.5} n)$ bits, 
using an improved vertex-expander hierarchy and better sparsification of shortcut graphs. We completely bypass deterministic graph sketching \cite{IzumiEWM23} and hit-and-miss families~\cite{KarthikP21}.

\item[Randomized Edge/Verex Faults.] 
We improve the size of Dory and Parter's~\cite{DoryP21} randomized 
edge fault labeling scheme 
from $O(\min\{f+\log n, \log^3 n\})$ bits to $O(\min\{f+\log n, \log^2 n\log f\})$ bits, shaving a $\log n/\log f$ factor. 
We also improve the size of Parter, Petruschka, and Pettie's~\cite{ParterPP24}
randomized vertex fault labeling scheme
from $O(f^3\log^5 n)$ bits 
to $O(f^2\log^6 n)$ bits, 
which comes closer to their $\Omega(f)$-bit lower bound~\cite{ParterPP24}.
\end{description}

\end{abstract}

%% file: 1_intro_new.tex
\section{Introduction}

\label{sec:introduction}

A \emph{labeling scheme} for a graph problem can be viewed as a distributed data structure in which all queries must be answered without inspecting the underlying graph, but only the \emph{labels} of the query arguments. Early work focused on labeling schemes for \emph{adjacency}~\cite{Breuer66,BreuerF67,KannanNR92}, which is connected to finding small \emph{induced universal graphs}~\cite{AlstrupDK17,AlstrupKTZ19}.

There are now many labeling schemes for basic navigation queries in rooted trees~\cite{AlstrupDK17,AbiteboulAKMR06,AlstrupHL14}, such as adjacency, ancestry, and least common ancestors. There are labeling schemes for computing \emph{distances} in general graphs~\cite{AlstrupGHP16b}, planar graphs~\cite{GavoillePPR04,BonamyGP22,GawrychowskiU23}, and trees~\cite{GavoillePPR04,AlstrupBR05,AlstrupGHP16a}, as well as labelings for approximate distances~\cite{TZ05,AbrahamG11}. There are several labeling schemes for answering queries about the pairwise edge- and vertex-connectivity in undirected graphs \cite{KatzKKP04,HsuL09,IzsakN12,PettieSY22}. 

\paragraph{Label Schemes under Faults.}

Courcelle and Twigg~\cite{CourcelleT07,CourcelleGKT08} initiated the study of \emph{forbidden set} or \emph{fault tolerant} labeling schemes. The idea is to support a standard connectivity/distance query, subject to \emph{faults} (deletions) of some subset $F$ of vertices or edges. Several fault tolerant labeling schemes focus on special graph classes such as bounded treewidth graphs \cite{CourcelleT07} and planar graphs \cite{CourcelleGKT08,AbrahamCGP16,Bar-NatanCGMW22,chechik2023optimal}. On general graphs, labeling schemes that handle at most one or two faults were shown for single-source reachability \cite{choudhary2016optimal} and single-source approximate distance \cite{BaswanaCHR20}. 

The first labeling scheme on general graphs under multiple faults was given by Dory and Parter~\cite{DoryP21}, for connectivity under edge faults. 
More precisely, they assigned labels to edges and vertices of an undirected $n$-vertex graph $G$ so that, given the labels of $\ang{s,t,F}$ where $F\subset E(G)$, one can determine if $s$ and $t$ are still connected in $G-F$. When $|F|\leq f$, they gave a Monte Carlo randomized construction of labels of size $O(\min\{f+\log n,\log^{3}n\})$ bits that answer \emph{each} query correctly with high probability. By increasing the size to $O(f\log n)$ bits, their scheme answers \emph{all} queries correctly, with high probability, though confirming this property seems to require an exponential time brute force search. Recently, Izumi, Emek, Wadayama, and Masuzawa~\cite{IzumiEWM23} gave a deterministic polynomial-time construction of labels of size $\tilde{O}(f^{2})$.

Parter and Petruschka~\cite{ParterP22a} considered the same problem, but with \emph{vertex} faults rather than edge faults, i.e., $F\subset V(G)$. They gave deterministic labeling schemes with length $O(\log n)$ for $f=1$, $O(\log^{3}n)$ for $f=2$, and a randomized scheme of length $\tilde{O}(n^{1-1/2^{f-2}})$ for general $f$. This year, Parter, Petruschka, and Pettie~\cite{ParterPP24} developed randomized and deterministic labeling schemes with label length $\tilde{O}(f^{3})$ and $\tilde{O}(f^{7})$, respectively. They also observed an $\Omega(f+\log n)$-bit lower bound for vertex faults (randomized or deterministic), which established a complexity separation between edge and vertex faults. See \cref{tab:prior-results}.

\begin{table}
\begin{tabular}{|l|l|l|l|}
\multicolumn{3}{l}{\bfseries{\textsf{\large Edge Fault Tolerant Connectivity Labels}}}\\
\multicolumn{1}{l}{\textsf{Reference}}
&
\multicolumn{1}{l}{\textsf{Label Size (Bits)}}
&
\multicolumn{1}{l}{\textsf{Guarantee}}
&
\multicolumn{1}{l}{\textsf{Notes}}\\\hline
\rb{-3}{Dory \& Parter~\cite{DoryP21}} &  $O(\min\{f + \log n, \log^3 n\})$ & Monte Carlo & Query correct w.h.p.\\\cline{2-4}
               &  $O(f\log n)$                      & Deterministic  & \emph{Existential bound}\\\hline
Izumi, Emek, Wadayama  & $O(f^2\log^2 n\log\log n)$ & \rb{-3}{Deterministic} & Polynomial construction\\
\& Masuzawa~\cite{IzumiEWM23}    & $O(f^2\log^3 n)$  &  & $\tilde{O}(mf^2)$ 
construction\\\hline
Trivial     & $\Omega(\log n)$ & any & trivial lower bound\\\hline
       &  $O(\min\{f+\log n,$    & \rb{-3}{Monte Carlo} & \rb{-3}{Query correct w.h.p.}\\
\rb{-3}{\bf{new}}       & \hfill $\log^2 n\log(f/\log^2 n)\})$ &&\\\cline{2-4}
         &  $O(\sqrt{f}\log^{2} n\log f)$                   &  & \emph{Existential bound}\\
                &  $O(\sqrt{f}\log^{2.25} n\log (f\log n))$          & \rb{3.5}{Deterministic} & Polynomial construction \\\hline\hline
\multicolumn{3}{l}{}\\
\multicolumn{3}{l}{\bfseries{\textsf{\large Vertex Fault Tolerant Connectivity Labels}}}\\\hline
Parter & $O(\log^3 n)$      & Deterministic & $f\leq 2$\\\cline{2-4}
\& Petruschka~\cite{ParterP22a}     & $\tilde{O}(n^{1-2^{-f+2}})$   & Monte Carlo   & $f\in[3,o(\log\log n)]$\\\hline
\rb{-3}{Parter, Petruschka}        &   $O(f^3\log^5 n)$            & Monte Carlo & Query correct w.h.p.\\\cline{2-4}
\rb{-3}{\& Pettie~\cite{ParterPP24}}    &   $O(f^7\log^{13} n)$  & Deterministic & Polynomial construction\\\cline{2-4}
            & $\Omega(f + \log n)$ & any & lower bound\\\hline
            & $O(f^2\log^6 n)$ & Monte Carlo & Query correct w.h.p.\\\cline{2-4}
   
{\bf new} &  $O(f^{4}\log^{7} n)$       &  & \emph{Existential bound}\\
 &  $O(f^4\log^{7.5} n)$       & \rb{3.5}{Deterministic} & Polynomial construction\\\hline\hline
\end{tabular}
\caption{\label{tab:prior-results}All \emph{Monte Carlo} results 
have a one-sided error probability of $1/\poly(n)$, i.e.,
they may report two vertices disconnected when they are, in fact, connected.  The \emph{existential} result of Dory and Parter~\cite{DoryP21} constructs labels in $O(mf\log n)$ time that, with high probability, answer all queries correctly. However, to verify this fact requires a brute force search.
The new existential results require solving an NP-hard problem, namely \emph{sparsest cut}.}
\end{table}

\paragraph{Our Results.}

We show new connectivity labels under both edge and vertex faults that improve the state-of-the-art as follows:
\begin{enumerate}
\item \label{result:det edge} Deterministic labels under edge faults of size $\tilde{O}(\sqrt{f})$ bits (\Cref{thm:shorter-deterministic-labels}). This simultaneously improves Dory and Parter's~\cite{DoryP21} labels of size $O(f\log n)$, 
which require exponential time to construct, 
and Izumi et al.'s~\cite{IzumiEWM23} efficiently 
constructed labels of size $\tilde{O}(f^{2})$. 
In fact, Izumi et al.~\cite{IzumiEWM23} stated that 
``it seems plausible that the $\Omega(f)$-bit lower bound holds.''  We refute this possibility.
\item \label{result:det vertex} Deterministic labels under vertex faults of size $O(f^{4}\log^{7.5}n)$ bits (\Cref{thm:vertex label}).
This improves the $O(f^{7}\log^{13}n)$-bit labels of \cite{ParterPP24}.
\item \label{result:rand edge} Randomized labels under edge faults of size $O(\min\{f+\log n,\log^{2}n\log f\})$ bits (\Cref{thm:edge-failure-long-labels,thm:rand-edge-labeling}). 
This improves the $O(\min\{f+\log n,\log^{3}n\})$-bit 
labels of Dory and Parter~\cite{DoryP21}.
\item \label{result:rand edge} Randomized labels under vertex faults of size $O(f^2 \log^6 n)$ bits (\Cref{thm:rand-vertex-labeling}). This improves the $O(f^3\log^5 n)$-bit labels of Parter et al.~\cite{ParterPP24}.
\end{enumerate}

\paragraph{Related Work: Connectivity Oracles.}

Our connectivity labels can be viewed as a distributed version of \emph{connectivity oracles under faults}. In this problem, we must build a centralized data structure for an input graph $G$ so that, given a query $\ang{s,t,F}$, we can check if $s$ and $t$ are still connected in $G-F$ as fast as possible using the centralized data structure. Connectivity oracles have been well-studied under both edge faults \cite{PatrascuT07,DuanP20,GibbKKT15} 
and vertex faults \cite{DuanP20,BrandS19,LongS22,PilipczukSSTV22,kosinas2023connectivity,long2024better} 
and the optimal preprocessing/query bounds have been proven \cite{PatrascuT07,LongS22,DuanP20,KopelowitzPP16,HenzingerKNS15},
either unconditionally, or conditioned on standard fine-grained complexity assumptions.

\paragraph{Our Techniques.}

The basis of Results \ref{result:det edge} and \ref{result:det vertex} is an \emph{expander hierarchy}. This is the first application of expanders in the context of fault-tolerant labeling schemes, although they have been widely applied in the centralized dynamic setting, e.g.~\cite{NanongkaiSW17,goranci2021expander,LongS22}. 

We give a clean definition of expander hierarchies for both \emph{edge expansion} (\Cref{def:edge hie}) and \emph{vertex expansion }(\Cref{def:vertex hie}) as well as simple algorithms for computing them. In the edge expansion version, our formulation turns out to be equivalent to \Patrascu and Thorup's \cite{PatrascuT07}, but our algorithm improves the quality by a $\Theta(\log n)$ factor. Our vertex expander hierarchy is new. 
Combined with the observation that any $\phi$-vertex-expanding set has an $O(1/\phi)$-degree 
Steiner tree  (\Cref{lemma:LowDegreeSteinerTree}),\footnote{Previously, \cite{LongS22} showed an $O(\log(n)/\phi)$-degree Steiner tree spanning any $\phi$-vertex-expanding set. } this implies a new \emph{low-degree hierarchy} that is strictly stronger 
and arguably cleaner than 
all previous low-degree hierarchies \cite{DuanP20,LongS22,long2024better,ParterPP24}, 
which are the critical structures behind vertex fault tolerant connectivity oracles.

To obtain our deterministic labels under edge faults, we first show that the edge expander hierarchy immediately leads to a simple $\tilde{O}(f)$-bit label (\Cref{thm:det-edge-labeling}), which already improves the state-of-the-art~\cite{IzumiEWM23,DoryP21}. Then, we introduce a new distributed coding technique based on 
Reed-Solomon codes to improve the label size to $\tilde{O}(\sqrt{f})$ bits, obtaining Result \ref{result:det edge}.

Our deterministic labeling scheme under vertex faults (Result \ref{result:det vertex}) employs a high-level strategy from Parter et al.~\cite{ParterPP24}, 
but our label is shorter by an $\tilde{\Theta}(f^{3})$ factor. Roughly, this improvement comes from two sources. 
First, Parter et al.~\cite{ParterPP24} employed the deterministic graph cut sketch of Izumi et al.~\cite{IzumiEWM23}, which contributes an $\tilde{\Theta}(f^{2})$ factor to the size.  We can bypass deterministic sketching and pay only an $\tilde{O}(f)$ factor because our low-degree hierarchy has an additional vertex expansion property. Second, Parter et al.~\cite{ParterPP24} constructed 
a \emph{sparsified shortcut graph} with arboricity $\tilde{O}(f^{4})$ using the 
hit-and-miss families by Karthik and Parter~\cite{KarthikP21}. 
We are able to use 
the simpler 
Nagamochi-Ibaraki sparsification \cite{NagamochiI92} 
to obtain a sparse 
shortcut graph with arboricity $\tilde{O}(f^{2})$. 
These two improvements cannot be applied in a modular way, 
so our final scheme ends up being rather different from \cite{ParterPP24}. 

All previous centralized connectivity oracles under vertex faults (including the one using vertex expanders \cite{LongS22}) crucially used 2D-range counting data structures, which seem inherently incompatible with the distributed labeling setting. Thus, our scheme is inherently 
different than the centralized oracles~\cite{LongS22,DuanP20}.

%% file: 2_edge_det_simple_new.tex
\section{Deterministic Edge Fault Connectivity Labels}

\label{sec:det-edge-faults}

The goal of this section is to prove the following theorem.
\begin{theorem}
\label{thm:shorter-deterministic-labels} Fix any undirected graph $G=(V,E)$ and integer $f\geq1$. There are deterministic labeling functions $L_{V}:V\to\{0,1\}^{\log n}$ and $L_{E}:E\to\{0,1\}^{O(\sqrt{f/\phi}\log(f/\phi)\log^{2}n)}$ such that given any query $\ang{s,t,F}$, $F\subset E$, $|F|\leq f$, one can determine whether $s$ and $t$ are connected in $G-F$ by inspecting only $L_{V}(s),L_{V}(t),\{L_{E}(e)\mid e\in F\}$. The construction time is exponential for $\phi=1/2$ and polynomial for $\phi=\Omega(1/\sqrt{\log n})$.
\end{theorem}

We remark that the above labeling scheme is actually more flexible. By reading only the labels of the failed edges $F$, it can compute a representation of connected components of $G-F$ in $\poly(f\log n)$ time. From this representation, we can, for example, count the number of connected components in $G-F$.
This is impossible in the vertex-failure setting for any 
vertex-labels of size $o(n^{1-1/f}/f)$. See \Cref{sec:lowerbound}.
Given the additional labels of $s$ and $t$, we can then check whether $s$ and $t$ are connected in $G-F$, in $O(1+\min\{\frac{\log\log n}{\log\log\log n},\,\frac{\log f}{\log\log n}\})$ time. We can also straightforwardly handle edge insertions.

To prove \Cref{thm:shorter-deterministic-labels}, we introduce two new tools into the context of labeling schemes. The first tool is the 
\emph{edge expander hierarchy}, for which we give an improved construction in \Cref{sec:edge exp hie}. 
This tool alone already leads to a simple and efficient deterministic labeling scheme of size $\tilde{O}(f)$ bits, 
improving prior work~\cite{IzumiEWM23,DoryP21}. 
In \Cref{sec:code share} we introduce a second 
tool, distributed \emph{code shares}, 
and in \Cref{sec:shorter edge label}
we combine the two tools and 
prove \Cref{thm:shorter-deterministic-labels}.

\subsection{First Tool: Edge Expander Hierarchies}\label{sec:edge exp hie}

In this section we recall \Patrascu{} and Thorup's~\cite{PatrascuT07} definition of an expander hierarchy, then give a new construction that improves the quality by a factor of $\Theta(\log n)$.

Given a graph $G=(V,E)$, a set $X\subseteq E$ and a vertex $u$, let $\deg_{X}(u)$ denote the number of edges from $X$ incident to $u$ and let 
$\deg_{X}(S)=\sum_{u\in S}\deg_{X}(u)$ denote the \emph{volume} of $S$ with respect to $X$. We say that $X$ is \emph{$\phi$-expanding} in $G$ if, for every cut $(S,V\setminus S)$, 
\[
|E_{G}(S,V\setminus S)|\ge\phi\min\{\deg_{X}(S),\deg_{X}(V\setminus S)\}.
\]
Consider a partition $\{E_{1},\dots,E_{h}\}$ of $E$. We denote $E_{\le\ell}:=\bigcup_{i\le\ell}E_{i}$, 
$E_{>\ell}:=\bigcup_{i>\ell}E_{i}$,
and $G_{\le\ell}:=G\cap E_{\le\ell}$.
We also write $\deg_{\ell}:=\deg_{E_{\ell}}$, $\deg_{\le\ell}:=\deg_{E_{\le\ell}}$, and so on.

\begin{defn}
[Expander Hierarchy]Given a graph $G=(V,E)$, an edge-partition ${\cal P}=\{E_{1},\dots,E_{h}\}$ of $E$ \emph{induces an $(h,\phi)$-expander hierarchy} of $G$ if, for every level $\ell\le h$ and every connected component 
$\Gamma$ of $G_{\leq\ell}$, $E_{\ell}\cap\Gamma$ is $\phi$-expanding in $\Gamma$. That is, for every cut $(S,\Gamma\setminus S)$ of $\Gamma$, we have 
\[
|E_{\le\ell}(S,V(\Gamma)\setminus S)|\ge\phi\min\{\deg_{\ell}(S),\deg_{\ell}(V(\Gamma)\setminus S)\}.
\]
Over all levels $\ell$, the set of all connected components $\Gamma$ of $G\setminus E_{>\ell}$ form a laminar family ${\cal C}$. Let ${\cal H}$ be the tree representation of ${\cal C}$. We also call $({\cal C},{\cal H})$ an $(h,\phi)$-expander hierarchy of $G$. 
\label{def:edge hie}
\end{defn}

\begin{figure}
\begin{centering}
\includegraphics[width=1\textwidth]{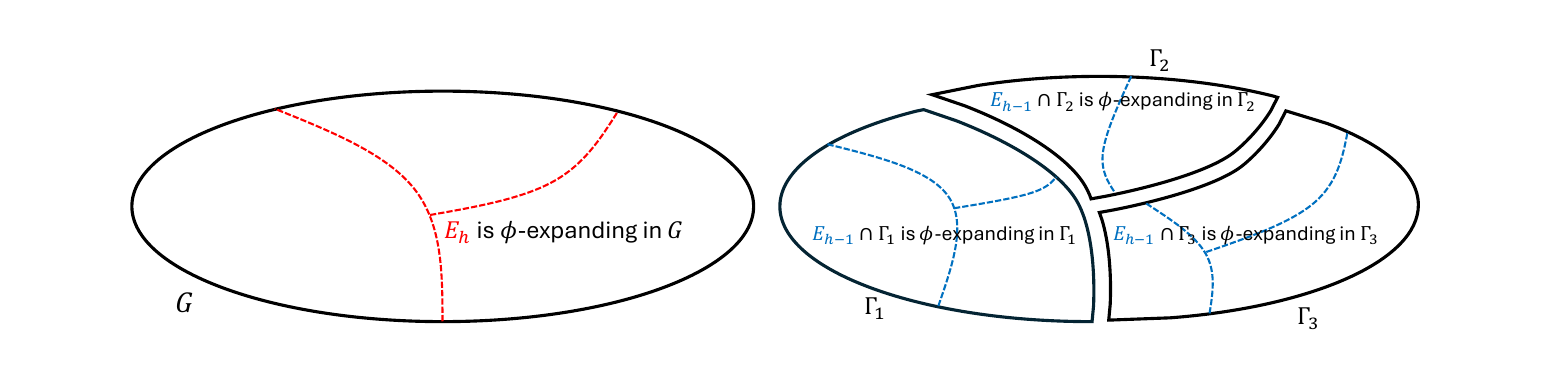}
\par\end{centering}
\caption{Illustration of the top two levels of an $(h,\phi)$-expander hierarchy. $E_h$ and $E_{h-1}$ are drawn in red and blue, respectively.}\label{fig:exp hie}

\end{figure}

See \Cref{fig:exp hie} for an example.
Below, we give an improved construction of the expander hierarchy.
\begin{theorem}\label{thm:edge exp hie} 
There exists an algorithm that, given a graph $G$, computes an $(h,\phi)$-expander hierarchy with $h\le\log n$ and $\phi=1/2$ in exponential time, or $\phi\ge\Omega(1/\sqrt{\log n})$ in polynomial time. 
\end{theorem}

\Patrascu and Thorup \cite{PatrascuT07} gave an exponential-time algorithm for $h\le\log m$ and $\phi=1/(2\log n)$ and a polynomial-time algorithm with $\phi\ge\Omega(1/\log^{1.5}n)$. \Cref{thm:edge exp hie} shaves one $\log n$ factor in $\phi$ for both settings.\footnote{The factor $1/2$ in \Cref{thm:edge exp hie} is quite artificial. It can be improved to $1$ if we slightly change the definition such that $X$ is $\phi$-expanding in $G$ if, for every cut $(S,V\setminus S)$, $|E_{G}(S,V\setminus S)|\ge\phi\min\{|E_{G}(S,V)\cap X|,|E_{G}(V\setminus S,V)\cap X|\}.$}

\Cref{thm:edge exp hie} is perhaps surprising. Recall a related and seemingly weaker concept of $\phi$-expander decomposition.\emph{ }A $\phi$-expander decomposition of a graph $G$ is an edge set $X\subseteq E$ such that, for each connected component $\Gamma$ of $G\setminus X$, $E\cap\Gamma$ is $\phi$-expanding in $G[\Gamma]$. It is known that there is no expander decomposition with $X\leq0.99|E|$ where $\phi=\omega(1/\log n)$ \cite{alev2018graph,moshkovitz2018decomposing}. Here, we give an expander \emph{hierarchy} with $\phi=1/2$. 

\Cref{thm:edge exp hie} follows immediately from the lemma below, inspired by~\cite[Lemma 3.1]{racke2014computing}.
\begin{lem}
There exists an exponential-time algorithm that, given a graph $G=(V,E)$, computes an edge set $X$ such that every connected component of $G\setminus X$ contains at most $n/2$ vertices and $X$ is $\frac{1}{2}$-expanding. In polynomial time, we can instead guarantee that $X$ is $\Omega(1/\sqrt{\log n})$-expanding.\label{lem:edge exp sep} 
\end{lem}

\begin{proof}
Initialize $X\gets E$. If $X$ is $\frac{1}{2}$-expanding, we are done. Otherwise, we repeatedly update $X$ as follows. Since $X$ is not $\frac{1}{2}$-expanding, there exists a vertex set $S$ where $|S|\le n/2$ such that $|E(S,V\setminus S)|<\frac{1}{2}\min\{\deg_{X}(S),\deg_{X}(V\setminus S)\}$. Update $X\gets X\cup E(S,V\setminus S)\setminus(X\cap E(S,S))$.

Let $X'$ denote $X$ after the update. Observe that every connected component of $G\setminus X'$ still contains at most $n/2$ vertices because $|S|\le n/2$. Moreover, since
\begin{align*}
\frac{1}{2}(2|X\cap E(S,S)|+|X\cap E(S,V\setminus S)|)=\frac{1}{2}\deg_{X}(S)>|E(S,V\setminus S)|,\intertext{we have}|X\cap E(S,S)|>|E(S,V\setminus S)|-\frac{1}{2}|X\cap E(S,V\setminus S)|\geq|E(S,V\setminus S)\setminus X|.
\end{align*}
Thus, $|X'|<|X|$ and there can be at most $|E|$ iterations of the procedure.

To get a polynomial time construction, we instead apply the sparsest cut algorithm of Arora et al.~\cite{AroraRV09} that, given $X$, either guarantees that $X$ is $\Omega(1/\sqrt{\log n})$-expanding or returns a set $S$ where $|S|\le n/2$ and $|E(S,V\setminus S)|<\frac{1}{2}\min\{\deg_{X}(S),\deg_{X}(V\setminus S)\}$. 
\end{proof}
\begin{proof}[Proof of \Cref{thm:edge exp hie}]Given $G$, compute the edge set $X$ from \Cref{lem:edge exp sep} and set $E_{h}\gets X$. To compute $E_{h-1},E_{h-2},\dots,E_{1}$, we recurse on each connected component $C$ of $G\setminus X$. We have $h\le\log n$ because each component $C$ has size $|C|\le n/2$. 
\end{proof}
\subsection{A Simple $\tilde{O}(f)$-Bit Labeling Scheme}\label{sec:simple edge label}

In this section, we prove \cref{thm:det-edge-labeling}, which uses $\tilde{O}(f)$-bit labels. It is the basis for our final $\tilde{O}(\sqrt{f})$-bit labeling scheme presented in \cref{thm:shorter-deterministic-labels}.
\begin{theorem}
\label{thm:det-edge-labeling} Fix any undirected graph $G=(V,E)$ and integer $f\geq1$. There are deterministic labeling functions $L_{V}:V\to\{0,1\}^{\log n}$ and $L_{E}:E\to\{0,1\}^{O(f\phi^{-1}\log^{2}n)}$ such that given any query $\ang{s,t,F}$, $F\subset E$, $|F|\leq f$, one can determine whether $s$ and $t$ are connected in $G-F$ by inspecting only $L_{V}(s),L_{V}(t),\{L_{E}(e)\mid e\in F\}$. The construction time is exponential for $\phi=1/2$ and polynomial for $\phi=\Omega(1/\sqrt{\log n})$. \end{theorem}

At a very high level, \Cref{thm:det-edge-labeling} is proved 
by adapting \Patrascu{} and Thorup's~\cite{PatrascuT07}
centralized edge-failure connectivity oracle to the distributed setting. Note that \Cref{thm:det-edge-labeling} already improves the 
state-of-the-art polynomial-time computable $\tilde{O}(f^2)$-bit labeling of~\cite{IzumiEWM23}.

Given a function $\level : E\to [h]$, let $\{E_1,\ldots,E_h\}$
be the corresponding edge partition, where $E_\ell = \level^{-1}(\ell)$.
Suppose $\{E_{1},\dots,E_{h}\}$ 
induces an $(h,\phi)$-expander hierarchy $\mathcal{H}$.
Define $T^{*}$ to be a minimum spanning tree with respect to $\level$,
and let $\Euler(T^{*})$ be its Euler tour, which is a list of length $n+2(n-1)$ that includes each vertex once (its first appearance) and each $T^{*}$-edge $\{u,v\}$ twice, as $(u,v)$ and $(v,u)$, according to a DFS traversal of $T^{*}$, starting from an arbitrary root vertex. 
Each vertex $u$ is identified by its 
position in $\Euler(T^*)$, denoted $\DFS(u)$.
We call each edge in $\Euler(T^{*})$ an \emph{oriented} edge. 
See \cref{fig:Euler} for a small example. 

For each $\ell\le h$, let $\mathcal{T}_{\leq\ell}$ be the set of 
\emph{level-$\ell$ trees} in the forest $T^{*}\cap E_{\le\ell}$. 
Observe that for each $T\in \mathcal{T}_{\leq \ell}$, 
$\Euler(T)$ is a subsequence of 
$\Euler(T^*)$, not necessarily contiguous. Furthermore, because $T^{*}$ is a minimum spanning tree with respect to $\level$, each connected component $\Gamma$ of $G_{\leq \ell}$ has a unique level-$\ell$ tree $T\in {\cal T}_{\leq \ell}$ such that $T$ spans $\Gamma$.

\begin{figure}[h]
\centering
\begin{tabular}{cc}
\parbox{1.5in}{\includegraphics[scale=0.5]{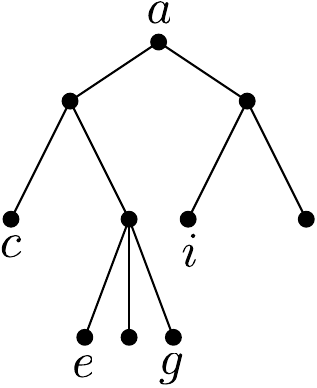}}
&
\parbox{4in}{$\Euler(T^*) = \begin{array}{l}
{\Big(}a,(a,b),b,(b,c),c,(c,b),(b,d),d,(d,e),e,(e,d),\\
\;\;\;\;(d,f),f,(f,d),(d,g),g,(g,d),(d,b),(b,a),\\
\;\;\;\;(a,h),h,(h,i),i,(i,h),(h,j),j,(j,h),(h,a){\Big).}\end{array}$}
\end{tabular}
\caption{\label{fig:Euler} Left: $T^*$ on vertex set $\{a,b,c,\ldots,j\}$, rooted at $a$.  
Right: $\Euler(T^*)$.}
\end{figure}

\begin{defn}
[Simple Deterministic Edge Labels] 
For $v\in V$, the vertex label $L_{V}(v)$ is just $\DFS(v)$.
Each \emph{non-tree} edge 
$e=\{u,v\}\not\in E(T^*)$ has a $2\log n$-bit label
$L_E(e)=(\DFS(u),\DFS(v))$.
Each \emph{tree edge} $e=\{u,v\}\in E(T^{*})$
is assigned an $O((f/\phi)\log^{2}n)$-bit label,
generated as follows.
\begin{enumerate}
\item \label{item:level simple} Store $(L_{V}(u),L_{V}(v),\level(e))$.

\item 
For each $\ell\in[\level(e),h]$, let $T_{\ell}\in\mathcal{T}_{\leq\ell}$ be the tree containing $e$. 
Write $\Euler(T_{\ell})$ as 
\[
\Euler(T_{\ell})=X\cdot(u,v)\cdot Y\cdot(v,u)\cdot Z.
\]
For each $W\in\{X,Y,Z\}$, 
\begin{enumerate}
\item \label{item:endpoint ET} 
Store the labels of the first and last elements 
of $W$, and store the labels of the first and last vertices in 
$W$ (i.e. $\min_{u\in W} \DFS(u)$ and $\max_{v\in W} \DFS(v)$).
\item \label{item:incident ET} Store the first $f/\phi+1$ level-$\ell$ 
non-tree edges incident to vertices in $W$.  
(Each such edge $\{u,v\}$ is encoded 
by $L_E(\{u,v\}) = (\DFS(u),\DFS(v))$.)
\end{enumerate}
\end{enumerate}
\label{def:simple edge label}
\end{defn}

\paragraph{The Query Algorithm.}

For each level $1\leq \ell\leq h$, let ${\cal T}_{\leq \ell, F}\subseteq {\cal T}_{\leq \ell}$ collect all level-$\ell$ trees $T$ such that $T$ intersects $F$. Recall that each connected component $\Gamma$ of $G_{\ell}$ has a unique $T\in {\cal T}_{\leq \ell}$ as its spanning tree. We define $G_{\leq \ell, F}$ to be a subgraph of 
$G_{\ell}$ that only collects the connected components of $G_{\leq \ell}$ whose spanning tree is in ${\cal T}_{\leq \ell,F}$.

Our goal is to sequentially build 
vertex partitions ${\cal P}_1,\ldots,{\cal P}_h$, 
where ${\cal P}_{\ell}$ is a partition of $V(G_{\leq \ell, F})$ that reflects the connected components of $G_{\leq \ell, F} - F$. 
In fact, we will compute, for each $T\in{\cal T}_{\leq\ell,F}$, a partition ${\cal P}_{\ell}[T]$ of $V(\Gamma)$ that reflects the connected component of $\Gamma-F$ (where $\Gamma$ is the connected component of $G_{\leq \ell}$ that has $T$ as its spanning tree), and then ${\cal P}_{\ell}$ is simply the union of ${\cal P}_{\ell}[T]$ over all $T\in{\cal T}_{\leq \ell, F}$. After ${\cal P}_h$ is computed, we can count the number of connected
components in $G-F$, or answer $s$-$t$ connectivity queries,
given $L_V(s),L_V(t)$.

Each ${\cal P}_{\ell}[T]$ has a compact representation as follows. We start with defining \emph{intervals}.
Consider a $T\in \mathcal{T}_{\leq\ell}$.
We can detect whether $T$ intersects $F$
and if so, enumerate 
$T\cap F=\{e_1,\ldots,e_{f_0}\}$
using \cref{item:endpoint ET} of the $F$ labels.
If we remove the oriented copies of $T\cap F$, 
$\Euler(T)$ breaks into a set of $2f_0+1$ intervals $\mathcal{J}(T,F)$. Note that from \Cref{item:endpoint ET}, for each interval in ${\cal J}(T,F)$, we can obtain the labels of its first and last elements and its first and last vertices. Towards the compact representation of ${\cal P}_{\ell}[T]$, we can think of each $J\in{\cal J}(T,F)$ as the vertex set $J\cap V(T)$. Each part $P\in {\cal P}_{\ell}[T]$ is represented by a set of intervals ${\cal J}_{P}\subseteq {\cal J}(T,F)$ such that $P = \bigcup_{J\in {\cal J}_{P}} J\cap V(T)$.

For the purposes of \emph{point location}, 
we will write $J\in \mathcal{J}(T,F)$ 
as $[\min_{u\in J\cap T} \DFS(u), \max_{v\in J\cap T} \DFS(v)]$.
Note that in general, an interval $[\DFS(u),\DFS(v)]$ in 
$\mathcal{J}(T,F)$ contains 
(the DFS numbers of) vertices outside of $T$.
Nonetheless, for vertices in $V(T)$, 
these intervals allow us to do correct
point location.

\begin{observation}[Point Location]\label{obs:pointlocation}
Suppose $x\in V(T)$, $T\in \mathcal{T}_{\leq \ell,F}$. 
If $[\DFS(u),\DFS(v)]$ is the (unique) 
interval in $\mathcal{J}(T,F)$
containing $\DFS(x)$ then there is a 
component in $T-F$ containing $u,x,$ and $v$.
\end{observation}

The query algorithm \emph{only} does point location on vertices $x$
known to be in $V(T)$, 
so \cref{obs:pointlocation} 
suffices for correctness.

\medskip 

Suppose $\mathcal{P}_{\ell-1}$ has been computed.  
For each $T\in \mathcal{T}_{\leq \ell}$ intersecting $F$, 
we enumerate 
$T\cap F=\{e_1,\ldots,e_{f_0}\}$
and
initialize $\mathcal{P}_\ell[T] \gets \mathcal{J}(T,F)$ (i.e. each part $P\in {\cal P}_{\ell}[T]$ is only one interval),
then proceed to unify parts of $\mathcal{P}_\ell[T]$
by applying rules {\bf R1}--{\bf R4} and the operation $\Unite_{T}(x,y)$. The $x,y$ in $\Unite_{T}(x,y)$ can be vertices in $V(T)$, intervals in ${\cal J}(T,F)$ or even parts in ${\cal P}_{\ell}[T]$, and $\Unite_{T}(x,y)$ will unite the parts containing $x$ and $y$ in ${\cal P}_{\ell}[T]$.

\begin{description}
    \item[R1.] If $J,J'$ are two intervals of $\Euler(T)$ 
    that share a common endpoint, say $J$ ends
    with `$u$' and $J'$ begins with `$(u,v)$', call $\Unite_{T}(J,J')$.

    \item[R2.] For each call to $\Unite_{T'}(x,y)$
    made in the construction of $\mathcal{P}_{\ell-1}[T']$, 
    $T'\subset T$, call $\Unite_{T}(x,y)$.

    \item[R3.] 
    For each non-tree edge $\{u,v\} \in E_\ell$ 
    encoded in the labels $\{L_E(e_i) \mid i\in [f_0]\}$,
    if $\{u,v\}\not\in F$, call $\Unite_{T}(u,v)$.
\end{description}

Rule {\bf R1} is implemented with \cref{item:endpoint ET} of the $F$-labels.
The enumeration of edges in {\bf R3}
uses \cref{item:incident ET} of the $F$-labels,
but to implement $\Unite_{T}(u,v)$ we need
to locate the intervals in $\mathcal{J}(T,F)$ 
containing $u,v$.
Since $\level(\{u,v\})=\ell$, 
both $u,v$ are in $V(T)$, 
and by \cref{obs:pointlocation} we can
locate the intervals containing $u,v$, 
given $\DFS(u),\DFS(v)$.

According to {\bf R2}, every $\Unite_{T'}(x,y)$ performed at level $\ell-1$ 
on some $T'\subset T$
is re-executed verbatim
as $\Unite_{T}(x,y)$ 
if $x,y$ are vertices.  
Since $x,y \in V(T') \subseteq V(T)$, \cref{obs:pointlocation}
lets us identify the intervals in $\mathcal{J}(T,F)$ containing
$x,y$. 
If $x,y$ were intervals from $\mathcal{J}(T',F)$
we can pick the first vertices from $x$ and $y$, say $x',y'$,
and call $\Unite_{T}(x',y')$. Once again, $x',y'\in V(T')\subseteq V(T)$, so \cref{obs:pointlocation} applies.

After executing {\bf R1}, 
$\mathcal{P}_\ell[T]$ reflects
the connected components of $T-F$. 
After executing {\bf R2}, 
$\mathcal{P}_\ell[T]$ reflects
the connected components of $(G_{\leq \ell-1}[V(T)] \cup T) - F$.
If it were the case that {\bf R3} had access to \emph{all}
level-$\ell$ non-tree edges, then it would be sufficient 
to find all connected components of 
$G_{\leq\ell}[V(T)]-F$.
However, \cref{item:incident ET} of the 
$F$-labels are only 
guaranteed to reveal up to $f/\phi+1$ 
level-$\ell$ non-tree edges per interval 
in $\mathcal{J}(T,F)$.

\begin{lem}
\label{lem:expansion-implies-connected} If $P_{1},\ldots,P_{k}\in\mathcal{P}_{\ell}[T]$ are those parts with $\deg_{\ell}(P_{i})>f/\phi$, then $\bigcup_{i=1}^{k}P_{i}$ are contained in a single connected component of $G_{\le\ell}-F$. 
\end{lem}

\begin{proof}
Suppose the claim is false, that there is some partition of $\{P_{1},\ldots,P_{k}\}$ into $A$ and $B$ which are disconnected. Then 
\[
\frac{|E_{\leq\ell}(A,B)|}{\min\{\deg_{\ell}(A),\deg_{\ell}(B)\}}<\frac{f}{f/\phi}=\phi,
\]
contradicting the fact that $\mathcal{H}$ is an $(h,\phi)$-expander hierarchy. 
\end{proof}

In light of \cref{lem:expansion-implies-connected}, 
we continue to unify parts according to a fourth rule.
\begin{description}
\item[R4.] If $P,P' \in\mathcal{P}_{\ell}[T]$ 
have $\deg_{\ell}(P),\deg_{\ell}(P')>f/\phi$,
call $\Unite_{T}(P,P')$. 
\end{description}

The \emph{full} partition $\mathcal{P}_{\ell}$ is obtained by taking the union of all $\mathcal{P}_{\ell}[T]$, for $T\in\mathcal{T}_{\leq\ell}$ intersecting $F$, plus the trivial partitions 
$\mathcal{P}_\ell[T] = \{V(T)\}$
for every $T\in\mathcal{T}_{\leq \ell}$ 
disjoint from $F$.

\begin{lem}[Correctness] 
If $P\in\mathcal{P}_{\ell}$, 
then $P$ is a connected component in 
$G_{\leq \ell}-F$. 
\end{lem}

\begin{proof}
Rules {\bf R1}--{\bf R3} are clearly sound
and \cref{lem:expansion-implies-connected} 
implies
{\bf R4} is sound.  We consider completeness.
If there is a path between 
$u,v\in V(T)$ in $G_{\leq \ell-1}-F$
then by induction on $\ell$, 
$u$ and $v$ will be in the same part of 
$\mathcal{P}_{\ell}$ after executing 
{\bf R1} and {\bf R2}.
Suppose $u,v$ are joined by a path in 
$G_{\leq \ell}-F$,
but $u,v$ are in different parts $P_u,P_v$ connected 
by a level-$\ell$ edge $e'=\{u',v'\}$. 
Because {\bf R3} could not be
applied, $e'$ is not contained 
in \cref{item:incident ET} of the $F$-edges
bounding the intervals in $\mathcal{J}(T,F)$
containing $u',v'$, 
implying $\deg_\ell(P_u),\deg_\ell(P_v) > f/\phi$,
but then by {\bf R4}, 
$P_u,P_v$ would have been united.
\end{proof}

Once we have constructed $\mathcal{P}_{h}$, a connectivity query $\ang{s,t,F}$ works as follows. First, identify the two intervals in $\Euler(T^{*})$ of the forest $T^{*}-F$ that contains $s$ and $t$. 
This can be done using the 
predecessor search over the endpoints
of the intervals in 
$O(\min\{\frac{\log\log n}{\log\log\log n},\,\frac{\log f}{\log\log n}\})$ time~\cite{PatrascuT06,PatrascuT14}, 
since there are $O(f)$ intervals, 
each represented by $O(\log n)$-bit numbers. 
Then we check whether the two intervals are in the 
same part of $\mathcal{P}_{h}$, corresponding to the 
same connected component of $G-F$.

%% file: 2.5_shorter_new.tex
\subsection{Second Tool: Code Shares}
\label{sec:code share}

\cref{thm:code-shares} gives 
the distributed coding scheme.  
We will only invoke it with $d=2$.

\begin{theorem}[Reed-Solomon Code Shares]\label{thm:code-shares}
Let $\mathbf{m}\in \mathbb{F}_q^k$ be a message, with $q>k$.  For any integer parameter $d\geq 2$, there are
$O(d\log q)$-bit \emph{code shares} 
$C_1,\ldots,C_{k}$ so that for any index 
set $J\subset [k]$ with 
$|J|\geq k/d$,
we can reconstruct $\mathbf{m}$ from the code shares $\{C_j \mid j\in J\}$ in polynomial time.
\end{theorem}

\begin{proof}
One can regard $\mathbf{m}$ as the coefficients of a degree-$(k-1)$
polynomial $g_1$ over $\mathbb{F}_q$, 
or even as a degree-$(\ceil{k/d}-1)$ 
polynomial $g_d$ over $\mathbb{F}_{q^d}$.
The code shares $(C_i)_{1\leq i\leq k}$
are defined to be distinct evaluations of $g_d$.
\[
C_i = (i,g_d(i)).
\]
Given the code shares $\{C_i \mid i\in J\}$ for $|J|\geq k/d$, we can reconstruct
$g_d$ and hence $\mathbf{m}$ in 
polynomial time via polynomial 
interpolation.
\end{proof}

\subsection{An $\tilde{O}(\sqrt{f})$-Bit Labeling Scheme}
\label{sec:shorter edge label}

\paragraph{High-level Idea.}
The labeling scheme of \cref{sec:simple edge label} is \emph{non-constructive} inasmuch as rule \textbf{R4} infers that two $P,P'$ are connected, not by finding a path between them, but by checking if their volumes $\deg_{\ell}(P),\deg_{\ell}(P')>f/\phi$. In this section, we give a labeling scheme that is \emph{even more} non-constructive. We can sometimes infer that a part $P$ with $\deg_{\ell}(P)<f/\phi$ is nonetheless in a connected component $C$ of $G_{\le\ell}-F$ with $\deg_{\ell}(C)>f/\phi$ without explicitly knowing an edge incident to $P$.

A key idea in the construction is to store a large volume of information about an interval of an Euler tour as \emph{code shares} distributed across labels of ``nearby'' edges. Given a sufficient number of code shares, we will be able to reconstruct the information about the interval.

\paragraph{Notations.}

We shall assume without loss of generality that the graph has degree 3. Given any $G'=(V',E')$ with irregular degrees, we form $G=(V,E)$ by substituting for each $v'\in V$ a $\deg(v')$-cycle, then attach each edge incident to $v'$ to a distinct vertex in the cycle, hence $|V|=|E'|/2$. Given labelings $L_{V}:V\to\{0,1\}^{*},L_{E}:E\to\{0,1\}^{*}$, we let $L_{E'}(e')=L_{E}(\phi(e')),L_{V'}(v')=L_{V}(\phi(v'))$, where $\phi(e')\in E$ is the edge corresponding to $e'$ and $\phi(v')\in V$ is any vertex in the $\deg(v')$-cycle of $v'$. Correctness is immediate, since $F'\subset E'$ disconnects $s,t\in V'$ iff $\phi(F')\subset E$ disconnects $\phi(s),\phi(t)\in V$.

Recall $T^{*}$ is the minimum spanning tree of $G$ with respect to the $\level$ function. Fix a level $\ell$. Let $T\in\mathcal{T}_{\leq\ell}$ be a tree spanning a component of $G_{\le\ell}$. 
For each $v\in V(T)$, let 
\[
\wt_{\ell}(v)=\left\{ \begin{array}{ll}
1 & \mbox{ if \ensuremath{v} is incident to a level-\ensuremath{\ell} non-tree edge,}\\
0 & \mbox{ otherwise},
\end{array}\right.
\]
and for each oriented tree edge $(u,v)$, let $\wt_{\ell}(u,v)=0$. If $S$ is an interval of vertices and oriented edges in $\Euler(T)$, $\wt_{\ell}(S)=\sum_{x\in S}\wt_{\ell}(x)$.

Let $[\alpha,\beta]$ denote the interval of $\Euler(T)$ starting at $\alpha$ and ending at $\beta$. Then $\dist_{\ell}(\alpha,\beta)=\dist_{\ell}(\beta,\alpha)\bydef\sum_{\gamma\in[\alpha,\beta]\setminus\{\alpha,\beta\}}\wt_{\ell}(\gamma)$.\footnote{We exclude the endpoints of the intervals from the sum just to avoid double counting at the endpoints when we sum distances of two adjacent intervals. } For any vertex/edge element $\alpha$ in $\Euler(T)$, the set of all vertices within distance $r$ is: 
\[
\Ball_{\ell}(\alpha,r)=\{v\in V(T)\mid\dist_{\ell}(\alpha,v)\leq r\}.
\]
We overload the $\Ball$-notation for edges 
$e=\{u,v\}$. Here $(u,v),(v,u)$ refer to the oriented occurrences of $e$ in an Euler tour, if $e$ is a tree edge.
\[
\Ball_{\ell}(e,r)=\left\{ \begin{array}{ll}
\Ball_{\ell}((u,v),r)\cup\Ball_{\ell}((v,u),r) & \mbox{ when \ensuremath{e\in E(T)} is a tree edge,}\\
\Ball_{\ell}(u,r)\cup\Ball_{\ell}(v,r) & \mbox{ when \ensuremath{e\not\in E(T)} is a non-tree edge.}
\end{array}\right.
\]

Henceforth, the only balls we consider have radius $r$, 
\[
r\bydef \ceil{\sqrt{f/\phi}}.
\]

Assume, without loss of generality, that $\wt(\Euler(T))$ is a power of 2, by padding the end with dummy weight-1 elements if necessary. For every $j\leq j_{\max}\bydef\ceil{\log(f/\phi)}$, define $\mathcal{I}_{j}$ to be a partition of $\Euler(T)$ into consecutive intervals with weight $2^{j}$, each the union of two intervals from $\mathcal{I}_{j-1}$.

\paragraph{The Scheme.}
The key idea of this labeling scheme is to focus on \emph{large gap edges}. See \cref{fig:LGE}.

\begin{figure}[h]
    \centering\includegraphics[scale=.4]{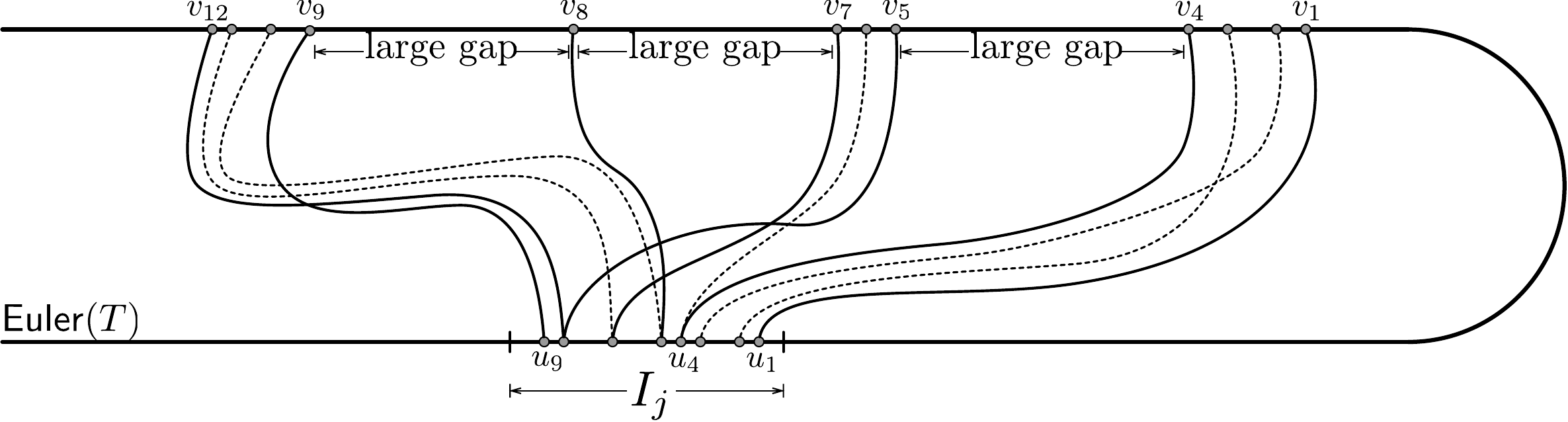}
    \caption{\label{fig:LGE}
    An interval 
    $I_j\in \mathcal{I}_j$ 
    with $\wt(I_j)=2^j$ ($j=3$).
    There are 12 level-$\ell$ edges in $E(I_j,\overline{I_j})$, 
    $\{u_1,v_1\},\ldots,\{u_{12},v_{12}\}$, 
    ordered by their non-$I_j$ endpoint.  
    The large gap edges of $I_j$ are $\{u_1,v_1\},\{u_4,v_4\},\{u_5,v_5\},\{u_7,v_7\},\{u_8,v_8\},\{u_9,v_9\},\{u_{12},v_{12}\}$.}
\end{figure}

\begin{defn}
[Large Gap Edges]\label{def:large-gap-edges} Fix an interval $I_{j}\in\mathcal{I}_{j}$ and let $E_{\ell}(I_{j}):=E_{\ell}(I_{j},\overline{I_{j}})$ be all level-$\ell$ edges with exactly one endpoint in $I_{j}$. We write $E_{\ell}(I_{j})=\{\{u_{1},v_{1}\},\{u_{2},v_{2}\},\dots\}$ such that, for all $i$, $u_{i}\in I_{j}$ and $v_{i}\in\overline{I_{j}}$. Order the edges according to $v_i$, so
\begin{equation}
\DFS(v_{1})\leq \DFS(v_{2}) \leq \cdots \leq \DFS(v_{|E_{\ell}(I_{j})|}).\label{eq:lge}
\end{equation}
If, for $q\in[1,|E_{\ell}(I_{j})|)$, $v_{q+1}\not\in\Ball_{\ell}(v_{q},r)$, then $\{u_{q},v_{q}\},\{u_{q+1},v_{q+1}\}$ are called \emphbf{large gap edges} w.r.t.~$\ell$ and $I_{j}$. The first and last edges $\{u_{1},v_{1}\},\{u_{|E_{\ell}(I_{j})|},v_{|E_{\ell}(I_{j})|}\}$ are always large gap edges. Define $\LGE_{\ell}(I_{j})\subseteq E_{\ell}(I_{j})$ to be the set of large gap edges and $\lge_{\ell}(I_{j})=|\LGE_{\ell}(I_{j})|$ to be their number. 
\end{defn}

Regard $\LGE_{\ell}(I_{j})$ as a \emph{message} $\mathbf{m}_{\ell}(I_{j})\in\F_{q}^{\lge_{\ell}(I_{j})}$, where $q=\poly(n)$ is large enough to encode a single edge. We apply \cref{thm:code-shares} with $d=2$
to break $\mathbf{m}_{\ell}(I_{j})$ into code shares so that given any set of $\lge_{\ell}(I_{j})/2$ shares we can reconstruct $\mathbf{m}_{\ell}(I_{j})$. If $\{u,v\}\in\LGE_{\ell}(I_{j})$ is a large gap edge with $u\in I_{j}$ and $v\in\overline{I_{j}}$, let $C_{\ell,j}(v,u)$ 
be the code share of 
$\{u,v\}$ w.r.t.~$\ell$ and $I_{j}$. Set $C_{\ell,j}(v,u)\bydef\bot$ if $\{u,v\}\notin\LGE_{\ell}(I_{j})$ is not a large gap edge. (Note that $C_{\ell,j}(u,v)$ would be a different code share w.r.t.~some interval $I_{j}'\ni v$.) To simplify notation, for each edge $e=\{u,v\}$, we define the \emph{level-$\ell$ code share of $e$} as a bundle $C_{\ell}(e)=\{C_{\ell,j}(u,v),C_{\ell,j}(v,u)\mid j\le j_{\max}\}$. Note that $C_{\ell}(e)$ also indicates whether $e$ is a large gap edge w.r.t.~all intervals $I_{j}\ni u$ and $I'_{j}\ni v$, 
for all $j\leq j_{\max}$.
\begin{defn}
[Shorter Deterministic Edge Labels]\label{def:shorter-deterministic-labels}
An $O\left(\sqrt{f/\phi}\log(f/\phi)\log^{2}n\right)$-bit label $L_{E}(e)$ for each edge $e=\{u,v\}$  is constructed as follows. 
\begin{enumerate}
\item \label{item:level shorter}Store $(L_{V}(u),L_{V}(v),\level(e))$.
\item \label{item:all edge shorter}For each $\ell\in[\level(e),h]$ and all level-$\ell$ edges $e'=\{u',v'\}$ incident to $\Ball_{\ell}(e,r)$ vertices,  
\begin{enumerate}
\item \label{item:position}
store $(\DFS(u'),\DFS(v'))$,
\item \label{item:codeshare}
store the level-$\ell$ code share bundle $C_{\ell}(e')$.
\end{enumerate}
\item \label{item:tree edge shorter}If $e\in T^{*}$, then for each $\ell\in[\level(e),h]$, let $T_{\ell}\in\mathcal{T}_{\leq\ell}$ be the tree containing $e$. Write $\Euler(T_{\ell})$ as 
\[
\Euler(T_{\ell})=X\cdot(u,v)\cdot Y\cdot(v,u)\cdot Z.
\]
\begin{enumerate}
\item \label{item:interval endpoint}
For each $W\in\{X,Y,Z\}$, %
store the labels of the first and last elements 
of $W$, and store the labels of the first and last vertices in 
$W$ (i.e. $\min_{u\in W} \DFS(u)$ and $\max_{v\in W} \DFS(v)$).
\item \label{item:store lge}For each $j\leq j_{\max}$, let $I_{j}^{(1)},I_{j}^{(2)},I_{j}^{(3)},I_{j}^{(4)}\in\mathcal{I}_{j}$ be the closest intervals on either side of $(u,v),(v,u)$ that do not contain $(u,v)$ or $(v,u)$. For each $k\in\{1,\ldots,4\}$, store $\lge_{\ell}(I_{j}^{(k)})$, and if $\lge_{\ell}(I_{j}^{(k)})\leq4r$, store $\LGE_{\ell}(I_{j}^{(k)})$.
\end{enumerate}
\end{enumerate}
The bit-length of the edge labeling 
is justified as follows.
\Cref{item:level shorter} has length $O(\log n)$. 
Since $C_{\ell}(e)$ uses  $O(j_{\max} \log n)$ bits by \Cref{thm:code-shares}, we have that
\Cref{item:all edge shorter} has length 
$O(hrj_{\max}\log n)$ $=O(\sqrt{f/\phi}\log(f/\phi)\log^{2}n)$. \Cref{item:tree edge shorter} has length $h\cdot(O(\log n)+rj_{\max}O(\log n))=O(\sqrt{f/\phi}\log(f/\phi)\log^{2}n)$ bits.
\end{defn}

\paragraph{The Query Algorithm.}
We initialize $\mathcal{P}_\ell[T]$, $T\in\mathcal{T}_{\leq \ell,F}$, 
exactly 
as in the proof of \cref{thm:det-edge-labeling},
and proceed to apply rules {\bf R1}--{\bf R3} as-is.
Note that we can implement {\bf R1} using 
\cref{item:interval endpoint}
and {\bf R3} using \cref{item:position}.
{\bf R2} is simply re-executing calls to 
$\Unite_{T'}$ from those $T'\in {\cal T}_{\leq \ell-1,F}$ such that $T'\subset T$.
We replace rule \textbf{R4} 
with the similar rule \textbf{R4'}.

\begin{itemize}
\item[{\bf R4'.}] Suppose that $\mathcal{Q}\subseteq\mathcal{P}_{\ell}[T]$ is such that for each $P\in\mathcal{Q}$, $P$ must be in a connected component $C$ of $G_{\ell}-F$ with $\deg_{\ell}(C)>f/\phi$. 
Then, unite all parts of $\mathcal{Q}$
with calls to $\Unite_{T}$.
\end{itemize}
To prove  correctness, we recall the definition of ${\cal J}=\mathcal{J}(T,F)$. Let $F\cap E(T)=\{e_{1},\ldots,e_{f_{0}}\}$ be the set of deleted tree edges of $T$, which we can enumerate using \Cref{item:interval endpoint} of the edge labels.
These edges break $\Euler(T)$ into a set of $2f_{0}+1$ intervals denoted by $\mathcal{J}(T,F)$. Each part $P\in\mathcal{P}_{\ell}[T]$ consists of a collection of intervals from $\mathcal{J}(T,F)$. Our goal is to prove the following.
\begin{lem}
\label{lem:correct shorter}For each interval $J\in\mathcal{J}$, we can either (i) list all other intervals $J'\in\mathcal{J}$ adjacent to $J$, or (ii) infer that $J$ is in a connected component $C$ of $G_{\le\ell}-F$ with $\deg_{\ell}(C)>f/\phi$.
\end{lem}

The above lemma implies correctness of the algorithm as we can keep applying \textbf{R3} and \textbf{R4'} to obtain the correct ${\cal P}_{\ell}[T]$ at the end.

Consider an interval $J\in \mathcal{J}$. Let the $F$-edges bounding $J$ be $\alpha_{0},\alpha_{1}$. Observe that we can partition $J$ into less than $2(j_{\max}+1)$ intervals from $\mathcal{I}_{0}\cup\cdots\cup\mathcal{I}_{j_{\max}}$. We consider each of these intervals $I_{j}\in{\cal I}_{j}$ individually. Below, we say that $F$ \emph{reveals} $e'$ if $e'$ is incident to a vertex in $\Ball_{\ell}(e,r)$, for some failed edge $e\in F$. Whenever $F$ reveals $e'$, 
\Cref{item:all edge shorter} of the edge labels
gives us the position of its endpoints
and the code share bundle for $e'$.
The following lemma is crucial.

\begin{figure}[h]
\centering\includegraphics[scale=.4]{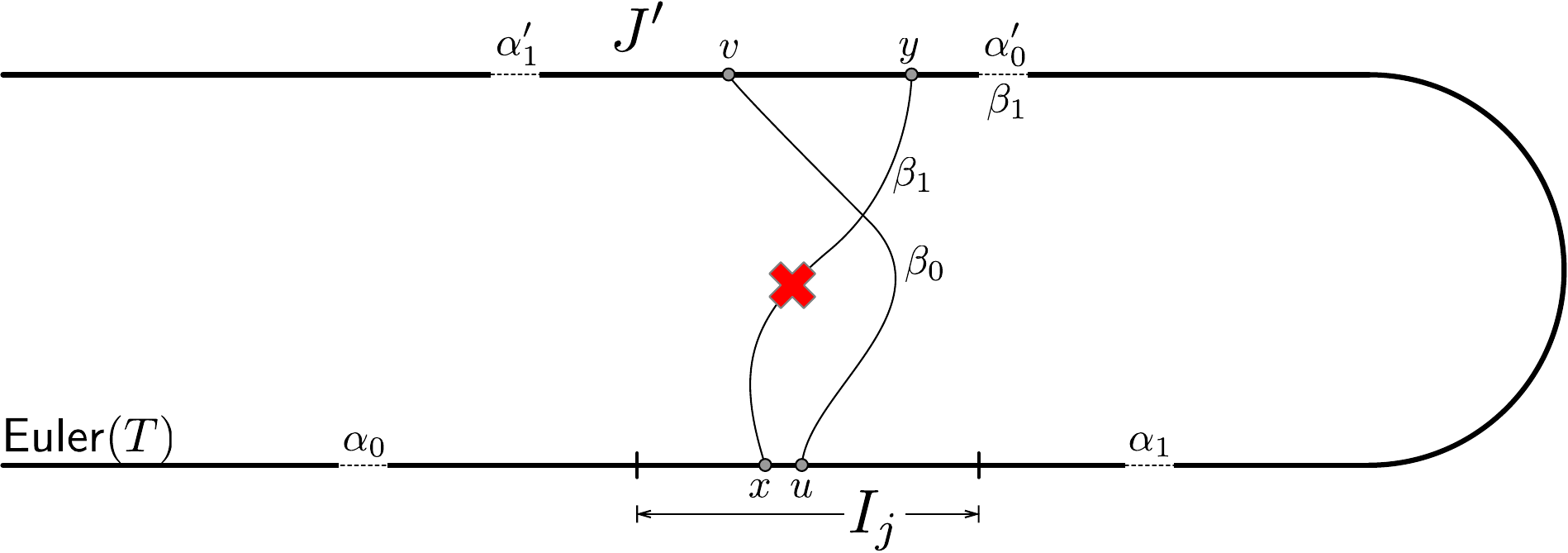}
\caption{\label{fig:LGE-known}Illustration of \Cref{lem:find edge}. An interval $I_{j}\subseteq J$ is incident to $J'$. $J,J'\in\mathcal{J}$ are bounded by $\alpha_{0},\alpha_{1}\in F$ and $\alpha_{0}',\alpha_{1}'\in F$, respectively. Either $\beta_{0}$ is a large gap edge, and stored in either $L_{E}(\alpha_{0})$ or $L_{E}(\alpha_{1})$, or it is stored in $L_{E}(\beta_{1})$, where $\beta_{1}=\{x,y\}\in F$ (if it exists), or $\beta_{1}=\alpha_{0}'$.}
\end{figure}

\begin{lem}
\label{lem:find edge}Consider an interval $I_{j}\subseteq J\in{\cal J}$ where $I_{j}\in{\cal I}_{j}$. If we have access to the set $\LGE_{\ell}(I_{j})$ of large gap edges, then we can check if another interval $J'\in\mathcal{J}$ is adjacent to $I_{j}$.
\end{lem}

\begin{proof}
Suppose that $J'\in\mathcal{J}$ is adjacent to $I_{j}$. Let $\beta_{0}=\{u,v\}\in E-F$ be the \emph{first} level-$\ell$ non-deleted edge with $u\in I_{j},v\in J'$, when ordered by $\DFS$ number. 
We claim that either $\beta_{0}\in\LGE_{\ell}(I_{j})$ is a large gap edge or $\beta_{0}$ is revealed by $F$. In either case, we learn the endpoints of $\beta_{0}$, which certifies 
that $J'$ is adjacent to $I_{j}$. 

Let $J'$ be bounded by $F$-edges $\alpha_{0}',\alpha_{1}'$. Consider a level-$\ell$ edge $\{x,y\}$ that is the predecessor of $\beta_{0}$ according to \Cref{eq:lge}. 
In particular, $x\in I_{j}$ and $y\notin I_{j}$ 
has the largest $\DFS(y)$ such that 
$\DFS(y)\le\DFS(v)$. 
Suppose $\beta_{0}\notin\LGE_{\ell}(I_{j})$. Then, $v\in\Ball_{\ell}(y,r)$ by definition. Now, we claim there is a $\beta_{1}\in F$ where $v\in\Ball_{\ell}(\beta_{1},r)$. Hence, $\beta_{0}$ is revealed by $F$ which would complete the proof. There are just two cases. 
\begin{itemize}
    \item If $y\in J'$, then $\{x,y\}\in F$ since $\beta_{0}$ is the first non-deleted edge. 
    We choose $\beta_{1}=\{x,y\}\in F$ 
    and therefore $v\in\Ball_{\ell}(\beta_{1},r)$. 
    \item If $y\notin J'$, then $\dist_{\ell}(v,\alpha'_{0})\le\dist_{\ell}(v,y)\leq r$ and therefore $v\in\Ball_{\ell}(\alpha'_{0},r)$. 
    We set $\beta_{1}=\alpha'_{0}\in F$ so $v\in\Ball_{\ell}(\beta_{1},r)$.\qedhere
\end{itemize}
\end{proof}

We are now ready to prove \Cref{lem:correct shorter}. 

\begin{proof}[Proof of \Cref{lem:correct shorter}]
There are three cases.
\paragraph{Case 1:} An $I_{j}\subseteq J$ has $\lge_{\ell}(I_{j})\leq4r$. By \Cref{item:store lge} of the label $L_{E}(\alpha_{0})$ or $L_{E}(\alpha_{1})$, we can access the whole set $\LGE_{\ell}(I_{j})$. So, by \Cref{lem:find edge}, we can list all intervals $J'\in\mathcal{J}$ adjacent to $I_{j}$.
\paragraph{Case 2:} An $I_{j}\subseteq J$ has $\lge_{\ell}(I_{j})>4r$, and $F$ reveals at least $(\lge_{\ell}(I_{j}))/2$ large gap edges in $\LGE_{\ell}(I_{j})$. Given $(\lge_{\ell}(I_{j}))/2$ code shares, by \cref{thm:code-shares}, we can also reconstruct $\LGE_{\ell}(I_{j})$. So, we can again list all intervals $J'\in\mathcal{J}$ adjacent to $I_{j}$, by \Cref{lem:find edge}. 
\paragraph{Case 3:} An $I_{j}\subseteq J$ has $\lge_{\ell}(I_{j})>4r$, and $F$ reveals at most $(\lge_{\ell}(I_{j}))/2$ large gap edges in $\LGE_{\ell}(I_{j})$. In this case, we claim that the connected component $C$ in $G_{\le\ell}-F$ containing $I_{j}$ has $\deg_{\ell}(C)>f/\phi$. Observe that for each \emph{unrevealed} large gap edge $\{u,v\}\in\LGE_{\ell}(I_{j})$ with $u\in I_{j}$, it must be that $\Ball_{\ell}(v,r)\subseteq C$ because there is no failed edge $e\in F$ within distance $r$ from $v$. Each $\Ball_{\ell}(v,r)$ has weight at least $2r$, 
and the sum of their weights can be at most four times the weight of their union. So $\deg_{\ell}(C)>(2r/4)\cdot(\lge_{\ell}(I_{j})/2)>(2r/4)\cdot(4r/2)=r^{2}\ge f/\phi$ as desired.
\end{proof}

%% file: vertex_failure_conn_labeling_scheme.tex
\section{Deterministic Vertex Fault Connectivity Labels}

This section is dedicated to proving \cref{thm:vertex label} concerning labels for \emph{vertex} faults.

\begin{theorem}[Improved Deterministic Vertex Labels]\label{thm:vertex label}
Fix any undirected graph $G=(V,E)$ with $n$ vertices and integer $f\geq 1$.
    There are deterministic labeling functions 
    $L_V : V \to \{0,1\}^{O(f^{4}\phi^{-1}\log^{7} n)}$ 
    such that given any query 
    $\ang{s,t,F}$, $F\subset V$, $|F|\leq f$,
    one can determine
    whether $s$ and $t$ are connected in 
    $G-F$ by inspecting only 
    $L_V(s),L_V(t),\{L_V(v) \mid v\in F\}$. The construction time is exponential when $\phi = 1$ and polynomial when $\phi = \Omega(1/\sqrt{\log n})$. The query time is $\poly(f,\log n)$.
\end{theorem}

To prove \Cref{thm:vertex label}, we first describe the underlying hierarchical structure of the algorithm in \Cref{sec:structure vertex}. This structure allows us to prove a divide-and-conquer lemma (\Cref{coro:MainCorrectness}) that is crucial for answering connectivity queries under vertex failures in a bottom-up manner on the hierarchy. Based on the divide-and-conquer lemma, in \Cref{sec:strategy vertex} we then describe how to answer connectivity queries assuming access to primitives that return information about the hierarchy. Finally, we describe how to implement these primitives 
in a distributed manner by providing 
the labeling scheme in 
\Cref{sec:implementing vertex label}.

\subsection{Overview and Challenges} 

We briefly discuss our approach at a very high level.  We would like to highlight the specific challenges that arise when tolerating vertex 
faults relative to edge faults.

\paragraph{The First Challenge.} We start with the ideal scenario that the input graph is already a vertex expander. The first challenge is how to obtain a \emph{stable connectivity certificate}. Recall that in the edge fault scenario (also assuming the graph is an edge expander), we can simply take a spanning tree as a stable connectivity certificate. 
The removal of any $f$ edges will only break the spanning tree into $f+1$ subtrees.
However, in general there is no upper bound 
on the number of subtrees when removing $f$ vertices.

This is a natural barrier to handling vertex faults, 
and several previous works, e.g. \cite{DuanP20,LongS22,ParterPP24}, 
follow the same idea to overcome it. 
Take a \emph{low-degree} spanning tree as a stable connectivity certificate. A vertex expander indeed admits a low-degree spanning tree; this is 
formally stated in \Cref{lemma:LowDegreeSteinerTree}. 

Using this idea, one can easily generalize our edge fault connectivity labeling scheme to obtain an $\wtilde{O}(f)$-size vertex fault connectivity labeling scheme for vertex expander input graphs. Roughly speaking, let $F$ be the vertex faults.
From the low-degree spanning tree $T$, we first obtain an initial partition ${\cal P}$ 
of $G\setminus F$ consisting of the connected components of $T\setminus F$.
Then we exploit the nature of a $\phi$-vertex expander $G$: all sets $A\in {\cal P}$ s.t. $|A\cup N_{G}(A)|>f/\phi$---call them \emph{giant sets}---must be inside the same connected 
component of $G\setminus F$. 
On the other hand, for each non-giant set $A\in{\cal P}$, its neighbor set $N_{G}(A)$ is of size at most $f/\phi$ and we can obtain $N_{G}(A)$ explicitly by designing labels on the Euler tour of the low-degree spanning tree. Therefore, we first merge all the giant sets of ${\cal P}$ together, and then merge further using $N_{G}(A)$ for each non-giant $A\in{\cal P}$.

\paragraph{The Second Challenge.} 
The second challenge arises when the input graph is \emph{not} a vertex expander. 
In fact, an input graph $G$ admitting a two-level vertex expander hierarchy already captures this challenge. 
A graph $G$ admits a two-level $\phi$-vertex expander hierarchy if there is a separator $X\subseteq V(G)$ s.t. $X$ is $\phi$-expanding in $G$, and each connected component $C$ of $G\setminus X$ is a $\phi$-vertex expander.

Let $F$ be the vertex faults. We assume that we are given an initial partition ${\cal P}$ of $V(G)\setminus F$ that captures 
(1) the connectivity of $C\setminus F$ for each connected component $C$ of $G\setminus X$, 
and 
(2) the connectivity of $X$ in $T\setminus F$ certified by the stable connectivity certificate $T$. For example, we can think of
\[
{\cal P} = {\cal P}_{X}\cup \bigcup_{\text{components $C$ of $G\setminus X$}}\{\text{components of $C\setminus F$}\},
\]
where ${\cal P}_{X}$ is some partition of $X\setminus F$ which is with respect to the connectivity of $X$ in $G\setminus F$, but may not fully capture the connectivity of $X$ in $G\setminus F$.

Clearly, we are done if we can further merge sets in ${\cal P}$ using 
edges incident to $X$, call them \emph{$X$-edges}. 
Let us try the same algorithm and see how it breaks. 
The fact that $X$ is $\phi$-expanding in $G$ tells us that all sets $A\in{\cal P}$ s.t. $|(A\cup N_{G}(A))\cap X| > f/\phi$ (giant sets) must belong to the same connected component of $G\setminus F$. Also, for each non-giant set $A\in{\cal P}$, $N_{G}(A)\cap X$ is of size at most $f/\phi$ and we can obtain it explicitly with $\tilde{O}(f)$-bit labels. 
Again, we first merge giant sets into one \emph{giant group}, 
and then for each $A\in{\cal P}$, merge it with each group intersecting $N_{G}(A)\cap X$. Let us try to confirm the correctness of this merging procedure. We use $\{x,v\}$ to denote an $X$-edge with $x\in X$.
\begin{enumerate}
\item Any two giant sets $A_{1},A_{2}\in{\cal P}$ are indeed merged.
\item Suppose an $X$-edge $\{x,v\}$ joins $A_{1},A_{2}\in{\cal P}$
with $x\in A_1, v\in A_2$, where $A_2$ is non-giant.
Then we will merge $A_{1}$ and $A_{2}$ because $A_{1}$ intersects $N_{G}(A_{2})\cap X$.
\item Suppose an $X$-edge $\{x,v\}$ joins $A_{1},A_{2}\in{\cal P}$
with $x\in A_1, v\in A_2$, but now $A_1$ is non-giant and $A_2$ is giant.
Although $A_{1}$ intersects $N_{G}(A_{2})\cap X$, we are \underline{\emph{not}} guaranteed that $A_{1}$ and $A_{2}$ are merged because $A_{2}$ is \underline{\emph{giant}}. In other words, $X$-edges in case 3 will not be detected by this method.
\end{enumerate}

This \emph{asymmetry} is the major difference between the vertex-fault case and the edge-fault case. To overcome it, the key observation is that if there exists a case-3 $X$-edge for a non-giant set $A_{1}$,
i.e. there exists an $X$-edge $\{x,v\}$ such that $x\in A_{1}$ and $v$ is in some giant set, 
then $A_{1}$ must be merged with the giant group. In other words, for a non-giant 
$A_{1}$, instead of knowing all case-3 $X$-edges incident to $A_1$, 
it suffices to check if any such case-3 $X$-edges \emph{exist}. 

Therefore, we will \emph{count} the number of case-3 $X$-edges for each non-giant set $A_{1}$. Roughly speaking, this is possible because this number is exactly $\delta_{\text{all}} - \delta_{\text{non-giant}} - \delta_{F}$, where
\begin{align*}
\delta_{\text{all}} &= \text{number of $X$-edges $\{x,v\}$ s.t. $x\in A_{1}$},\\
\delta_{\text{non-giant}} &= \text{number of $X$-edges $\{x,v\}$ s.t. $x\in A_{1}$ and $v$ is in some non-giant $A_{2}$},\\
\delta_{F} &= \text{number of $X$-edges $\{x,v\}$ s.t. $x\in A_{1}$ and $v\in F$}.
\end{align*}
We will not elaborate now 
on how to count these numbers. 

\subsection{The Structure}
\label{sec:structure vertex}
\subsubsection{The Basis: A Vertex Expander Hierarchy}

In this section, we construct an 
expander hierarchy for vertex expansion 
similar to the edge expansion version 
from \Cref{sec:edge exp hie}. 
This will be the basis of our structure.

For any graph $G=(V,E)$, a vertex cut $(L,S,R)$ is a partition of $V$ such that $L,R\neq\emptyset$ and there is no edge between $L$ and $R$. For any vertex set $X\subseteq V$, we say that $X$ is \emph{$\phi$-vertex-expanding} in $G$ if for every vertex cut $(L,S,R)$ in $G$, 
\[
|S|\ge\phi\min\{|X\cap(L\cup S)|,\, |X\cap(R\cup S)|\}.
\]
Consider a partition $\{V_{1},\dots,V_{h}\}$ of $V$. We denote $V_{\le\ell}:=\bigcup_{i\le\ell}V_{i}$ and $V_{>\ell}:=\bigcup_{i>\ell}V_{i}$. 
Let $G_{\leq \ell}$ be the graph 
induced by $V_{\leq \ell}$.
\begin{defn}
[Vertex Expander Hierarchy]Given a graph $G=(V,E)$, a vertex-partition ${\cal P}=\{V_{1},\dots,V_{h}\}$ of $V$ induces an \emph{$(h,\phi)$-vertex-expander hierarchy} if, for every level $\ell\le h$ and every connected component $\Gamma$ in $G_{\leq\ell}$, $V_{\ell}\cap\Gamma$ is $\phi$-vertex-expanding in $\Gamma$. That is, for every vertex cut $(L,S,R)$ of $\Gamma$, 
\[
|S|\ge\phi\min\{|V_{\ell}\cap(L\cup S)|,\, |V_{\ell}\cap(R\cup S)|\}.
\]

From ${\cal P}$, the connected components $\Gamma$ in $G_{\leq\ell}$ for all levels $\ell$ form a laminar family ${\cal C}$. Let ${\cal H}$ be the tree representation of ${\cal C}$. We also call $({\cal C},{\cal H})$ an $(h,\phi)$-vertex-expander hierarchy of $G$.
\label{def:vertex hie}
\end{defn}

The following theorem is analogous to \Cref{thm:edge exp hie}.
\begin{theorem}
There exists an algorithm that, given a graph $G$, computes an $(h,\phi)$-vertex-expander hierarchy with $h\le\log n$ and $\phi=1$ in exponential time, or $h\le\log n$ and $\phi\ge\Omega(1/\sqrt{\log n})$ in polynomial time. \label{thm:vertex exp hie}
\end{theorem}

Long and Saranurak's \cite{LongS22} vertex expander hierarchy is weaker, both qualitatively and structually. To be precise, the Long-Saranurak hierarchy only guarantees $\phi \geq 1/n^{o(1)}$, but it admits almost-linear construction time. Furthermore, the expander components in the 
Long-Saranurak hierarchy may not form a laminar family. 
The proof of \Cref{thm:vertex exp hie} follows immediately from the lemma below. 
\begin{lem}
There exists an exponential-time algorithm that, given a graph $G=(V,E)$, computes a vertex set $X$ such that every connected component of $G\setminus X$ contains at most $n/2$ vertices and $X$ is $1$-vertex-expanding. In polynomial time, we instead guarantee that $X$ is $\Omega(1/\sqrt{\log n})$-vertex-expanding.\label{lem:vertex exp sep}
\end{lem}

The proofs of \Cref{thm:vertex exp hie} and \Cref{lem:vertex exp sep} follow in exactly the same way as how we proved the analogous results in the edge version. We include them for completeness.
\begin{proof}
Initialize $X\gets V$. If $X$ is $1$-expanding, we are done. Otherwise, we repeatedly update $X$ as follows. Since $X$ is not $1$-expanding, there exists a vertex cut $(L,S,R)$ where $|L|\le n/2$ such that $|S|<\min\{|X\cap(L\cup S)|,|X\cap(R\cup S)|\}$. Update $X\gets X\setminus (X\cap L)\cup S$.

Let $X'$ denote $X$ after the update. Observe that every connected component of $G\setminus X'$ still contains at most $n/2$ vertices because $|L|\le n/2$. Moreover, $|X'|<|X|$ because, while we added at most $|S \setminus X|$ new vertices to $X$, we removed $|X\cap L|>|S \setminus X|$ vertices from $X$ where the inequality holds because $|S\cap X| + |S\setminus X|=|S| < |X\cap(L\cup S)|=|X\cap L|+|S\cap X|$. Therefore, there are at most $|V|$ iterations before $X$ is $1$-vertex-expanding. This concludes the proof of the exponential-time algorithm.

To get polynomial time, we instead apply the sparsest cut algorithm by \cite{feige2005improved} that, given $X$, either guarantees that $X$ is $\Omega(1/\sqrt{\log n})$-vertex-expanding or returns a vertex cut $(L,S,R)$ where $|L|\le n/2$ such that $|S|<\min\{|X\cap(L\cup S)|,|X\cap(R\cup S)|\}$.
\end{proof}
\begin{proof}
[Proof of \Cref{thm:vertex exp hie}]Given $G$, compute the vertex set $X$ from \Cref{lem:vertex exp sep} and set $V_{h}\gets X$. To compute $V_{h-1},V_{h-2},\dots,V_{1}$, we recurse on each connected component $C$ of $G\setminus X$. We have $h\le\log n$ because each component $C$ has size $|C|\le n/2$. 
\end{proof}

\paragraph{Notation in subsequent sections.}
Let $(\cal{C},\cal{H})$ be an $(h,\phi)$-vertex-expander hierarchy of $G$. 
For each level-$\ell$  component $\Gamma \in \cal{C}$, we define $\gamma := V_\ell \cap \Gamma$ to be the \emph{core} of $\Gamma$. The following \Cref{ob:VertexEH} is straightforward from the definition.

\begin{observation}
\label{ob:VertexEH}
We have the following.
\begin{enumerate}
\item\label{Property:VEH2} There is no edge connecting two disjoint components in ${\cal C}$.
\item\label{VEH2} For each component $\Gamma\in{\cal C}$, its core $\gamma$ is $\phi$-vertex-expanding in $\Gamma$.
\item The cores $\{\gamma\mid\Gamma\in{\cal C}\}$ partition $V(G)$.
\end{enumerate}
\end{observation}

By convention, the core of a $\Gamma$ 
decorated with subscripts, superscripts, and diacritic marks inherits those decorations, e.g.,
$\hat{\gamma}_i^j$ is the core of $\hat{\Gamma}_i^j$.
For two components $\Gamma,\Gamma'$ s.t. $\Gamma'\preceq \Gamma$ (resp. $\Gamma'\prec \Gamma$), we also write $\gamma'\preceq \gamma$ (resp. $\gamma'\prec\gamma$). For each component $\Gamma$, 
$N(\Gamma)$ denotes the set of 
neighbors of $\Gamma$ in $V-\Gamma$.
Define $N_{\hat{\gamma}}(\Gamma) = N(\Gamma) \cap \hat{\gamma}$, which is only non-empty when
$\gamma\prec\hat{\gamma}$.  We call such $N_{\hat{\gamma}}(\Gamma)$ \emph{neighbor sets}. 
For each vertex $v\in V(G)$, we use $\gamma_{v}$ to denote the unique core containing $v$, 
so $\Gamma_{v}$ denotes the corresponding 
component of $\gamma_{v}$.

\subsubsection{The Initial Structure: Low-Degree Steiner Trees and Shortcut Graphs}

Based on the vertex expander hierarchy, we construct low-degree Steiner trees and shortcut graphs, both of which will help answer connectivity queries.

\paragraph{Low-Degree Steiner Trees.} 

For each component $\Gamma\in{\cal C}$, using \Cref{lemma:LowDegreeSteinerTree}, we will compute a Steiner tree $T_{\gamma}$ with maximum degree $\Delta = O(1/\phi)$ that spans the core $\gamma$ in $G[\Gamma]$. 
Sometimes, we will call the vertices in Steiner tree $T_{\gamma}$ \emph{nodes}, just to be consistent with the terminology of \emph{extended Steiner trees}
introduced later. 
In particular, each node $u\in \gamma\subseteq V(T_{\gamma})$ is a \emph{terminal node}, and the other nodes $V(T_{\gamma})\setminus \gamma$ are \emph{Steiner nodes}. Observe that each vertex in $\Gamma$ will correspond to at most one node in $T_{\gamma}$, 
and vertices in $\gamma$ are in one-to-one 
correspondence with terminal nodes in $T_{\gamma}$.

A hierarchy with such low-degree Steiner trees but without the vertex-expanding property was first introduced by Duan and Pettie~\cite{DuanP20} 
as the \emph{low-degree hierarchy}, 
which has been shown to be useful for the vertex-failure connectivity problem in both the centralized \cite{DuanP20,LongS22} and labeling scheme \cite{ParterPP24} settings. Roughly speaking, these Steiner trees are useful because they serve as connectivity certificates. By the low-degree property, when $f$ vertices fail, each Steiner tree will be broken into 
at most $O(f/\phi)$ subtrees, 
each of them still being connected in the new graph. 
The query algorithm need only look for edges that reconnect the subtrees rather than determine connectivity from scratch.

Long and Saranurak~\cite{LongS22} gave an almost linear time 
algorithm to compute an $O(\log n/\phi)$-degree Steiner tree
spanning a $\phi$-vertex expanding set $A$ in $G$.
We give an improved algorithm that computes an 
$O(1/\phi)$-degree Steiner trees based on F\"urer and Raghavachari~\cite{FurerR94}, albeit with a slower running time. 
This improvement to the degree will shave logarithmic factors off 
our final label size. 
The algorithm below can be of independent interest. 
Its proof is deferred to \Cref{app:low-deg steiner tree}.

\begin{restatable}[Low-degree Steiner Trees]{lem}{steinertree}
    Given a graph $G$ such that a set $A\subseteq V(G)$ is $\phi$-vertex-expanding in $G$, there is an algorithm that computes an $O(1/\phi)$-degree Steiner tree that spans $A$ in $G$. The running time is $O(mn\log n)$.
\label{lemma:LowDegreeSteinerTree}
\end{restatable}

\paragraph{The Neighborhood Hitter ${\cal S}$.}
We want the Steiner trees to have low degree $\Delta$ so that $f$ 
vertex failures generate at most $f\Delta$ subtrees.  This argument 
only requires that \emph{failed} vertices have low degree. 
Following~\cite{ParterPP24},
we generate a partition $\mathcal{S}=(S_1,\ldots,S_{f+1})$ of the vertex set,
and build a version of the data structure for each $S_i\in \mathcal{S}$, which
one can think of as vertices that are \emph{not allowed} to fail.
By the pigeonhole principle, for any failure set $F$ there exists an $S=S_i$
such that $S\cap F=\emptyset$.  Thus, it is fine if, 
in the data structure with failure-free set $S$, 
all $S$-vertices have \emph{unbounded} degrees.

The main benefit of having a failure-free $S$ is to 
effectively reduce neighborhood sizes, as follows.
If we were to generate the partition $\mathcal{S}$ randomly, then with high probability either (i) $N_{\hat{\gamma}}(\Gamma)\cap S \neq \emptyset$ or 
(ii) $|N_{\hat{\gamma}}(\Gamma)| \leq \lambda_{\nb} \bydef O(f\log n)$.
In case (i) we can link $T_{\hat{\gamma}}$ and $T_{\gamma}$ without increasing
the degrees of non-$S$ vertices by much (see \emph{extended cores} below), 
and in case (ii) we have a good upper bound on $|N_{\hat{\gamma}}(\Gamma)|$.
In fact, it is possible to achieve this guarantee deterministically using 
the method of conditional expectations~\cite{ParterPP24}. Concretely, we just compute ${\cal S}$ by invoking \Cref{lemma:NeighborHitter} with all such neighbor sets $N_{\hat{\gamma}}(\Gamma)$ 
as the inputs.

\begin{lemma}[{\cite[Lemma 8.1]{ParterPP24}}]
Given a graph $G$ with a polynomial number of vertex 
sets $\{B_{k}\subseteq V\mid 1\leq k\leq \poly(n)\}$, 
there is a deterministic algorithm that computes a partition 
${\cal S} = \{S_{1},...,S_{f+1}\}$ of $V(G)$ s.t.~for each $S_{i}$ and $B_{k}$, 
either $S_{i}\cap B_{k}\neq \emptyset$ or 
$|B_{k}|\leq O(f\log n)$. The running time is polynomial.\label{lemma:NeighborHitter}
\end{lemma}

Henceforth we use $S$ to refer to an arbitrary part of the 
partition $\mathcal{S}$.
In the preprocessing phase we generate a data structure for 
each $S\in \mathcal{S}$, but in the context of a 
query $\ang{s,t,F}$, $S$ refers to any part for which $S\cap F=\emptyset$.

\paragraph{Extended Cores and Extended Steiner Trees.} 
Each component $\Gamma\in{\cal C}$ has an \emph{extended core} $\gamma^{\ext}$ defined
as follows:
\[
\gamma^{\ext} \bydef \gamma \cup \bigcup_{\substack{\gamma'\prec \gamma\\\text{s.t.~$N_{\gamma}(\Gamma')\cap S\neq\emptyset$}}} \gamma'.
\]
Observe that if $\gamma' \nsubseteq \gamma^{\ext}$, then $|N_{\gamma}(\Gamma')|\leq \lambda_{\nb} = O(f\log n)$.
Whenever $N_{\hat{\gamma}}(\Gamma)\cap S\neq \emptyset$ is non-empty, 
let $s_{\hat{\gamma}}(\Gamma)\in N_{\hat{\gamma}}(\Gamma)\cap S$ 
be an arbitrary representative in its neighborhood set.

Just as each core $\gamma$ has a Steiner tree $T_{\gamma}$,
the extended core $\gamma^{\ext}$ has an \emph{extended Steiner tree} $T_{\gamma^{\ext}}$. Each tree node in $V(T_{\gamma^{\ext}})$ is either a \emph{terminal node} or \emph{Steiner node}. As we will see in the construction, the terminal nodes are one-one corresponding to vertices in $\gamma^{\ext}$. Each Steiner node corresponds to exactly one vertex in $\Gamma$, while each vertex in $\Gamma$ can correspond to arbitrary numbers of Steiner node in $V(T_{\gamma^{\ext}})$.

\medskip

\noindent{\underline{Construction of $T_{\gamma^{\ext}}$.}} The construction of $T_{\gamma^{\ext}}$ is as follows.
\begin{enumerate}
\item First, we make a copy of $T_\gamma$, and for each strict descendant $\gamma' \subseteq \gamma^{\ext}$,
make a copy of $T_{\gamma'}$.

\item Let $P_{\gamma'\to\gamma}$ be a copy of an arbitrary simple path in the 
graph $G[\Gamma'\cup\{s_{\gamma}(\Gamma')\}]$ connecting the vertex $s_{\gamma}(\Gamma')\in \gamma$ (which corresponds to the terminal node $s_{\gamma}(\Gamma')\in V(T_{\gamma})$) to some vertex 
$v'\in \Gamma'$ such that $v'$ corresponds to some tree node in $V(T_{\gamma'})$.
\item Finally, we obtain $T_{\gamma^{\ext}}$ by attaching the copy of $T_{\gamma'}$ (for all strict descendants $\gamma'\subseteq \gamma^{\ext}$) to the copy of $T_{\gamma}$ using the path $P_{\gamma'\to\gamma}$. That is, we glue the endpoint $s_{\gamma}(\Gamma')$ of $P_{\gamma'\to\gamma}$ to the terminal node $s_{\gamma}(\Gamma')\in V(T_{\gamma})$, and glue the other endpoint $v'$ of $P_{\gamma'\to\gamma}$ to the tree node $v'\in V(T_{\gamma'})$.
\end{enumerate}

By the construction, $V(T_{\gamma^{\ext}})$ is made up of $V(T_{\gamma})$, $V(T_{\gamma'})$ of each strict descendant $\gamma'\subseteq \gamma^{\ext}$, and the \underline{internal} nodes of each path $P_{\gamma'\to\gamma}$.
We define the terminal nodes in $V(T_{\gamma^{\ext}})$ to be 
\begin{itemize}
\item the terminal nodes in $V(T_{\gamma})$ (they one-one correspond to vertices in $\gamma$), and 
\item the terminal nodes in $V(T_{\gamma'})$ for each strict descendant $\gamma'\subseteq \gamma^{\ext}$ (they one-one correspond to vertices in $\gamma'$).
\end{itemize}
Other tree nodes in $V(T_{\gamma^{\ext}})$ are Steiner nodes. By definition, the terminal nodes in $V(T_{\gamma^{\ext}})$ one-one correspond to vertices in $\gamma^{\ext}$.

\medskip

\noindent{\underline{Properties of $T_{\gamma^{\ext}}$.}} First, $T_{\gamma^{\ext}}$ has some kind of low degree guarantee of non-$S$ vertices, as shown in \Cref{lemma:ExtLowDegree}.

\begin{lemma}
For each vertex $v\in \Gamma\setminus S$, the tree nodes corresponding to $v$ have total $T_{\gamma^{\ext}}$-degree at most $O(h\Delta)$. 
\label{lemma:ExtLowDegree}
\end{lemma}
\begin{proof}
Recall that $V(T_{\gamma^{\ext}})$ is made up of $V(T_{\gamma})$, $V(T_{\gamma'})$ of each strict descendant $\gamma'\subseteq \gamma^{\ext}$, and the \underline{internal} nodes of each path $P_{\gamma'\to\gamma}$.
\begin{itemize}
\item $v$ can correspond to at most one tree node in $V(T_{\gamma})$. This tree node have degree at most $\Delta$ in $T_{\gamma^{\ext}}$ because it is not an endpoint of any path $P_{\gamma'\to\gamma}$ (since $v\notin S$).

\item For each strict descendant $\gamma'\subseteq \gamma^{\ext}$, if $v\in \Gamma'$, $v$ can correspond to at most one tree node in $V(T_{\gamma'})$ (note that if $v\notin \Gamma'$, $v$ must correspond to no tree node in $V(T_{\gamma'})$ since $V(T_{\gamma'})\subseteq \Gamma'$). This tree node have $T_{\gamma^{\ext}}$-degree at most $\Delta + 1$ because it has degree at most $\Delta$ in $T_{\gamma'}$ and it has degree $1$ in $P_{\gamma'\to\gamma}$ if it is an endpoint of $P_{\gamma'\to\gamma}$.

The number of such $\gamma'$ is at most $h$ since the number of components containing $v$ is at most $h$. Therefore, this part contributes at most $h(\Delta + 1)$.

\item For each path $P_{\gamma'\to\gamma}$, if $v\in \Gamma'$, $v$ can correspond to at most one internal node of $P_{\gamma'\to\gamma}$ (note that if $v\notin \Gamma'$, $v$ will correspond to no internal node of $P_{\gamma'\to\gamma}$ since all internal nodes are in $\Gamma'$). This node has $T_{\gamma^{\ext}}$-degree $2$.

Again the number of such path $P_{\gamma'\to\gamma}$ is at most $h$, and this part contributes at most $2h$.
\end{itemize}
\end{proof}

The second property is the simple \Cref{ob:ExtSteinerSubtree}. We will exploit it in the proof of \Cref{lemma:RecursiveConn} and \Cref{sec:strategy vertex}.

\begin{observation}
For each $\gamma'\prec\gamma$ such that $\gamma'\subseteq \gamma^{\ext}$, its Steiner tree $T_{\gamma'}$, or even $T_{\gamma'}\cup P_{\gamma'\to\gamma}$, is a subtree of $T_{\gamma^{\ext}}$, where $T_{\gamma'}\cup P_{\gamma'\to\gamma}$ denote the tree obtained by gluing the endpoint $v'$ of $P_{\gamma'\to\gamma}$ to the tree node $v'\in V(T_{\gamma'})$.
\label{ob:ExtSteinerSubtree}
\end{observation}

\medskip

\Cref{lemma:ExtLowDegree} shows that
when $f$ vertices fail, they break $T_{\gamma^{\ext}}$ into at most 
$O(fh\Delta)= O(f\log n/\phi)$ 
subtrees, since $S$-vertices are not allowed to fail. The analogue of \cref{lemma:ExtLowDegree} 
in \cite{ParterPP24} 
creates extended Steiner trees with degree
$\Delta + h$ rather than $O(h\Delta)$,
but might not satisfy the property that
$T_{\gamma'}$ is a subtree of $T_{\gamma^{\ext}}$, 
which is used in our labeling scheme.

\paragraph{The Shortcut Graphs.} 
The \emph{global shortcut graph} $\hat{G}$ is also defined 
w.r.t.~an arbitrary $S\in\mathcal{S}$.
It is formed by adding \emph{shortcut edges} to $G$, each
with an assigned \emph{type}.\footnote{Previous papers call a shortcut edge an artificial edge and call a shortcut graph an auxiliary graph. We change the names to make them more descriptive.}
For each component $\Gamma\in {\cal C}$ and 
each strict ancestor $\hat{\Gamma}$ of $\Gamma$, we define 
\begin{align*}
\hat{N}_{\hat{\gamma}}(\Gamma) &=
\left\{
\begin{array}{ll}
N_{\hat{\gamma}}(\Gamma) & \text{ if $N_{\hat{\gamma}}(\Gamma)$ is disjoint from $S$,}\\
\{s_{\hat{\gamma}}(\Gamma)\} & \text{ if $N_{\hat{\gamma}}(\Gamma)$ intersects $S$.}
\end{array}
\right.\\
\hat{N}(\Gamma) &= \bigcup_{\hat{\gamma}} \hat{N}_{\hat{\gamma}}(\Gamma).
\end{align*}

$\hat{G}$ is a \emph{typed} multigraph with the same vertex set as $G$ and
\begin{align*}
E(\hat{G}) &= E(G) \cup \bigcup_{\Gamma} \Clique(\hat{N}(\Gamma)),
\end{align*}
where all edges in $E(G)$ have type \emph{original} and
all edges in the clique $\Clique(\hat{N}(\Gamma))$ have type $\gamma$.

The \emph{shortcut graph w.r.t. $\Gamma\in \mathcal{C}$}, 
denoted by $\hat{G}_{\Gamma}$ is the subgraph of $\hat{G}$
induced by edges with both endpoints in $\Gamma$ 
and at least one in the core $\gamma$.

The idea of adding shortcut edges has appeared in prior works \cite{DuanP20,LongS22,ParterPP24,ChanPR11} on the 
vertex-failure connectivity problem. Intuitively, the simplest way to add shortcut edges is to add a clique on $N(\Gamma)$ for each component $\Gamma$. With the shortcut edges, when failed vertices come, if some component $\Gamma$ 
is \emph{unaffected} (it has no failed vertices)
the query algorithm can ignore vertices in $\Gamma$, and use the shortcut edges to capture the connectivity provided by $\Gamma$. 
However, generally the performance of the algorithm depends on the sparsity of the shortcut edges, so this naive construction will not give good algorithms. Indeed, most prior work \cite{DuanP20,LongS22,ParterPP24} introduced different sparsification techniques on shortcut edges. In our work, we sparsify the shortcut edges by adding a clique on the sparsified neighbor set $\hat{N}(\Gamma)$ instead of the original one $N(\Gamma)$.

\subsubsection{Structures Affected by Queries}

In this subsection, we will define notations and terms related to a particular query $\ang{s,t,F}$, and then introduce the query strategy from a high level. Recall that $S \in{\cal S}$ represents any part of the partition disjoint from $F$.

\paragraph{Affected Components/Cores, Affected Edges, and Query Graphs.} We introduce the following notations and terms.
\begin{itemize}
\item 
For each component $\Gamma\in {\cal C}$, if $\Gamma$ intersects $F\cup\{s,t\}$, we say $\Gamma$ is an \emph{affected component} and $\gamma$ is an \emph{affected core}, otherwise they are \emph{unaffected}. 
\item 
For each edge $e=\{u,v\}\in E(\hat{G}_{\Gamma})$ in the shortcut graph w.r.t.~$\Gamma$, it is an \emph{affected edge} if the type of $e$ is $\gamma'$ for some affected $\gamma'$.
Let $\hat{E}_{\Gamma,\aff}\subseteq E(\hat{G}_{\Gamma})$ collect all affected 
edges in $\hat{G}_\Gamma$, 
and let $\hat{E}_{\Gamma,\unaff} = E(\hat{G}_{\Gamma})\setminus \hat{E}_{\Gamma,\aff}$ be the set of \emph{unaffected} edges.
\item 
For each affected component $\Gamma$, its \emph{query set} is $Q_{\Gamma} = \gamma\cup \bigcup_{\affected\ \gamma'\prec \gamma}\gamma'$ and its \emph{extended query set} is $Q^{\ext}_{\Gamma} = \gamma^{\ext}\cup \bigcup_{\affected\ \gamma'\prec \gamma}\gamma'$. 
\item We define the \emph{query graph} $\hat{G}^{\qry}_{\Gamma}$ of $\Gamma$ to be $\hat{G}^{\qry}_{\Gamma} = \hat{G}_{\Gamma}[Q^{\ext}_{\Gamma}]\setminus \hat{E}_{\Gamma,\aff}$. Namely, the query graph $\hat{G}^{\qry}_{\Gamma}$ is the subgraph of $\hat{G}_{\Gamma}$ induced by the extended query set $Q^{\ext}_{\Gamma}$, excluding all affected edges.
\end{itemize}

\begin{observation}
The number of affected components is at most $h(f+2)$.
\label{ob:AffectedNum}
\end{observation}

\subsubsection{A Divide-and-Conquer Lemma}

Next, we state the key lemma, \Cref{lemma:MainCorrectness}, saying the connectivity after failures can be captured by the structure we defined in previous sections. Roughly, for any affected component $\Gamma$, the connectivity between vertices can be captured by either (1) shortcut edges in $\Gamma$, (2) extended Steiner trees $T_{\gamma^{\ext}}$, or (3) the recursive structure on affected children of $\Gamma$.
Naturally, this equivalence hints at a divide-and-conquer strategy by querying bottom-up from the hierarchy. We will formally describe this strategy in \Cref{sec:strategy vertex}.

Before stating \Cref{lemma:MainCorrectness}, we introduce some notations. For an undirected graph $H$ and a subset of vertices $A\subseteq V(H)$, we define $\Conn(A,H)$ to be an undirected graph on vertices $A$ s.t. an edge $(u,v)$ exists in $\Conn(A,H)$ if and only if $u$ and $v$ are connected in $H$. We note that when $H$ refers to an extended Steiner tree and $A\subseteq V(G)$ refers to a set of original vertices, this notation $\Conn(A,H)$ is still well-defined as long as each vertex in $A$ corresponds to exactly one terminal node in $H$. For an extended Steiner tree $T_{\gamma^{\ext}}$, we use $T_{\gamma^{\ext}}\setminus F$ to denote the forest by removing all nodes corresponding to vertices in $F$.

\begin{lemma}
\label{lemma:MainCorrectness}
Let $\Gamma\in{\cal C}$ be an affected component. Each pair of vertices $x,y\in Q_{\Gamma}\setminus F$ are connected in $G[\Gamma]\setminus F$ if and only if they are connected in the union of 
\begin{enumerate}
\item\label{Conn1} $\Conn(Q^{\ext}_{\Gamma}\setminus F, \hat{G}^{\qry}_{\Gamma}\setminus F)$, 
\item\label{Conn2} $\Conn(\gamma^{\ext}\setminus F, T_{\gamma^{\ext}}\setminus F)$, and 
\item\label{Conn3} $\bigcup_{\Gamma_{\child}} \Conn(Q_{\Gamma_{\child}}\setminus F, G[\Gamma_{\child}]\setminus F)$, 
where the union is over all \emph{affected} 
children $\Gamma_{\child}$ of $\Gamma$.
\end{enumerate}
\end{lemma}

\begin{proof}
It is relatively simple to see that $x,y$ are connected in $G_{\Gamma}\setminus F$ if they are connected in the union, because each of $\hat{G}^{\qry}_{\Gamma}\setminus F$, $T_{\gamma^{\ext}}\setminus F$ and $G[\Gamma_{\child}]\setminus F$ only 
use valid connectivity in $G[\Gamma]\setminus F$. To be precise, for the graph $\hat{G}^{\qry}_{\Gamma}\setminus F$, consider an edge $e=\{u,v\}$ in it.
\begin{itemize}
\item If $e$ has type \emph{original} 
then $e$ also exists in $G[\Gamma]\setminus F$.
\item Otherwise, $e$ is a shortcut edge with some type $\gamma'$ such that $\gamma'$ is unaffected. By definition, $u,v \in N(\Gamma')$ 
and $\Gamma'$ is disjoint from $F$, hence
$u,v$ are connected in $G[\Gamma]\setminus F$.
\end{itemize}
Similarly, $T_{\gamma^{\ext}}\setminus F$ is a subgraph of $G[\Gamma]\setminus F$
and 
$G[\Gamma_{\child}]\setminus F$ is obviously a subgraph of $G[\Gamma]\setminus F$.

From now we focus on proving the other direction. Let $P_{xy}$ be a simple path connecting $x$ and $y$ in $G[\Gamma]\setminus F$. We can write $P_{xy}$ as $P_{xy} = P_{1}\circ P_{2}\circ \cdots\circ P_{\ell}$ where the endpoints $u_{i},v_{i}$ of $P_{i}$ are the $Q_{\Gamma}$-vertices and each subpath $P_{i}$ are internally disjoint with $Q_{\Gamma}$. It suffices to show that for each subpath $P_{i}$, $u_{i}$ and $v_{i}$ are connected in the union. 

\medskip

\noindent\textbf{Case (a)}. Suppose $\gamma_{u_{i}},\gamma_{v_{i}}\prec \gamma$ (recall that $\gamma_{u_{i}},\gamma_{v_{i}}$ are the cores containing $u_{i},v_{i}$ respectively). Then there is a child $\Gamma_{\child}$ of $\Gamma$ s.t. $\gamma_{u_{i}},\gamma_{v_{i}}\preceq \gamma_{\child}$ and all vertices in $P_{i}$ are inside $\Gamma_{\child}\setminus F$. To see this, assume for contradiction that $P_{i}$ contains two vertices from two different children $\Gamma_{\child},\Gamma'_{\child}$ of $\Gamma$. However, by property \ref{Property:VEH2} of the hierarchy, $P_{i}$ must go through some vertex in $\gamma$ under this assumption, contradicting that $P_{i}$ is internally disjoint from $Q_{\Gamma}$. Furthermore, we know $\gamma_{\child}$ is affected because $\gamma_{u_{i}}$ and $\gamma_{v_{i}}$ are affected (since $\gamma_{u_{i}},\gamma_{v_{i}}\subseteq Q_{\Gamma}$ and $Q_{\Gamma}$ only collects affected cores). Also, note that $u_{i},v_{i}\in Q_{\Gamma_{\child}}\setminus F$ because $Q_{\Gamma}\cap \Gamma_{\child} = Q_{\Gamma_{\child}}$. Putting it all together, we know $u_{i}$ and $v_{i}$ are connected in $\Conn(Q_{\Gamma_{\child}}\setminus F, G[\Gamma_{\child}]\setminus F)$ (i.e. Part \ref{Conn3}) by the existence of $P_{i}$.  

\medskip

In what follows, we will argue that arbitrary two vertices $u,v\in Q_{\Gamma}\setminus F$ are connected in the union, if they satisfy
\begin{itemize}
\item[(i)] $u\in \gamma$ or $v\in \gamma$, and
\item[(ii)] there is an original edge $e=\{u,v\}\in E(G[\Gamma])$ or there is an \emph{unaffected} component $\Gamma'\prec \Gamma$ s.t. $u,v\in N(\Gamma')$.
\end{itemize}
Indeed, for each above subpath $P_{i}$ not in Case (a), its endpoints $u_{i}$ and $v_{i}$ must satisfy the conditions (i) and (ii). Condition (i) is easy to see. For condition (ii), if $P_{i}$ has only one edge, that this edge $e=\{u_{i},v_{i}\}$ is an original edge in $G[\Gamma]$. If $P_{i}$ has more than one edge, all internal vertices of $P_{i}$ fall in unaffected cores, because $P_{i}$ is internally disjoint from $Q_{\Gamma}$. Again by property \ref{Property:VEH2} of the hierarchy, there exists an unaffected component $\Gamma'$ containing all internal vertices of $P_{i}$, so $u_{i},v_{i}\in N({\Gamma'})$.

\medskip

\noindent\textbf{Case (b).} Suppose condition (ii) tells there is an original edge $e=\{u,v\}\in E(G[\Gamma])$. Then this edge is added to $\hat{G}_{\Gamma}$ because $u\in \gamma$ or $v\in \gamma$, and it is inside $\hat{G}^{\qry}_{\Gamma}$ because it has type original and $u,v\in Q_{\Gamma}\setminus F\subseteq Q^{\ext}_{\Gamma}\setminus F$. Therefore, $u$ and $v$ are connected in $\Conn(Q^{\ext}_{\Gamma}\setminus F, \hat{G}^{\qry}_{\Gamma}\setminus F)$ (i.e. Part \ref{Conn1}).

\medskip

From now on, we suppose condition (ii) tells there is an unaffected component $\Gamma'\prec\Gamma$ s.t. $u,v\in N({\Gamma'})$. We can further assume $u,v,\Gamma'$ satisfy condition 
\begin{itemize}
\item[(iii)] there is a path $P$ connecting $u$ and $v$ with all internal vertices inside ${\Gamma'}$, and $P$ intersects the core $\gamma'$.
\end{itemize}
This is without loss of generality by letting $\Gamma'$ be the \emph{minimal} component s.t. $u,v\in N({\Gamma'})$.

\medskip

\noindent\textbf{Case (c)}. Suppose $u,v\in \gamma$. 

\medskip

{\underline{Subcase (c1)}}. Suppose $N_{\gamma}({\Gamma'})$ is disjoint from $S$. Recall the way we add shortcut edges, we add a clique on $\hat{N}(\Gamma')\supseteq N_{\gamma}({\Gamma'})$ with type $\gamma'$. Because $u,v\in N(\Gamma')\cap \gamma = N_{\gamma}(\Gamma')$, there is a shortcut edge $e=\{u,v\}$ in this clique. Note that $e$ will be added into $\hat{G}_{\Gamma}$ because $u\in \gamma$ or $v\in \gamma$, and $e$ will survive in $\hat{G}^{\qry}_{\Gamma}\setminus F$ because $u,v\in Q^{\ext}_{\Gamma}\setminus F$ and the type $\gamma'$ is unaffected. Therefore, $u$ and $v$ are connected in $\Conn(Q^{\ext}_{\Gamma}\setminus F, \hat{G}^{\qry}_{\Gamma}\setminus F)$ (i.e. Part \ref{Conn1}).

\medskip

{\underline{Subcase (c2).}} Otherwise $N_{\gamma}(\Gamma')$ intersects $S$. By definition, $\gamma'$ is in the extended core $\gamma^{\ext}$ of component $\Gamma$. Condition (iii) tells that $P$ intersect $\gamma'$. Let $w_{u},w_{v}\in P\cap \gamma'$ be the $P$-vertices closest to $u$ and $v$ respectively. \Cref{claim:b2} tells that $w_{u}$ and $w_{v}$ are connected in $\Conn(\gamma^{\ext}\setminus F, T_{\gamma^{\ext}}\setminus F)$ (i.e. Part \ref{Conn2}).

\begin{claim}
\label{claim:b2}
We have $w_{u},w_{v}\in \gamma^{\ext}\setminus F$, and the terminal nodes $w_{u}$ and $w_{v}$ are connected in $T_{\gamma^{\ext}}\setminus F$.
\end{claim}
\begin{proof}
First, we have $w_{u},w_{v}\in \gamma'\setminus F\subseteq \gamma^{\ext}\setminus F$. Next, the terminal nodes $w_{u}$ and $w_{v}$ are connected in $T_{\gamma^{\ext}}\setminus F$ is because (1) the terminal nodes $w_{u}$ and $w_{v}$ fall in $T_{\gamma'}$ (since the vertices $w_{u},w_{v}$ are inside $\gamma'$), (2) the Steiner tree $T_{\gamma'}$ connects $w_{u}$ and $w_{v}$, (3) $T_{\gamma'}$ is a subtree of $T_{\gamma^{\ext}}$ (since $\gamma'\subseteq \gamma^{\ext}$) and (4) $T_{\gamma'}$ is disjoint from $F$ (since ${\Gamma'}$ is unaffected).
\end{proof}

It remains to show that $u\in\gamma$ and $w_{u}\in \gamma'$ are connected in the union, and that $w_{v}\in\gamma'$ and $v\in \gamma$ are connected in the union. Actually, these two claims can be derived from \Cref{lemma:RecursiveConn}. Here we show that $u$ and $w_{u}$ satisfy the requirements of \Cref{lemma:RecursiveConn} (so do $v$ and $w_{v}$ by a similar argument). For $u$, trivially $u\notin F$ and $u\in \gamma$. For $w_{u}$, recall that $\gamma_{w_{u}}$ denotes the core containing $w_{u}$, and we know $\gamma_{w_{u}} = \gamma'$. Hence, as we mentioned above, $\gamma_{w_{u}}\prec\gamma$, $\gamma_{w_{u}}$ is unaffected and $N_{\gamma}(\Gamma_{w_{u}})$ intersects $S$. Because we take the $w_{u}\in P\cap \gamma'$ closest to $u$, the subpath of $P$ from $u$ to $w_{u}$ either has one edge or has all interval vertices inside some $\gamma''\prec\gamma_{w_{u}}$, which means $e=\{u,w_{u}\}\in E(G[\Gamma])$ or $u,w_{u}\in N({\Gamma''})$.

\medskip

\noindent\textbf{Case (d).} Suppose $u\in \gamma$ and $v\in \gamma_{v}\prec \gamma$. Note that $u$ and $v$ satisfy the requirements of \Cref{lemma:RecursiveConn}, so they are connected in the union.

\begin{lemma}
Let $u,v$ be two vertices satisfying that
\begin{itemize}
\item $u\notin F$ and $u\in \gamma$;
\item $v\notin F$, $\gamma_{v}\prec \gamma$ and either
\begin{itemize}
\item $\gamma_{v}$ is affected or
\item $\gamma_{v}$ is unaffected and $N_{\gamma}(\Gamma_{v})$ intersects $S$;
\end{itemize}
\item The edge $e=\{u,v\}\in G[\Gamma]$, or there is an unaffected component $\gamma'\prec\gamma_{v}$ s.t. $u,v\in N({\Gamma'})$.
\end{itemize}
We have $u$ and $v$ are connected in the union.
\label{lemma:RecursiveConn}
\end{lemma}
\begin{proof}
First observe that $\gamma_{v}$ belongs to the extended core $\gamma^{\ext}$ of $\gamma$ because $N_{\gamma}(\Gamma_{v})$ intersects $S$.

Similar to the Case (b) above, if there is an original edge $e=\{u,v\}\in E(G[\Gamma])$, then this edge is added to $\hat{G}_{\Gamma}$ because $u\in \gamma$, and it survives in $\hat{G}^{\qry}_{\Gamma}\setminus F$ because it has type original, $u\in \gamma\in Q^{\ext}_{\Gamma}$ and $v\in \gamma_{v}\subseteq \gamma^{\ext}\subseteq Q^{\ext}_{\Gamma}$. Therefore, $u$ and $v$ are connected in $\Conn(Q^{\ext}_{\Gamma}\setminus F, \hat{G}^{\qry}_{\Gamma}\setminus F)$ (i.e. Part \ref{Conn1}).

From now we assume there is an unaffected component $\gamma'\prec \gamma_{v}$ s.t. $u,v\in N(\Gamma')$. The following argument is similar to the Case (d) above. Again, by choosing $\gamma'$ whose component ${\Gamma'}$ is the minimal one that contains all internal vertices of $P$, we can assume there is a path $P$ connecting $u$ and $v$ with all internal vertices falling in ${\Gamma'}$, and $P$ intersects $\gamma'$.

\medskip

\noindent\underline{Case (1).} Suppose $N_{\gamma_{v}}({\Gamma'})$ intersects $S$, and $v\neq s_{\gamma_{v}}({\Gamma'})$. We show that $v$ and $v'\overset{\text{def}}{=}s_{\gamma_{v}}(\Gamma')$ are connected in the union.

\begin{itemize}
\item If $\gamma_{v}$ is affected, the vertices $v$ and $v'$ are connected by Part \ref{Conn3} by the following reasons. Let $\Gamma_{\child}$ be the child of $\Gamma$ such that $\Gamma_{v}\preceq \Gamma_{\child}$. First, $\gamma_{\child}$ is affected because $\gamma_{v}\subseteq Q_{\Gamma}$ (it means $\gamma_{v}$ is affected) and $\gamma_{v}\preceq \gamma_{\child}$. Second, $v,v'\in Q_{\Gamma_{\child}}\setminus F$ because $v,v'\in \gamma_{v}\setminus F$ ($v'\notin F$ because $v'\in S$ and $S$ is disjoint from $F$), and $\gamma_{v}\subseteq  Q_{\Gamma}\cap {\Gamma_{\child}} = Q_{\Gamma_{\child}}$. Third, $v$ and $v'$ are connected in $G[\Gamma_{\child}]\setminus F$ because $v,v'\in N({\Gamma'})$, ${\Gamma'}$ is an unaffected component, and $\Gamma'\prec \Gamma_{v}\preceq \Gamma_{j}$. 

\item If $\gamma_{v}$ is unaffected and $N_{\gamma}(\Gamma_{v})$ intersects $S$, the terminal nodes $v$ and $v' = s_{\gamma_{v}}({\Gamma'})$ are connected in $\Conn(\gamma^{\ext}\setminus F, T_{\gamma^{\ext}}\setminus F)$ (i.e. Part \ref{Conn2}) by the following reasons. First, $T_{\gamma_{v}}$ is a subtree of $T_{\gamma^{\ext}}$ (since $\gamma_{v}\subseteq \gamma^{\ext}$ from $N_{\gamma}(\Gamma_{v})$ intersects $S$). Furthermore, $T_{\gamma_{v}}$ has no vertex in $F$ (since $\gamma_{v}$ is unaffected), and the terminal nodes in $T_{\gamma^{\ext}}$ corresponding to vertices $v,v'$ are on the subtree $T_{\gamma_{v}}$ (since $v,v'\in \gamma_{v}$). 
\end{itemize}
Therefore, it suffices to show that $u$ is connected to $v'$ in the union, which can be reduced to the following cases.

\medskip

\noindent\underline{Case (2).} Suppose $N_{\gamma_{v}}({\Gamma'})$ is disjoint from $S$, or $v = s_{\gamma_{v}}({\Gamma'})$. Further assume that $N_{\gamma}({\Gamma'})$ is disjoint from $S$. Recall the construction of shortcut edges, we have $u\in N_{\gamma}({\Gamma'})\subseteq \hat{N}(\Gamma')$ and $v \in \hat{N}(\Gamma')$ (if $N_{\gamma_{v}}(\Gamma')$ is disjoint from $S$, then $v\in N_{\gamma_{v}}(\Gamma')\subseteq \hat{N}(\Gamma')$; if $v = s_{\gamma_{v}}(\Gamma')$, then $v = s_{\gamma_{v}}(\Gamma')\subseteq \hat{N}(\Gamma')$), so there is a shortcut edge connected $u'$ and $v$ with type $\gamma'$. This edge is in $\hat{G}_{\Gamma}$ because $u\in\gamma$, and it survives in $\hat{G}^{\qry}_{\Gamma}\setminus F$ because $\gamma'$ is unaffected, $u\in \gamma\subseteq Q^{\ext}_{\Gamma}$ and $v\in \gamma_{v}\subseteq \gamma^{\ext}\subseteq Q^{\ext}_{\Gamma}$. Therefore, $u$ and $v$ are connected in $\Conn(Q^{\ext}_{\Gamma}\setminus F, \hat{G}^{\qry}_{\Gamma}\setminus F)$ (i.e. Part \ref{Conn1}).

\medskip

\noindent\underline{Case (3).} Suppose $N_{\gamma_{v}}({\Gamma'})$ is disjoint from $S$, or $v = s_{\gamma_{v}}({\Gamma'})$. Further assume that $N_{\gamma}({\Gamma'})$ intersects $S$. Let $u' = s_{\gamma}({\Gamma'})$. Let $w$ be the $P$-vertex in $\gamma'$ closest to $u$.

First, $v$ and $u'$ are connected in $\Conn(Q^{\ext}_{\Gamma}\setminus F, \hat{G}^{\qry}_{\Gamma}\setminus F)$ (i.e. Part \ref{Conn1}) by the following reasons. When adding shortcut edges, we have $v,u'\in\hat{N}(\Gamma')$ (by the same reason as above), so there is a shortcut edge connecting $u'$ and $v$ with type $\gamma'$. This edge is in $\hat{G}_{\Gamma}$ because $u\in\gamma$, and it survives in $\hat{G}^{\qry}_{\Gamma}\setminus F$ because $\gamma'$ is unaffected, $u'\in \gamma\subseteq Q^{\ext}_{\Gamma}$, $u'\notin F$ (since $u'\in S$) and $v\in \gamma_{v}\subseteq \gamma^{\ext}\subseteq Q^{\ext}_{\Gamma}$. 

Next, $u'$ and $w$ are connected in $\Conn(\gamma^{\ext}\setminus F, T_{\gamma^{\ext}}\setminus F)$ (i.e. Part \ref{Conn2}). This is because in the extended Steiner tree $T_{\gamma^{\ext}}$, the terminal nodes $u'$ and $w$ fall in subtree $T_{\gamma'}\cup P_{\gamma'\to\gamma}$ (from \Cref{ob:ExtSteinerSubtree}), and all vertices in $T_{\gamma'}\cup P_{\gamma'\to\gamma}$ are not in $F$ (since all vertices in $T_{\gamma'}$ fall in the unaffected component $\Gamma'$, and all vertices in $P_{\gamma'\to\gamma}$ also fall in $\Gamma'$ except one endpoint $s_{\gamma}(\Gamma')$, which is in $S$).

It remains to show that $w$ and $u$ are connected in the union. We can reuse the argument in the whole proof of \Cref{lemma:RecursiveConn}. We now verify that $u$ and $w$ satisfy the conditions of \Cref{lemma:RecursiveConn}. Because $u$ is unchanged, $u\notin F$ and $u\in \gamma$. For $w$ (note that $\gamma_{w} = \gamma'$), because $w\in\gamma_{w}\prec\gamma_{v}\prec\gamma$ and $\gamma_{v}$ is unaffected, we know $w\notin F$, $\gamma_{w}\prec\gamma$ and $\gamma_{w}$ is unaffected. Moreover, $N_{\gamma}(\gamma_{w})$ intersects $S$ is from the assumption of case (3). Next, consider the subpath of $P$ from $u$ to $w$, denoted by $P'$. If $P'$ has only one edge $e=\{u,w\}$, then we are good. Otherwise, $P'$ has some internal vertices. Note that by our choice of $w$, the internal vertices of $P'$ are all in $\Gamma'\setminus \gamma'$. By Property \Cref{Property:VEH2} of the hierarchy, they must belong to ${\Gamma'_{\child}}$ for some child $\gamma'_{\child}$ of $\gamma'$. In other words, the unaffected core $\gamma'_{\child}\prec \gamma_{w}$ has $u,w\in N({\Gamma'_{\child}})$. In conclusion $u,w$ indeed satisfy the conditions of \Cref{lemma:RecursiveConn}. Finally, this recursive argument will stop because when we reach Case (3) again, the $\gamma'$ this time will have depth larger than the one of the last time, and the hierarchy has finite depth.
\end{proof}
\end{proof}

\subsubsection{An Improved Divide-and-Conquer Lemma via Sparsified Shortcut Graphs}
\label{sect:DivideAndConquer}

The goal of this section is to show an improved version \Cref{lemma:MainCorrectness} which proves precisely the same statement except that the query graph $\hat{G}^{\qry}_{\Gamma}$ is replaced by its sparsified version $\wtilde{G}^{\qry}_{\Gamma}$. 

\begin{lem}
\label{coro:MainCorrectness}
Let $\Gamma\in{\cal C}$ be an affected component. For each pair of vertices $x,y\in Q_{\Gamma}\setminus F$, they are connected in $G[\Gamma]\setminus F$ if and only if they are connected in the union of 
\begin{enumerate}
\item $\Conn(Q^{\ext}_{\Gamma}\setminus F, \wtilde{G}^{\qry}_{\Gamma}\setminus F)$, 
\item $\Conn(\gamma^{\ext}\setminus F, T_{\gamma^{\ext}}\setminus F)$, and 
\item $\Conn(Q_{\Gamma_{\child}}\setminus F, G[\Gamma_{\child}]\setminus F)$ for all affected children $\Gamma_{\child}$ of $\Gamma$.
\end{enumerate}
\end{lem}

Working the sparsified query graph $\wtilde{G}^{\qry}_{\Gamma}$ is crucial for bounding the size of the vertex label to be $\poly(f \log n)$ in \Cref{sec:implementing vertex label}. This technique was also used in \cite{ParterPP24}. 
However, sparsified query graphs are not crucial for understanding the overall strategy of the algorithm in \Cref{sec:strategy vertex}. Hence, during the first read, we suggest readers assume $\wtilde{G}^{\qry}_{\Gamma} = \hat{G}^{\qry}_{\Gamma}$, skip this section, and continue until when sparsified query graphs are needed in \Cref{sec:implementing vertex label}.

Below, we define sparsified shortcut graphs and sparsified query graphs and then prove \Cref{coro:MainCorrectness}.

\paragraph{Sparsified Shortcut Graphs.} Here, we define the sparsified shortcut graph $\wtilde{G}_{\Gamma}$ of the shortcut graph $\hat{G}_{\Gamma}$ with respect to each component $\Gamma\in{\cal C}$. Roughly speaking but not precisely, we will sparsify the subgraph $\hat{G}_{\Gamma}[\gamma^{\ext}]$,
abbreviated as $\hat{G}_{\gamma^{\ext}}$, 
into a graph $\wtilde{G}_{\gamma^{\ext}}$ with arboricity $\wtilde{O}(f^{2})$ while keeping the (pairwise) connectivity unchanged under $\wtilde{O}(f^{2})$ vertex failures. 
We will see later why it should tolerate $\wtilde{O}(f^{2})$ 
failures rather than $f$.  
Let us recall the guarantees of the 
Nagamochi-Ibaraki~\cite{NagamochiI92} sparsifiers.

\begin{lemma}[Nagamochi and Ibaraki~\cite{NagamochiI92}]
\label{coro:FaultTolerantSparsify}
Given a \emph{simple} undirected graph $R$ and a parameter $d$, there is a deterministic algorithm that computes a subgraph $\wtilde{R}$ of $R$ with $V(\wtilde{R}) = V(R)$ satisfying the following.
\begin{enumerate}
\item $\wtilde{R}$ has arboricity $d$.
\item Given arbitrary vertex failures $F\subseteq V(R)$ s.t. $|F| < d$, each pair of vertices $u,v\in V(R)\setminus F$ 
are connected in $R\setminus F$ 
if and only if they are connected in $\wtilde{R}\setminus F$.
\end{enumerate}
\end{lemma}

A graph with arboriticy $d$ is one whose edge-set can be partitioned into $d$ forests and as a consequence, the edge-set can be \emph{oriented} so that the maximum out-degree is $d$.
Parter et al.~\cite{ParterPP24} also computed low-arboricity sparsifiers deterministically, but used hit-miss hash families~\cite{KarthikP21}
rather than Nagamochi-Ibaraki sparsification.

Finally, we are ready to formally describe how to sparsify $\hat{G}_{\Gamma}$ into $\wtilde{G}_{\Gamma}$.
\begin{enumerate}
\item Let $\hat{G}_{\gamma^{\ext}} = \hat{G}_{\Gamma}[\gamma^{\ext}]$ be the subgraph of $\hat{G}_{\Gamma}$ induced by the extended core $\gamma^{\ext}$, and let $\hat{G}_{\gamma^{\ext}}^{\sp}$ be the \emph{simple} graph corresponding to $\hat{G}_{\gamma^{\ext}}$, i.e., a bundle of edges with the same endpoints but different types collapse to one edge.
\item We obtain the sparsified graph $\wtilde{G}^{\sp}_{\gamma^{\ext}}$ of $\hat{G}^{\sp}_{\gamma^{\ext}}$ by applying \Cref{coro:FaultTolerantSparsify} with parameter 
$d = \lambda_{\arbo} \bydef f + h^{2}f\lambda_{\nb} + 1 = O(f^{2}\log^{3}n)$.
\item Obtain $\wtilde{G}_{\gamma^{\ext}}$ by including, 
for each $\{u,v\} \in E(\wtilde{G}^{\sp}_{\gamma^{\ext}})$, 
\emph{all} edges with endpoints 
$\{u,v\}$ in $\hat{G}_{\gamma^{\ext}}$ (such edges may have various types).
\item Lastly, obtain $\wtilde{G}_{\Gamma}$ by 
substituting $\wtilde{G}_{\gamma^{\ext}}$ for
$\hat{G}_{\gamma^{\ext}}$ in $\hat{G}_{\Gamma}$. 
Namely, $E(\wtilde{G}_{\Gamma}) = (E(\hat{G}_{\Gamma})\setminus E(\hat{G}_{\gamma^{\ext}}))\cup E(\wtilde{G}_{\gamma^{\ext}})$. 
By this definition, 
$\wtilde{G}_{\Gamma}[\gamma^{\ext}] = \wtilde{G}_{\gamma^{\ext}}$.
\end{enumerate}

\paragraph{Sparsified Query Graphs.}
Recall that we define the \emph{query graph} $\hat{G}^{\qry}_{\Gamma}$ of $\Gamma$ to be $\hat{G}^{\qry}_{\Gamma} = \hat{G}_{\Gamma}[Q^{\ext}_{\Gamma}]\setminus \hat{E}_{\Gamma,\aff}$. 
Naturally, we define a \emph{sparsified query graph} $\wtilde{G}^{\qry}_{\Gamma} = \wtilde{G}_{\Gamma}[Q^{\ext}_{\Gamma}]\setminus \hat{E}_{\Gamma,\aff}$.
That is, the sparsified query graph $\hat{G}^{\qry}_{\Gamma}$ is the subgraph of $\wtilde{G}_{\Gamma}$ induced by the extended query set $Q^{\ext}_{\Gamma}$, excluding all affected edges. 

Now, we proceed to prove \Cref{coro:MainCorrectness}.

\begin{lemma}
Let $\Gamma\in{\cal C}$ be an affected component. Two vertices $x,y\in \gamma^{\ext}\setminus F$ are connected in $\hat{G}^{\qry}_{\Gamma}[\gamma^{\ext}]\setminus F$ if and only if they are connected in $\wtilde{G}^{\qry}_{\Gamma}[\gamma^{\ext}]\setminus F$.
\label{lemma:ExtendedCoreSubgraphEquivalance}
\end{lemma}
\begin{proof}
By definition, $\hat{G}^{\qry}_{\Gamma}[\gamma^{\ext}] = \hat{G}_{\Gamma}[\gamma^{\ext}]\setminus \hat{E}_{\Gamma,\aff}$ and $\wtilde{G}^{\qry}_{\Gamma}[\gamma^{\ext}] = \wtilde{G}_{\Gamma}[\gamma^{\ext}]\setminus \hat{E}_{\Gamma,\aff}$, because $\gamma^{\ext}\subseteq Q^{\ext}_{\Gamma}$. Recall that when sparsifying $\hat{G}_{\Gamma}$ into $\wtilde{G}_{\Gamma}$, we construct $\wtilde{G}_{\gamma^{\ext}} = \wtilde{G}_{\Gamma}[\gamma^{\ext}]$ from $\hat{G}_{\gamma^{\ext}} = \hat{G}_{\Gamma}[\gamma^{\ext}]$ by the following way.
\begin{itemize}
\item First let $\hat{G}_{\gamma^{\ext}}^{\sp}$ be the corresponding simple graph of $\hat{G}_{\gamma^{\ext}}$.
\item $\wtilde{G}_{\gamma^{\ext}}^{\sp}$ is the sparsified simple graph from applying \Cref{coro:FaultTolerantSparsify} on $\hat{G}_{\gamma^{\ext}}^{\sp}$ with parameter $d = f + h^{2}f\lambda_{\nb} + 1$.
\item Lastly let $\wtilde{G}_{\gamma^{\ext}} = \hat{G}_{\gamma^{\ext}}\cap E(\wtilde{G}^{\sp}_{\gamma^{\ext}})$.
\end{itemize}

Trivially, $x$ and $y$ are connected in $\hat{G}^{\qry}_{\Gamma}[\gamma^{\ext}]$ if they are connected in $\wtilde{G}^{\qry}_{\Gamma}[\gamma^{\ext}]$, because $\wtilde{G}^{\qry}_{\Gamma}[\gamma^{\ext}]$ is a subgraph of $\hat{G}^{\qry}_{\Gamma}[\gamma^{\ext}]$ by definition.

From now we focus on the other direction. It suffices to show that, for each edge $\hat{e}\in E(\hat{G}^{\qry}_{\Gamma}[\gamma^{\ext}]) = E(\hat{G}_{\gamma^{\ext}})\setminus \hat{E}_{\Gamma,\aff}$ with endpoints $u,v\notin F$, we have $u$ and $v$ are connected in $\wtilde{G}^{\qry}_{\Gamma}[\gamma^{\ext}]\setminus F = (\wtilde{G}_{\gamma^{\ext}}\setminus\hat{E}_{\Gamma,\aff})\setminus F$. First, if $\{u,v\}\in E(\wtilde{G}_{\gamma^{\ext}}^{\sp})$, then trivially $e\in E(\wtilde{G}_{\gamma^{\ext}})$ by the construction, which immediately gives $u$ and $v$ are connected in $(\wtilde{G}_{\gamma^{\ext}}\setminus \hat{E}_{\Gamma,\aff})\setminus F$. 

Hence, we assume $u$ and $v$ are not adjacent in $\wtilde{G}^{\sp}_{\gamma^{\ext}}$ from now. Recall the construction of shortcut graphs, for each $\gamma'\prec \gamma$, the shortcut edges with type $\gamma'$ must have their endpoints in $\hat{N}(\Gamma')$, and we know that $|\hat{N}(\Gamma')|\leq h\cdot\lambda_{\nb}$. Let 
\[
V_{F} = ((F\cup\bigcup_{\text{affected }\gamma'} \hat{N}(\Gamma'))\cap \gamma^{\ext})\setminus \{u,v\}.
\]
That is, $V_{F}$ collects all $\gamma^{\ext}$-vertices which are failed or incident to some affected edges, excluding $u$ and $v$. We have that $|V_{F}|\leq f + h^{2}(f+2)\lambda_{\nb}$, because the number of affected $\gamma'$ is at most $h(f+2)$ by \Cref{ob:AffectedNum}.

Let $\wtilde{G}_{\gamma^{\ext},\valid}^{\sp}$ be the corresponding simple graph of $(\wtilde{G}_{\gamma^{\ext}}\setminus\hat{E}_{\Gamma,\aff})\setminus F$. Because $\hat{e}\in E(\hat{G}_{\gamma^{\ext}})$, we have $\{u,v\}\in E(\hat{G}^{\sp}_{\gamma^{\ext}})$, and furthermore $\{u,v\}\in E(\hat{G}^{\sp}_{\gamma^{\ext}}\setminus V_{F})$ because $u,v\notin V_{F}$. This means $u$ and $v$ are connected in $\wtilde{G}^{\sp}_{\gamma^{\ext}}\setminus V_{F}$ by \Cref{coro:FaultTolerantSparsify} and the fact that $|V_{F}|\leq d-1$. Finally, by \Cref{claim:4.2}, $u$ and $v$ are connected in $\wtilde{G}^{\sp}_{\gamma^{\ext},\valid}$, so they are also connected in $(\wtilde{G}_{\gamma^{\ext}}\setminus \hat{E}_{\Gamma,\aff})\setminus F$ as desired.

\begin{claim}
$\wtilde{G}^{\sp}_{\gamma^{\ext}}\setminus V_{F}$ is a subgraph of $\wtilde{G}_{\gamma^{\ext},\valid}^{\sp}$.
\label{claim:4.2}
\end{claim}
\begin{proof}
Consider an edge $\{x,y\}\in E(\wtilde{G}_{\gamma^{\ext}}^{\sp}\setminus V_{F})$, and we will show that $\{x,y\}\in E(\wtilde{G}^{\sp}_{\gamma^{\ext},\valid})$. Because we have assumed that $u$ and $v$ are not adjacent in $\wtilde{G}^{\sp}_{\gamma^{\ext}}$, we have either $x$ or $y$ is not in $\{u,v\}$, say $x\notin\{u,v\}$. From the construction of $\wtilde{G}_{\gamma^{\ext}}$, there must be an edge $\wtilde{e}\in E(\wtilde{G}_{\gamma^{\ext}})$ connecting $x,y$. 

We now show that $\wtilde{e}$ is an unaffected edge. Assume for contradiction that $\wtilde{e}$ is affected. Let $\gamma'$ be the type of $\wtilde{e}$. Then $\gamma'$ is affected, and $x\in \hat{N}_{\Gamma'}$. Because we also have $x\notin\{u,v\}$ and $V(\wtilde{G}^{\sp}_{\gamma^{\ext}}) = \gamma^{\ext}$, we know $x\in (\hat{N}_{\Gamma'}\cap \gamma^{\ext})\setminus\{u,v\}\subseteq V_{F}$, a contradiction.

Now because $\wtilde{e}\in E(\wtilde{G}_{\gamma^{\ext}})$ is an unaffected edge connecting $x,y\notin F$ (because $F\subseteq V_{F}$), $\wtilde{e}$ shows up in $(\wtilde{G}_{\gamma^{\ext}}\setminus \hat{E}_{\Gamma,\aff})\setminus F$, and $\{x,y\}\in E(\wtilde{G}^{\sp}_{\gamma^{\ext},\valid})$.
\end{proof}
\end{proof}

\begin{corollary}
Let $\Gamma\in{\cal C}$ be an affected component. Two vertices $x,y\in Q^{\ext}_{\Gamma}\setminus F$ are connected in $\hat{G}^{\qry}_{\Gamma}\setminus F$ if and only if they are connected in $\wtilde{G}^{\qry}_{\Gamma}\setminus F$.
\label{coro:ConnEqSparsified}
\end{corollary}
\begin{proof}
Note that the edge set $E(\hat{G}_{\Gamma})\setminus E(\hat{G}_{\Gamma}[\gamma^{\ext}])$ is exactly the same as $E(\wtilde{G}_{\Gamma})\setminus E(\wtilde{G}_{\Gamma}[\gamma^{\ext}])$. This means $E(\hat{G}^{\qry}_{\Gamma})\setminus E(\hat{G}^{\qry}_{\Gamma}[\gamma^{\ext}])$ is exactly the same as $E(\wtilde{G}^{\qry}_{\Gamma})\setminus E(\wtilde{G}^{\qry}_{\Gamma}[\gamma^{\ext}])$, because $\hat{G}^{\qry}_{\Gamma},\wtilde{G}^{\qry}_{\Gamma}$ are just constructed from $\hat{G}_{\Gamma},\wtilde{G}_{\Gamma}$ by removing affected edges and taking the restriction on vertices $Q^{\ext}_{\Gamma}$. Combining \Cref{lemma:ExtendedCoreSubgraphEquivalance}, we get this corollary immediately.
\end{proof}

Combining \Cref{lemma:MainCorrectness} and \Cref{coro:ConnEqSparsified}, we get \Cref{coro:MainCorrectness} immediately.

\subsection{The Strategy for Handling Queries} 
\label{sec:strategy vertex}
In this section, we will describe the query algorithm, but we will assume some interfaces
along the way. The labeling scheme for these interfaces will be deferred to the next subsection.

As we discussed above, the query algorithm will solve the connectivity of $Q_{\Gamma}\setminus F$ on $G[\Gamma]\setminus F$ for all affected components $\Gamma$ in a bottom-up order. In fact, here we already need an interface $\ListAffectedComps()$ which lists all the affected components in a bottom-up order.

At a particular step $\Gamma$, naturally the output is a partition ${\cal P}_{\Gamma}$ of $Q_{\Gamma}\setminus F$ which capture the connectivity (i.e. two vertices $x,y\in Q_{\Gamma}\setminus F$ are in the same group of ${\cal P}_{\Gamma}$ if and only if they are connected in $G[\Gamma]\setminus F$). However, in our implementation, we will represent ${\cal P}_{\Gamma}$ implicitly by a partition ${\cal K}_{\Gamma}$ of \emph{subtrees} on the Steiner trees. %

\begin{definition}[Subtrees w.r.t. $F$]
\label{def:Subtrees}
For each affected core $\gamma$, we break the Steiner tree $T_{\gamma}$ into \emph{subtrees} by removing failed vertices $F$, and let ${\cal T}_{\gamma}$ be the set of these subtrees. Similarly, for the extended Steiner tree $T_{\gamma^{\ext}}$, we define the subtrees ${\cal T}_{\gamma^{\ext}}$ by removing \emph{all nodes} corresponding to failed vertices. For each subtree $\tau$, we let $V^{\tmn}(\tau)$ denote the terminal nodes in $\tau$, and without ambiguity, $V^{\tmn}(\tau)$ also refers to the vertices in $G$ corresponding to the terminal nodes in $V^{\tmn}(\tau)$.
\end{definition}

For each affected component $\Gamma$, we define
\[
{\cal T}_{\Gamma} = \bigcup_{\affected\ \gamma'\preceq \gamma}{\cal T}_{\gamma'}
\]
Note that $\{V^{\tmn}(\tau)\mid \tau\in{\cal T}_{\Gamma}\}$ partitions $Q_{\Gamma}\setminus F$, because $\{\gamma'\setminus F\mid \affected\ \gamma'\preceq \gamma\}$ partitions $Q_{\Gamma}\setminus F$ (by \Cref{ob:VertexEH} and the definition of $Q_{\Gamma}$), and for each affected $\gamma'\preceq \gamma$, $\{V^{\tmn}(\tau)\mid \tau\in{\cal T}_{\gamma'}\}$ partitions $\gamma'\setminus F$.
Moreover, each subtree $\tau\in{\cal T}_{\Gamma}$ certifies that vertices in $V^{\tmn}(\tau)$ are connected in $G[\Gamma]\setminus F$. Therefore, we can indeed represent the partition ${\cal P}_{\Gamma}$ of vertices $Q_{\Gamma}\setminus F$ using a partition ${\cal K}_{\Gamma}$ of subtrees ${\cal T}_{\Gamma}$. More precisely, the output of the step at $\Gamma$ is a partition ${\cal K}_{\Gamma}$ of ${\cal T}_{\Gamma}$ s.t. its mapping
${\cal P}_{\Gamma}=\{\bigcup_{\tau\in K}V^{\tmn}(\tau)\mid K\in {\cal K}_{\Gamma}\}$ on $Q_{\Gamma}\setminus F$ capture the connectivity of $Q_{\Gamma}\setminus F$ in $G[\Gamma]\setminus F$.

\paragraph{Solving New Connectivity at $\Gamma$.} As shown in \Cref{coro:MainCorrectness}, the new connectivity at step $\Gamma$ is the connectivity of $Q^{\ext}_{\Gamma}\setminus F$ in $\wtilde{G}^{\qry}_{\Gamma}\setminus F$ and the connectivity of (terminal nodes) $\gamma^{\ext}\setminus F$ in $T_{\gamma^{\ext}}\setminus F$. Again, we want a partition ${\cal P}^{\ext}_{\Gamma}$ of $Q^{\ext}_{\Gamma}\setminus F$ that captures the new connectivity, and in the implementation we will represent ${\cal P}^{\ext}_{\Gamma}$ implicitly by a partition ${\cal K}^{\ext}_{\Gamma}$ of subtrees of 
extended Steiner subtrees.

Formally, we define
\[
{\cal T}^{\ext}_{\Gamma} = {\cal T}_{\gamma^{\ext}} \cup \bigcup_{\substack{\affected\ \gamma'\prec\gamma\\\text{s.t. $\gamma'$ is not in $\gamma^{\ext}$}}}{\cal T}_{\gamma'}.
\]
We also have $\left\{V^{\tmn}(\tau)\mid \tau\in{\cal T}^{\ext}_{\Gamma}\right\}$ 
partitions $Q^{\ext}_{\Gamma}\setminus F$. Our goal is to compute a partition ${\cal K}^{\ext}_{\Gamma}$ of ${\cal T}^{\ext}_{\Gamma}$ with its mapping ${\cal P}^{\ext}_{\Gamma} = \{\bigcup_{\tau\in K} V^{\tmn}(\tau)\mid K\in{\cal K}^{\ext}_{\Gamma}\}$ satisfying the following.
\begin{enumerate}
\item\label{req1} Each $u,v\in Q^{\ext}_{\Gamma}\setminus F$ in the same part of ${\cal P}^{\ext}_{\Gamma}$ 
are connected in $G[\Gamma]\setminus F$.
\item\label{req2} Each $u,v\in Q^{\ext}_{\Gamma}\setminus F$ connected in $\wtilde{G}^{\qry}_{\Gamma}\setminus F$ are in the same part of ${\cal P}^{\ext}_{\Gamma}$.
\item\label{req3} Each $u,v\in \gamma^{\ext}\setminus F\subseteq Q^{\ext}_{\Gamma}\setminus F$ whose terminal nodes are connected in $T_{\gamma^{\ext}}\setminus F$
are in the same part of ${\cal P}^{\ext}_{\Gamma}$.
\end{enumerate}
Note that requirement \ref{req3} is automatically satisfied because each pair of terminal nodes $u,v\in \gamma^{\ext}\setminus F$ connected in $T_{\gamma^{\ext}}\setminus F$ must belong to $V^{\tmn}(\tau)$ for the same subtree $\tau\in {\cal T}_{\gamma^{\ext}}$.

To compute a partition ${\cal K}^{\ext}_{\Gamma}$ satisfying 
requirements \ref{req1} and \ref{req2}, we will exploit the following key observation, \Cref{lemma:GiantComponent}, from the vertex expander hierarchy. Roughly speaking, those subtrees with too many $\gamma$-vertices inside and surrounding must belong to the same connected component of $G_{\Gamma}\setminus F$. 

Before stating \Cref{lemma:GiantComponent}, we introduce some notations. For a graph $G'$ and a vertex set $A\subseteq V(G')$, we let $N^{G'}(A)$ denote the neighbors of $A$ in $G'$. For each subtree $\tau\in {\cal T}^{\ext}_{\Gamma}$, let $V^{\tmn}_{\gamma}(\tau) = V^{\tmn}(\tau)\cap \gamma$ be the $\gamma$-vertices inside $V^{\tmn}(\tau)$,\footnote{In fact, $V^{\tmn}_{\gamma}(\tau)$ is not empty only when $\tau\in {\cal T}_{\gamma^{\ext}}$.} and let $N_{\gamma}(\tau) = N^{\wtilde{G}^{\qry}_{\Gamma}}(V^{\tmn}(\tau))\cap \gamma$ be the intersection of $\gamma$ and the neighbors of $V^{\tmn}(\tau)$ in $\wtilde{G}^{\qry}_{\Gamma}$.

\begin{lemma}
\label{lemma:GiantComponent}
Let $\tau_{1},\tau_{2}\in{\cal T}^{\ext}_{\Gamma}$ be two subtrees such that $|N_{\gamma}(\tau_{1})| + |V^{\tmn}_{\gamma}(\tau_{1})| > f/\phi$ and 
$|N_{\gamma}(\tau_{2})| + |V^{\tmn}_{\gamma}(\tau_{2})|>f/\phi$. 
Then $V^{\tmn}(\tau_{1})$ and $V^{\tmn}(\tau_{2})$ are contained in the same connected component of $G[\Gamma]\setminus F$.
\end{lemma}
\begin{proof}
Assume for contradiction that $\tau_{1}$ and $\tau_{2}$ belong to two different connected components $C_{1}$ and $C_{2}$ of $G[\Gamma]\setminus F$. 
It must be that 
$N^{G[\Gamma]}(C_i) \subseteq F$ for $i\in\{1,2\}$.

We clearly have 
$V^{\tmn}_{\gamma}(\tau_{1})\cup N_{\gamma}(\tau_{1})\subseteq \left(C_{1}\cup N^{G[\Gamma]}(C_{1})\right)\cap \gamma$.
$V^{\tmn}_\gamma(\tau_1)$ is contained in 
$C_1\cap \gamma$ by assumption.  A $y\in N_\gamma(\tau_1)$ is joined by an edge $\{x,y\}\in E(\wtilde{G}^{\qry}_{\Gamma})$ s.t. $x\in \tau_1$.
This edge can have type \emph{original}, 
or type $\gamma'$ for some 
unaffected $\gamma'\preceq \gamma$.  
Either way, there is a path from $x$ to $y$ whose
internal vertices are in $\Gamma'$. Note that 
$\Gamma'$ is disjoint from $F$, so
$y\in C_1 \cup N^{G[\Gamma]}(C_1)$.
We have shown that for $i\in\{1,2\}$,
\[
\left|(C_{i}\cup N^{G[\Gamma]}(C_{i}))\cap \gamma\right|\geq \left|V^{\tmn}(\tau_{i})\right| + \left|N_{\gamma}(\tau_{i})\right|>f/\phi.
\]
Because $\gamma$ is $\phi$-vertex-expanding in $G[\Gamma]$, we have
\begin{align*}
\left|N^{G[\Gamma]}(C_{1})\right| &\geq \phi\cdot \min\left\{\left|(C_{1}\cup N^{G[\Gamma]}(C_{1}))\cap \gamma\right|, \left|(\Gamma\setminus C_{1})\cap \gamma\right|\right\}>\phi\cdot (f/\phi) = f.
\end{align*}
However, $|N^{G[\Gamma]}(C_{1})|>f$ contradicts the fact that $N^{G[\Gamma]}(C_{1})\subseteq F$.
\end{proof}

We say a subtree $\tau\in{\cal T}^{\ext}_{\Gamma}$ is \emph{giant} if $|V^{\tmn}_{\gamma}(\tau)| + |N_{\gamma}(\tau)|>f/\phi$, otherwise it is \emph{non-giant}. 
We start from a trivial partition of ${\cal T}^{\ext}_{\Gamma}$ in which each subtree forms a singleton group, and then merge the groups according to the following rules.
\begin{enumerate}
\item[{\bf R1.}] Put all giant subtrees into the same group, called the \emph{giant group}. 
Whenever a connected group $\{\tau_i\}$ of subtrees
collectively has 
$|V^{\tmn}_\gamma(\{\tau_i\})|+|N_\gamma(\{\tau_i\})| > f/\phi$,
merge it with the giant group.
\item[{\bf R2.}] Let $\tau_{x}\in {\cal T}^{\ext}_{\Gamma}$ be non-giant and $\tau_{y}\in {\cal T}^{\ext}_{\Gamma}$ be such that $N_{\gamma}(\tau_{x})$ intersects $V^{\tmn}(\tau_{y})$.
Then merge the groups containing $\tau_{x}$ and $\tau_{y}$.
\item[{\bf R3.}] Let $\tau_{y}\in{\cal T}^{\ext}_{\Gamma}$ be non-giant and $\tau_{x}\in {\cal T}^{\ext}_{\Gamma}$ be giant such that $N_{\gamma}(\tau_{x})$ intersects $V^{\tmn}(\tau_{y})$. Then merge $\tau_{y}$ with the giant group.
\end{enumerate}
Before we move on to discuss the implementation of these rules, we first show that the partition ${\cal K}^{\ext}_{\Gamma}$ generated as above will satisfy requirements \ref{req1} and \ref{req2}. The requirement \ref{req1} is satisfied because the rules are using valid connectivity in $G[\Gamma]\setminus F$. In particular, ${\bf R1}$ is safe according to \Cref{lemma:GiantComponent}. For each of ${\bf R2}$ and ${\bf R3}$, $N_{\gamma}(\tau_{x})$ intersects $V^{\tmn}(\tau_{y})$ means there is an edge in $E(\wtilde{G}^{\qry}_{\Gamma})$ connecting $x\in V^{\tmn}(\tau_{x})$ and $y\in V^{\tmn}(\tau_{y})$, which means $x$ and $y$ are connected in $G[\Gamma]\setminus F$ because this edge is unaffected and $x,y\notin F$ (note that $V^{\tmn}(\tau_{x})$ and $V^{\tmn}(\tau_{y})$ are disjoint from $F$). To show that requirement \ref{req2} is satisfied, we will show that for each edge $\{x,y\}\in E(\wtilde{G}^{\qry}_{\Gamma}\setminus F)$ connecting $V^{\tmn}(\tau_{1})$ and $V^{\tmn}(\tau_{2})$, the subtrees $\tau_{1}$ and $\tau_{2}$ will be merged by one of the rules. Recall that the edges in graph $\hat{G}^{\qry}_{\Gamma}$ must be incident to $\gamma$, and $\wtilde{G}^{\qry}_{\Gamma}$ is a subgraph of $\hat{G}^{\qry}_{\Gamma}$. Hence we have either $x\in \gamma$ or $y\in\gamma$, and without loss of generality, we assume $y\in\gamma$.

\begin{itemize}
\item When both $\tau_{1}$ and $\tau_{2}$ are giant they will be merged by {\bfseries R1}.
\item When both $\tau_{1}$ and $\tau_{2}$ are non-giant 
they will be merged by {\bfseries R2}.
\item When $\tau_1$ is non-giant and $\tau_2$ is giant,
they will be merged by either {\bfseries R2} or {\bfseries R3} (depending on whether $y\in V^{\tmn}(\tau_{2})$ or $y\in V^{\tmn}(\tau_{1})$).
\end{itemize}

The following primitives will allow us to 
implement rules 
{\bfseries R1}, {\bfseries R2}, and {\bfseries R3}.
Here $\Gamma$ is an affected component and $\tau\in\mathcal{T}^{\ext}_{\Gamma}$.
\begin{description}
    \item[$\ListSubtrees^{\ext}(\Gamma)$:] 
    List all subtrees of ${\cal T}^{\ext}_{\Gamma}$.
    \item[$\ListSubtrees(\Gamma)$:]
    List all subtrees of $\mathcal{T}_\Gamma$.
    \item[$\ListTerminals(\tau, \Gamma)$:] 
    Return up to $f/\phi+1$ elements of 
    $V^{\tmn}_{\gamma}(\tau)$.
    \item[$\ListNeighbors(\tau, \Gamma)$:] 
    Return up to $f/\phi+1$ elements of 
    $N_{\gamma}(\tau)$.
    \item[$\IsTerminal(v,\tau)$:] Return \textsf{true} iff 
    $v$ is a terminal in $\tau$.
    \item[$\PickTerminal(\tau)$:] Return any element
    of $V^{\tmn}(\tau)$.
    \item[$\EnumFromGiant(\tau,\Gamma)$:] Requirement: $\tau\in\mathcal{T}_{\gamma^{\ext}}$, and $\tau$ is 
    non-giant.
    Return the \emph{number} of edges 
    in $E(\wtilde{G}_{\Gamma}^{\qry} \setminus F)$
    joining 
    $V^{\tmn}_{\gamma}(\tau)$ and
    $\bigcup_{\text{giant }\tau_{x}\in{\cal T}^{\ext}_{\Gamma}} V^{\tmn}(\tau_{x})$.
\end{description}

Clearly $\ListSubtrees^{\ext}(\Gamma)$,
$\ListTerminals(\tau,\Gamma)$,
and $\ListNeighbors(\tau,\Gamma)$ can be used to 
list all subtrees and determine which are giant.
This suffices to implement rule {\bfseries R1}.
$\ListNeighbors$ and $\IsTerminal$ suffice
to implement rule {\bfseries R2}.
We implement rule {\bfseries R3} \emph{non-constructively}.
In other words, we do not find a \emph{specific} giant
$\tau_x$ such that 
$N_\gamma(\tau_x)\cap V^{\tmn}(\tau_y) \neq \emptyset$,
but merely infer that there exists such a $\tau_x$,
if $\EnumFromGiant(\tau_y,\Gamma)>0$.
Observe that the condition of rule 
{\bfseries R3} is $\tau_y \in \mathcal{T}_\Gamma^{\ext}$.
If $\tau_y \not\in \mathcal{T}_{\gamma^{\ext}}$ 
then we cannot call $\EnumFromGiant$, but
in this case $\tau_y \in \mathcal{T}_\Gamma^{\ext}
\setminus \mathcal{T}_{\gamma^{\ext}}$ implies
$V_\gamma^{\tmn}(\tau_y)=\emptyset$, 
so {\bfseries R3} cannot be applied anyway.

\paragraph{Obtain ${\cal K}_{\Gamma}$ from ${\cal K}^{\ext}_{\Gamma}$ and all ${\cal K}_{\Gamma_{\child}}$.} Next, we discuss how to compute ${\cal K}_{\Gamma}$ for a particular affected component $\Gamma$, assuming we have already got the 
${\cal K}^{\ext}_{\Gamma}$ above (capturing the new connectivity at $\Gamma$) and the ${\cal K}_{\Gamma_{\child}}$ of all affected children $\Gamma_{\child}$ of $\Gamma$. 

We call $\ListSubtrees(\Gamma)$ and form their 
initialize partition ${\cal K}_{\Gamma}$ as:
\[
{\cal K}^{0}_{\Gamma} = \{\{\tau\}\mid \tau\in {\cal T}_{\gamma}\}\cup\bigcup_{\text{$\Gamma$'s affected children $\Gamma_{\child}$}} {\cal K}_{\Gamma_{\child}}.
\]
That is, any subtrees grouped in $\mathcal{K}_{\Gamma_{\child}}$
are initially grouped in $\mathcal{K}_\Gamma^0$.
It remains to group subtrees according to the connectivity
learned from $\mathcal{K}_\Gamma^{\ext}$.

For each group $K^{\ext}\in{\cal K}^{\ext}_{\Gamma}$, 
we merge all the groups $K^{0}\in {\cal K}^{0}_{\Gamma}$ 
such that $V^{\tmn}(K^{0})$ intersects $V^{\tmn}(K^{\ext})$. 
By \cref{ob:ExtSteinerSubtree},
$V^{\tmn}(\tau)$ is either contained in or disjoint from $V^{\tmn}(\tau^{\ext})$, for $\tau\in K^0 \in {\cal K}^0_\Gamma$ and $\tau^{\ext}\in K^{\ext}\in {\cal K}^{\ext}_\Gamma$.
Thus we can detect whether they intersect by
calling $v\gets \PickTerminal(\tau)$ and $\IsTerminal(v,\tau^{\ext})$.
After all grouping, 
the final partition is ${\cal K}_{\Gamma}$.

\paragraph{Answering Query $\ang{s,t,F}$.} 
We compute the partitions $\{{\cal K}_{\Gamma}\}$ for all affected $\Gamma$ in postorder,
culminating in the root partition ${\cal K}_{\Gamma_{\root}}$.
Then $s$ and $t$ are connected iff there exists a 
$K\in {\cal K}_{\Gamma_{\root}}$ and 
$\tau_s,\tau_t\in K$
such that 
$\IsTerminal(s,\tau_s)=\IsTerminal(t,\tau_t)=\textsf{true}$.

\subsection{The Labeling Scheme: Implementing the Strategy}
\label{sec:implementing vertex label}

Finally, we describe a labeling scheme that supports the interfaces required by the query algorithm. Note that we still fix an $S_{i}\in{\cal S}$ and assume that vertices in $S_{i}$ will never fail, and the labeling scheme we describe is only for this $S_{i}$. At the very end, we will discuss our final labeling scheme and how to find a non-failed $S_{i}$ for a query. 

We restate the interfaces in a formal way as follows.
\begin{itemize}
\item $\ListAffectedComps()$ outputs the identifiers of affected components in a bottom-up order.
\item $\ListSubtrees(\Gamma), \ListSubtrees^{\ext}(\Gamma)$ receive the identifier of an affected component $\Gamma$, and list the profiles of subtrees in ${\cal T}_{\Gamma}$ and ${\cal T}^{\ext}_{\Gamma}$ respectively.
\item $\ListTerminals(\tau,\Gamma)$ receives the identifier of an affected component $\Gamma$ and the profile of a subtree $\tau\in {\cal T}^{\ext}_{\Gamma}$, and either detects $|V^{\tmn}_{\gamma}(\tau)|>f/\phi$ or outputs the profiles of all vertices in $V^{\tmn}_{\gamma}(\tau)$.
\item $\ListNeighbors(\tau,\Gamma)$ receives the identifier of an affected component $\Gamma$ and the profile of a subtree $\tau\in {\cal T}^{\ext}_{\Gamma}$, and either detects $|V^{\tmn}_{\gamma}(\tau)| + |N_{\gamma}(\tau)|>f/\phi$ or outputs the profiles of all vertices in $N_{\gamma}(\tau)$.
\item $\IsTerminal(v,\tau)$ receives the profile of a vertex $v$ and the profile of a subtree $\tau$ (an arbitrary one from \Cref{def:Subtrees}), and outputs whether $v$ is in $V^{\tmn}(\tau)$ or not.
\item $\PickTerminal(\tau)$ receives the profile of a subtree $\tau\in{\cal T}_{\gamma}$ of some affected $\gamma$, and outputs the profile of an arbitrary vertex in $V^{\tmn}(\tau)$.
\item $\EnumFromGiant(\tau_{y},\Gamma)$ receives the profile of a non-giant subtree $\tau_{y}\in{\cal T}_{\gamma^{\ext}}$, and outputs the number of edges in $E(\wtilde{G}^{\qry}_{\Gamma}\setminus F)$ connecting $\bigcup_{\giant\ \tau_{x}\in{\cal T}^{\ext}_{\Gamma}} V^{\tmn}(\tau_{x})$ and $V^{\tmn}_{\gamma}(\tau_{y})$.
\end{itemize}

We make some remarks on the terms \emph{identifiers} and \emph{profiles}. For those objects defined in the preprocessing phase (i.e. vertices in $G$, components/cores\footnote{A component $\Gamma$ and its core $\gamma$ has the same identifier.}, Steiner trees and extended Steiner trees), we assign each of them a distinct $O(\log n)$-bit integer as its identifier, denoted by $\id(\cdot)$. Note that $O(\log n)$ bits suffice because the number of such objects is polynomial. During the query phase, we may further generate \emph{profiles} for the objects (including those defined in the query phase, e.g. subtrees from \Cref{def:Subtrees}). A profile is something that representing this object and generally it is not only a single integer. 

\subsubsection{The Euler Tours of (Extended) Steiner Trees} 

Like what we did for the edge fault connectivity labels, we will exploit the Euler tour to ``linearize'' the (extended) Steiner trees. However, the Euler tours here are slightly different from the definition in \Cref{sec:det-edge-faults}.

\begin{definition}[Euler Tours]
For each (extended) Steiner tree $T$, we define $\Euler(T)$ to be its Euler tour, which is a list that includes all \emph{occurrence} of nodes according to a DFS traversal of $T$, starting from an arbitrary root node. For convenience, we add two \emph{virtual occurences} $\sstart(T)$ and $\eend(T)$ at the front and the end of $\Euler(T)$ as the ``guards''.
See \Cref{fig:EulerInVertexFault} for a small example.
\end{definition}

\begin{figure}[h]
\centering
\begin{tabular}{cc}
\parbox{1.5in}{\includegraphics[scale=0.5]{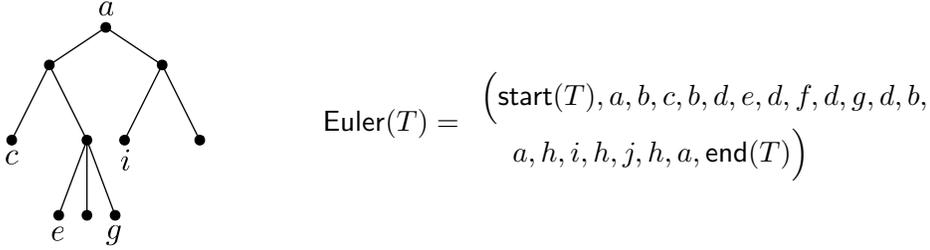}}
&
\parbox{4in}{$\Euler(T) = \begin{array}{l}
{\Big(}\sstart(T),a,b,c,b,d,e,d,f,d,g,d,b,\\
\;\;\;\;a,h,i,h,j,h,a,\eend(T){\Big)}\\
\end{array}$}
\end{tabular}
\caption{\label{fig:EulerInVertexFault} Left: an (extended) Steiner tree $T$, rooted at $a$.\ \ \ \ Right: $\Euler(T)$.}
\end{figure}

For each occurrence $v_{\oc}\in \Euler(T)$, we say a node $v_{\nd}\in V(T)$ \emph{owns} $v_{\oc}$ if $v_{\oc}$ is an occurrence of $v_{\nd}$, and naturally the vertex that \emph{owns} $v_{\oc}$ is the the vertex in $G$ corresponding to the node in $V(T)$ that owns $v_{\oc}$. We let $\pos(v_{\oc})$ denote the position of $v_{\oc}$ in $\Euler(T)$. In particular, $\pos(\sstart(T)) = 0$ and $\pos(\eend(T)) = 2|V(T)|$, since the number of non-virtual occurrences in $\Euler(T)$ is $2|V(T)|-1$. The profile of an occurrence $v_{\oc}\in \Euler(T)$ is defined as $\profile(v_{\oc}) = (\id(v),\id(T),\pos(v_{\oc}))$, where $v$ is the vertex that owns $v_{\oc}$. Exceptionally, $\profile(\sstart(T)) = (\perp, \id(T), \pos(\sstart(T)))$ and $\profile(\eend(T)) = (\perp, \id(T), \pos(\eend(T)))$.

\medskip

\noindent{\underline{Label $\occurrences(v,T)$.}} For each vertex $v\notin S$ and each (extended) Steiner tree $T\in\{T_{\gamma}\mid v\in V(T_{\gamma})\}\cup\{T_{\gamma^{\ext}}\mid v\in V(T_{\gamma^{\ext}})\}$, the label $\occurrences(v,T)$ stores the profiles of occurrences in $\Euler(T)$ owned by $v$, i.e. $\occurrences(v,T) = (\profile(v_{\oc,1}),\profile(v_{\oc,2}),...)$, where $v_{\oc,1},v_{\oc,2},\cdots\in \Euler(T)$ are owned by $v$.

\subsubsection{Profiles of Vertices, Components and Subtrees}

We define the profiles of vertices, components and subtrees, and use them to implement interfaces $\ListAffectedComps(),\ListSubtrees(\Gamma),\ListSubtrees^{\ext}(\Gamma)$ and $\IsTerminal(v,\tau)$.

\medskip

\noindent{\underline{Label $\profile(\Gamma)$.}} The profile of a component $\Gamma\in {\cal C}$, denoted by $\profile(\Gamma)$, is constructed as follows, and we will store $\profile(\Gamma)$ at each vertex $v\in \Gamma$.
\begin{itemize}
\item Store $\id(\Gamma)$, $\id(T_{\gamma})$ and $\id(T_{\gamma^{\ext}})$.
\item For each ancestor $\hat{\Gamma}$ s.t. $\Gamma\prec \hat{\Gamma}$, store $\id(\hat{\Gamma})$ and a one-bit indicator indicating whether $N_{\hat{\gamma}}(\Gamma)$ intersects $S$.
\end{itemize}

\noindent{\underline{Label $\profile(v)$.}} The profile of a vertex $v\in V(G)$, denoted by $\profile(v)$, is constructed as follows, and we will store $\profile(v)$ at vertex $v$. 
\begin{itemize}
\item Store $\id(v)$.
\item For each (extended) Steiner tree $T\in \{T_{\gamma_{v}}\}\cup\{T_{\gamma^{\ext}}\mid v\in \gamma^{\ext}\}$ (i.e. $v$ has a corresponding terminal node in $T$), let $\pi(v,T)$ denote the occurrence $\Euler(T)$ owned by the \underline{terminal node} (corresponding to the vertex) $v$ with the least $\pos(v_{\oc})$\footnote{In fact, we can fix $\pi(v,T)$ to be an arbitrary occurrence owned by the terminal node $v$ (no need to be the one with the least $\pos(v_{\oc})$).}, called the \emph{principal occurrence} of the vertex $v$ on $T$, and we store $\profile(\pi(v,T))$. 
\end{itemize}

\noindent\underline{Implement $\ListAffectedComps()$.} Note that we can already implement the interface $\ListAffectedComps()$ by inspecting the profiles of components storing at vertices $\{s,t\}\cup F$. In fact, we can even support the basic operations in \Cref{ob:BasicOperations}.

\begin{observation}
\label{ob:BasicOperations}
By inspecting the profiles of components storing at vertices $\{s,t\}\cup F$, we can suport the following.
\begin{itemize}
\item List $\id(\Gamma)$ of all affected components $\Gamma$.
\item Given $\id(\Gamma)$ of an affected component $\Gamma$, get $\id(T_{\gamma})$ and $\id(T_{\gamma^{\ext}})$. 
\item Given $\id(T)$ of an (extended) Steiner tree of some (unknown) affected component $\Gamma$, decide whether it is a Steiner tree or an extended Steiner tree and get $\id(\Gamma)$.
\item Given $\id(\Gamma_{1}),\id(\Gamma_{2})$ of two affected components $\Gamma_{1},\Gamma_{2}$, distinguish among (1) $\Gamma_{1}\preceq \Gamma_{2}$, (2) $\Gamma_{2}\preceq\Gamma_{1}$, and (3) no ancestral relation between $\Gamma_{1}$ and $\Gamma_{2}$.
\item Given $\id(\Gamma'),\id(\Gamma)$ of two affected components $\Gamma',\Gamma$ s.t. $\Gamma'\prec\Gamma$, decide if $\gamma'\subseteq \gamma^{\ext}$.
\end{itemize}
\end{observation}

\noindent\underline{Profiles of Subtrees.} For an (extended) Steiner tree $T\in\{T_{\gamma},T_{\gamma^{\ext}}\mid \text{$\gamma$ is affected}\}$, we will represent a subtree $\tau\in {\cal T}$ (${\cal T}$ is the set of subtrees of $T$ w.r.t. $F$ from \Cref{def:Subtrees}) by a collection of \emph{intervals} on $\Euler(T)$.

\begin{definition}[Intervals]
For each (extended) Steiner tree $T\in\{T_{\gamma},T_{\gamma^{\ext}}\mid \text{$\gamma$ is affected}\}$, we break $\Euler(T)$ into \emph{intervals} by removing $\sstart(T)$, $\eend(T)$ and all occurrences owned by failed vertices, and let ${\cal I}_{T}$ denote these intervals. For each interval $I$, we use $\ell_{\oc,I}$ and $r_{\oc,I}$ to denote the left and right \emph{outer endpoints} of $I$. To be precise, for an interval including occurrences from position $a$ to $b$, the outer endpoints of $I$ are the occurrences at position $a-1$ and $b+1$. The profile of an interval $I\in{\cal I}_{T}$ is $\profile(I) = (\profile(\ell_{\oc,I}),\profile(r_{\oc,I}))$. Furthermore, we define $V^{\tmn}(I)$ to be the set of vertices with its principal occurrences $\pi(v,T)$ falling in $I$.
\end{definition}

\begin{observation}
\label{ob:IntervalsAndSubtrees}
There is an assignment that assigns a subset of interval ${\cal I}_{\tau}\subseteq {\cal I}_{T}$ to each subtree $\tau\in{\cal T}$, satisfying the following.
\begin{itemize}
\item For each subtree $\tau\in{\cal T}$, $V^{\tmn}(\tau) = \bigcup_{I\in{\cal I}_{\tau}}V^{\tmn}(I)$.
\item $\{{\cal I}_{\tau}\mid\tau\in{\cal T}\}$ forms a partition of ${\cal I}_{T}$.
\end{itemize}
\end{observation}

We state the relations between intervals ${\cal I}_{T}$ and subtrees ${\cal T}$ in \Cref{ob:IntervalsAndSubtrees}. These are well-known facts from the nature of DFS traversals. Furthermore, by inspecting only $\profile(v)$ for vertices $v\in F$, we can obtain $\profile(I)$ for all $I\in {\cal I}_{T}$ explicitly. Furthermore, it is an easy exercise to compute such a partition $\{{\cal I}_{\tau}\mid \tau\in {\cal T}\}$ of ${\cal I}_{T}$ with properties in \Cref{ob:IntervalsAndSubtrees} (roughly speaking, we can just simulate the DFS traversal using a stack). %

Finally, for each subtree $\tau\in{\cal T}$, we define its profile to be $\profile(\tau) = \{\profile(I)\mid I\in{\cal I}_{\tau}\}$, i.e. we collect $\profile(I)$ of all $I\in{\cal I}_{\tau}$ into $\profile(\tau)$.

\medskip

\noindent\underline{Implement $\ListSubtrees(\Gamma)$ and $\ListSubtrees^{\ext}(\Gamma)$.} We take $\ListSubtrees(\Gamma)$ as an example, and we can implement $\ListSubtrees^{\ext}(\Gamma)$ in a similar way. We have already computed the profiles of all subtrees in \Cref{def:Subtrees}. We just scan all subtrees $\tau$ and output $\profile(\tau)$ if $\tau\in{\cal T}_{\Gamma}$. To check whether $\tau\in{\cal T}_{\Gamma}$, recall that ${\cal T}_{\Gamma}$ collect subtrees from Steiner tree $T_{\gamma'}$ for all affected $\Gamma'\preceq \Gamma$. Note that from $\profile(\tau)$, we can obtain $\id(T)$ of the (extended) Steiner tree $T$ s.t. $\tau\in{\cal T}$, so using \Cref{ob:BasicOperations}, we can get the $\id(\Gamma')$ of the affected $\Gamma'$ s.t. $T_{\gamma'} = T$ (or know such $\Gamma'$ does not exist), and check whether $\Gamma'\preceq \Gamma$.

\medskip

\noindent\underline{Implement $\IsTerminal(v,\tau)$.} Note that $V^{\tmn}(\tau) = \bigcup_{I\in{\cal I}_{\tau}}V^{\tmn(I)}$ by \Cref{ob:IntervalsAndSubtrees}, so to check whether $v\in V^{\tmn}(\tau)$, it suffices to check whether $v\in V^{\tmn}(I)$ for each $I\in{\cal I}_{\tau}$. From $\profile(I)$ (stored in $\profile(\tau)$), we can get $\id(T)$ of the (extended) Steiner tree $T$ s.t. $I$ is on $\Euler(T)$, and further get the position of the outer endpoints of $I$, i.e. $\pos(\ell_{\oc,I})$ and $\pos(r_{\oc,I})$. Lastly, we look at $\profile(\pi(v,T))$ (stored in $\profile(v)$) of the principal occurrence of $v$ on $T$, and check if $\pos(\ell_{\oc,I})<\pos(\pi(v,T))<\pos(r_{\oc,I})$.

\subsubsection{Labels on Euler Tours} We introduce some labels related to Euler tours, and then use them to implement the interfaces $\PickTerminal(\tau), \ListTerminals(\tau,\Gamma)$ and $\ListNeighbors(\tau,\Gamma)$. 

\medskip

\noindent{\underline{Label $\SuccTerminal(v_{\oc}, T)$.}} For each component $\Gamma\in{\cal C}$ and each occurrence $v_{\oc}\in \Euler(T_{\gamma})$, we construct a label $\SuccTerminal(v_{\oc},T_{\gamma})$. Let $x\in\gamma$ be the vertex whose principal occurrence $\pi(x,T_{\gamma})$ is to the right of $v_{\oc}$ and has the least $\pos(\pi(x,T_{\gamma}))$. We store $\profile(x)$ in the label $\SuccTerminal(v_{\oc},T_{\gamma})$, and without ambiguity, $\SuccTerminal(v_{\oc},T_{\gamma})$ also refers to this vertex $x$ in our analysis. We store the label $\SuccTerminal(v_{\oc},T_{\gamma})$ at the vertex $v$ owning $v_{\oc}$. Exceptionally, when $v_{\oc}$ is $\sstart(T_{\gamma})$ or $\eend(T_{\gamma})$, we store $\SuccTerminal(v_{\oc},T_{\gamma})$ at all vertices inside $\Gamma$.

\medskip

\noindent{\underline{Implement $\PickTerminal(\tau)$.}} Note that the input guarantees that $\tau\in{\cal T}_{\gamma}$ for some affected $\gamma$. By \Cref{ob:IntervalsAndSubtrees}, we have $V^{\tmn}(\tau) = \bigcup_{I\in{\cal I}_{\tau}} V^{\tmn}(I)$, so it suffices to obtain a vertex of $V^{\tmn}(I)$ of some $I\in{\cal I}_{\tau}$. For a particular $I\in{\cal I}_{\tau}$, if $V^{\tmn}(I)$ is not empty, the vertex $x = \SuccTerminal(\ell_{\oc,I})$ must be inside $V^{\tmn}(I)$ by definition. Hence, we just need to check whether $x\in V^{\tmn}(I)$, or equivalently, whether $\pos(\ell_{\oc,I})<\pos(\pi(x,T_{\gamma}))<\pos(r_{\oc,I})$. Indeed, this is doable because we can obtain $\pos(\ell_{\oc,I}),\pos(r_{\oc,I}),\id(T_{\gamma})$ from $\profile(\tau)$ and obtain $\pos(\pi(x,T_{\gamma}))$ from $\profile(x)$ ($\profile(x)$ is obtained from the label $\SuccTerminal(\ell_{\oc,I}, T_{\gamma})$ by \Cref{ob:AccessSuccTerminal}).

\begin{observation}
Given $\profile(I)$ of some interval $I\in {\cal I}_{\tau}$ of some $\tau\in{\cal T}_{\gamma}$, we can access $\SuccTerminal(\ell_{\oc,I}, T_{\gamma})$.
\label{ob:AccessSuccTerminal}
\end{observation}
\begin{proof}
We can obtain $\profile(\ell_{\oc,I})$ from $\profile(I)$. If $\ell_{\oc,I}$ is owned by some vertex $v$ (i.e. $\ell_{\oc,I}$ is not $\sstart(T_{\gamma})$ or $\eend(T_{\gamma})$), $v$ must be inside $F$, so the label $\SuccTerminal(\ell_{\oc,I}, T_{\gamma})$ stored at $v$ is accessible and we can locate it using $\profile(\ell_{\oc,I})$. If $\ell_{\oc,I}$ is $\sstart(T_{\gamma})$ or $\eend(T_{\gamma})$, $\SuccTerminal(v_{\oc}, T_{\gamma})$ is accessible from an arbitrary vertex $v\in \Gamma\cap (\{s,t\}\cup F)$. Such $v$ exists because $\Gamma$ is affected (${\cal T}_{\gamma}$ is well-defined only when $\Gamma$ is affected).
\end{proof}

\medskip

\noindent{\underline{Label $\InnerTerminals^{\to}_{\gamma}(v_{\oc},T_{\gamma^{\ext}})$.}} For each component $\Gamma\in{\cal C}$ and each occurrence $v_{\oc}\in \Euler(T_{\gamma^{\ext}})$ \underline{not} owned by $S$-vertices, we construct the label $\InnerTerminals^{\to}_{\gamma}(v_{\oc}, T_{\gamma^{\ext}})$ as follows.

Let $V^{\to}_{v_{\oc}}$ collect all vertices $x\in \gamma$ with principal occurrence $\pi(x,T_{\gamma^{\ext}})$ to the right of $v_{\oc}$. We sort vertices $x\in V^{\to}_{v_{\oc}}$ in order of $\pi(x,T_{\gamma^{\ext}})$ from leftmost to rightmost. The label $\InnerTerminals^{\to}_{\gamma}(v_{\oc}, T_{\gamma^{\star}})$ will take the first $f/\phi+1$ vertices in $V^{\to}_{v_{\oc}}$, and for each vertex $v\in \InnerTerminals^{\to}_{\gamma}(v_{\oc}, T_{\gamma^{\star}})$, we store $\profile(v)$ in the label $\InnerTerminals^{\to}_{\gamma}(v_{\oc}, T_{\gamma^{\star}})$. 

We will store the label $\InnerTerminals^{\to}_{\gamma}(v_{\oc}, T_{\gamma^{\ext}})$ at the vertex $v$ owning $v_{\oc}$. Exceptionally, when $v_{\oc}$ is a virtual occurrence (i.e. $\sstart(T_{\gamma^{\ext}})$ or $\eend(T_{\gamma}^{\ext})$), we store $\InnerTerminals^{\to}_{\gamma}(v_{\oc}, T_{\gamma^{\ext}})$ at all vertices $v\in \Gamma$.

\medskip

\noindent{\underline{Implement $\ListTerminals(\tau,\Gamma)$.}} Recall that the input guarantees that $\tau\in{\cal T}^{\ext}_{\Gamma}$. By definition, $V^{\tmn}_{\gamma}(\tau)$ is not empty only when $\tau\in{\cal T}_{\gamma^{\ext}}$. Using $\id(\tau)$ and \Cref{ob:BasicOperations}, we can easily check if $\tau\in {\cal T}_{\gamma^{\ext}}$, so from now we assume $\tau\in{\cal T}_{\gamma^{\ext}}$. For each interval $I\in{\cal I}_{\tau}$, we define $V^{\tmn}_{\gamma}(I) = V^{\tmn}(I)\cap \gamma$ to be the $\gamma$-vertices in $V^{\tmn}(I)$. From \Cref{ob:IntervalsAndSubtrees}, we know $V^{\tmn}_{\gamma}(\tau) = \bigcup_{I\in{\cal I}_{\tau}} V^{\tmn}_{\gamma}(I)$. By \Cref{lemma:ListTerminals}, it suffices to inspect $\InnerTerminals(\ell_{\oc,I},T_{\gamma^{\ext}})$ for all $I\in{\cal I}_{\tau}$. Note that we can access $\InnerTerminals(\ell_{\oc,I}, T_{\gamma^{\ext}})$ by \Cref{ob:AccessInnerTerminals}.

\begin{observation}
\label{ob:AccessInnerTerminals}
Given $\profile(\tau)$ of some $\tau\in{\cal T}_{\gamma^{\ext}}$, we can access $\InnerTerminals^{\to}_{\gamma}(\ell_{\oc,I}, T_{\gamma^{\ext}})$ for all $I\in{\cal I}_{\tau}$.
\end{observation}
\begin{proof}
This is similar to \Cref{ob:AccessSuccTerminal}. For each $I\in{\cal I}_{\tau}$, we can obtain $\profile(\ell_{\oc,I})$ from $\profile(\tau)$. If $\ell_{\oc,I}$ is owned by some vertex $v$, we can access $\InnerTerminals^{\to}_{\gamma}(\ell_{\oc,I}, T_{\gamma^{\ext}})$ from the vertex $v$. Otherwise, $\ell_{\oc,I}$ is $\sstart(T_{\gamma})$ or $\eend(T_{\gamma})$, and we can access $\InnerTerminals^{\to}_{\gamma}(\ell_{\oc,I}, T_{\gamma^{\ext}})$ from the vertex $v$ from an arbitrary vertex $v\in \Gamma\cap (\{s,t\}\cup F)$ because $\Gamma$ is affected (${\cal T}_{\gamma^{\ext}}$ is well-defined only when $\gamma$ is affected).
\end{proof}

\begin{lemma}
\label{lemma:ListTerminals}
Given $\profile(I)$ of an interval $I\in{\cal I}_{\tau}$ of some $\tau\in{\cal T}_{\gamma^{\ext}}$, we can either detect $|V^{\tmn}_{\gamma}(I)|> f/\phi$ or output the profiles of all vertices in $V^{\tmn}_{\gamma}(I)$, if we can access $\InnerTerminals^{\to}_{\gamma}(\ell_{\oc,I}, T_{\gamma^{\ext}})$.
\end{lemma}
\begin{proof}
We consider a candidate $\wtilde{V}^{\tmn}_{\gamma}(I)$ of $V^{\tmn}_{I}(I)$ defined as 
\[
\wtilde{V}^{\tmn}_{\gamma}(I) = \{x\in \InnerTerminals^{\to}_{\gamma}(\ell_{\oc,I}, T_{\gamma^{\ext}})\mid \pos(\ell_{\oc,I},T_{\gamma^{\ext}})<\pos(\pi(x,T_{\gamma^{\ext}}))<\pos(r_{\oc,I}, T_{\gamma^{\ext}}) \}.
\]
By definition, we have $\wtilde{V}^{\tmn}_{\gamma}(I)\subseteq V^{\tmn}_{\gamma}(I)$. Thus, if $|\wtilde{V}^{\tmn}_{\gamma}(I)|>f/\phi$, we detect $|V^{\tmn}_{\gamma}(I)|>f/\phi$. 

From now, assume $|\wtilde{V}^{\tmn}_{\gamma}(I)|\leq f/\phi$, and we claim that $\wtilde{V}^{\tmn}_{\gamma}(I) = V^{\tmn}_{\gamma}(I)$. Assume for contradiction that there exists a vertex $x\in V^{\tmn}_{\gamma}(I)\setminus \wtilde{V}^{\tmn}_{\gamma}(I)$. Note that when constructing $\InnerTerminals^{\to}_{\gamma}(\ell_{\oc,I},T_{\gamma^{\ext}})$, the vertex $x$ is inside $V^{\to}_{\ell_{\oc,I}}$ because  $\pi(x,T_{\gamma^{\ext}})\in I$ is indeed to the right of $\ell_{\oc,I}$. The only scenario in which $x$ is not selected in $\wtilde{V}^{\tmn}_{\gamma}(I)$ is when the size of $\InnerTerminals^{\to}_{\gamma}(\ell,T_{\gamma^{\ext}})$ reaches $f/\phi+1$ and each vertex $x'\in \InnerTerminals^{\to}_{\gamma}(\ell,T_{\gamma^{\ext}})$ has $\pos(\ell_{\oc,I})<\pos(\pi(x',T_{\gamma^{\star}}))\leq \pos(\pi(x,T_{\gamma^{\star}}))$. However, this means these $f/\phi+1$ vertices $x'$ will be picked into $\wtilde{V}^{\tmn}_{\gamma}(I)$, contradicting $|\wtilde{V}^{\tmn}_{\gamma}(I)|\leq f/\phi$.
\end{proof}

\begin{corollary}
Given an interval $I\in{\cal I}_{\tau}$ of some non-giant subtree $\tau\in{\cal T}_{\gamma^{\ext}}$, we have $V^{\tmn}_{\gamma}(I)\subseteq \InnerTerminals^{\to}_{\gamma}(\ell_{\oc,I}, T_{\gamma^{\ext}})$.
\label{coro:InnerTerminals}
\end{corollary}
\begin{proof}
Recall the candidate $\wtilde{V}^{\tmn}_{\gamma}(I)$ in the proof of \Cref{lemma:ListTerminals}. By definition, $\wtilde{V}^{\tmn}_{\gamma}(I)\subseteq \InnerTerminals^{\to}_{\gamma}(\ell_{\oc,I}, T_{\gamma^{\ext}})$. Furthermore, because $\tau$ is non-giant, we have $V^{\tmn}_{\gamma}(I) = \wtilde{V}^{\tmn}_{\gamma}(I)$.
\end{proof}

\medskip

\noindent{\underline{Labels $\NeighborEdge^{\to}_{\gamma}(v_{\oc},T)$ and $\NeighborVertex^{\to}_{\gamma}(v_{\oc},T)$}.} For each component $\Gamma\in{\cal C}$, each (extended) Steiner tree $T\in\{T_{\gamma^{\ext}}\}\cup\{T_{\gamma'}\mid \gamma'\prec \gamma\}$, and each occurrence $v_{\oc}\in \Euler(T)$ \underline{not} owned by $S$-vertices, we construct the label $\NeighborEdge^{\to}_{\gamma}(v_{\oc},T)$ and $\NeighborVertex^{\to}_{\gamma}(v_{\oc},T)$ as follows.

For each undirected edge $\{u,v\}\in E(\wtilde{G}_{\Gamma})$ with type $\eta$, we treat it as two directed edges $(u,v,\eta)$ and $(v,u,\eta)$. Then we let $E^{\to}_{v_{\oc}}$ be a list collecting all \emph{directed} edges $e = (x,y,\eta)\in E(\wtilde{G}_{\Gamma})$ s.t. $y\in \gamma$ and $x$ has a principal occurrence $\pi(x,T)$ to the right of $v_{\oc}$. We sort edges $e=(x,y,\eta)\in E^{\to}_{v_{\oc}}$ in order of $\pi(x,T)$ from leftmost to rightmost. 

In what follows, we construct a list $\boundary^{\to}_{\gamma}(v_{\oc}, T)$ of edges in $E^{\to}_{v_{\oc}}$ by considering edges in $E^{\to}_{v_{\oc}}$ one by one according to the order. Meanwhile, we will maintain another list 
\[
\NeighborVertex^{\to}_{\gamma}(v_{\oc}, T) = \{y\mid (x,y,\eta)\in \boundary^{\to}_{\gamma}(v_{\oc}, T)\}.
\]
Each edge $(x,y,\eta)$ in $\NeighborEdge^{\to}_{\gamma}(v_{\oc}, T)$ is stored in the form $(\pos(\pi(x,T)),\id(y),\id(\eta))$, and each vertex $y$ in $\NeighborVertex^{\to}_{\gamma}(v_{\oc}, T)$ is stored as $\profile(y)$.

Let $e=(x,y,\eta)\in E^{\to}_{v_{\oc}}$ be the current edge, and we add $e$ into $\boundary^{\to}_{\gamma}(v_{\oc}, T)$ if all the following three conditions hold at this moment.
\begin{enumerate}
\item There is no edge $e'=(x',y',\eta')\in\boundary^{\to}_{\gamma}(v_{\oc},T)$ s.t. $y'=y$ and $\eta' = \eta$.
\label{Cond:boundary2}
\item The number of edges $e'=(x',y',\eta')\in\boundary^{\to}_{\gamma}(v_{\oc},T)$ s.t. $y' = y$ is smaller than $h(f+2)+1$.
\label{Cond:boundary3}
\item 
The size of $\NeighborVertex^{\to}_{\gamma}(v_{\oc}, T)$ should be smaller than $h(f+2)\lambda_{\nb} + f/\phi + 1$. %
\label{Cond:boundary4}
\end{enumerate}
Note that conditions \ref{Cond:boundary3} and \ref{Cond:boundary4} guarantee that the number of edges in $\NeighborEdge^{\to}_{\gamma}(v_{\oc},T)$ is at most $(h(f+2)+1)(h(f+2)\lambda_{\nb}+f/\phi+1)$.

We will store the labels $\NeighborEdge^{\to}_{\gamma}(v_{\oc},T)$ and $\NeighborVertex^{\to}_{\gamma}(v_{\oc},T)$ at the vertex $v$ owning $v_{\oc}$. Exceptionally, when $v_{\oc}$ is $\sstart(T)$ or $\eend(T)$, we store $\NeighborEdge^{\to}_{\gamma}(v_{\oc},T)$ and $\NeighborVertex^{\to}_{\gamma}(v_{\oc},T)$ at all vertices $v\in \bar{\Gamma}$, where $\bar{\Gamma} = \Gamma$ if $T = T_{\gamma^{\ext}}$ and $\bar{\Gamma} = \Gamma'$ if $T = T_{\gamma'}$ for some $\gamma'\prec\gamma$.

\medskip

\noindent{\underline{Implement $\ListNeighbors(\tau,\Gamma)$}.} From $\id(\tau)$, we can obtain $\id(T)$ where $T$ is the (extended) Steiner tree owning $\tau$.
For each interval $I\in{\cal I}_{\tau}$, we let 
\[
N_{\gamma}(I) = N^{\wtilde{G}^{\qry}_{\Gamma}}(V^{\tmn}(I))\cap \gamma
\]
be the neighbors of $V^{\tmn}(I)$ in graph $\wtilde{G}^{\qry}_{\Gamma}$ falling in $\gamma$. Furthermore, let
\begin{equation}
N^{\inc}_{\gamma}(I) = \{y\mid \{x,y\}\in E(\wtilde{G}^{\qry}_{\Gamma}),x\in V^{\tmn}(I)\}\cap \gamma
\label{eq:Ninc}
\end{equation}
denote the vertices in $\gamma$ adjacent to some $x\in V^{\tmn}(I)$ in the graph $\wtilde{G}^{\qry}_{\Gamma}$. By definition, $N^{\inc}_{\gamma}(I)$ may include vertices in $V^{\tmn}(I)$, and $N_{\gamma}(I) = N^{\inc}_{\gamma}(I)\setminus V^{\tmn}(I)$ is exactly the set by excluding $V^{\tmn}(I)$ from $N^{\inc}_{\gamma}(I)$. Therefore,
\[
N_{\gamma}(I)\subseteq N^{\inc}_{\gamma}(I)\subseteq N_{\gamma}(I)\cup V^{\tmn}_{\gamma}(I).
\]

To answer $\ListNeighbors(\tau,\Gamma)$, it suffices to apply \Cref{lemma:ListNeighbors} (with access guarantee from \Cref{ob:AccessNeighborVertex}) on each $I\in{\cal I}_{\tau}$ by the following reasons. By \Cref{ob:IntervalsAndSubtrees}, we have $V^{\tmn}(\tau) = \bigcup_{I\in{\cal I}_{\tau}} V^{\tmn}(I)$. Therefore, we have
\[
N_{\gamma}(\tau)\subseteq 
\bigcup_{I\in{\cal I}_{\tau}} N^{\inc}_{\gamma}(I)\subseteq N_{\gamma}(\tau)\cup V^{\tmn}_{\gamma}(\tau).
\]
If \Cref{lemma:ListNeighbors} detect $|N^{\inc}_{\gamma}(I)|>f/\phi$ for some $I\in{\cal I}_{\tau}$, it means $|V^{\tmn}_{\gamma}(\tau)| + |N_{\gamma}(\tau)|> f/\phi$. Otherwise, \Cref{lemma:ListNeighbors} outputs (the identifiers of) all vertices in $N^{\inc}_{\gamma}(I)$ for all $I\in{\cal I}_{\tau}$. To get $N_{\gamma}(\tau)$, we just need to remove from $\bigcup_{I\in{\cal I}_{\tau}}N^{\inc}_{\gamma}(I)$ vertices in $V^{\tmn}_{\gamma}(\tau)$. Note that we can check whether a vertex $v\in \bigcup_{I\in{\cal I}_{\tau}}N_{\gamma}(I)$ is in $V^{\tmn}_{\gamma}(\tau)$ using $\IsTerminal(v,\tau)$, or even just get $V^{\tmn}_{\gamma}(\tau)$ explicitly using $\ListTerminals(\tau,\Gamma)$.

\begin{lemma}
\label{lemma:ListNeighbors}
Given $\profile(I)$ of an interval $I\in{\cal I}_{\tau}$ of some $\tau\in{\cal T}^{\ext}_{\Gamma}$, we can either detect $|N^{\inc}_{\gamma}(I)|>f/\phi$ or output the profiles of all vertices in $N^{\inc}_{\gamma}(I)$, if we can access $\NeighborEdge^{\to}_{\gamma}(\ell_{\oc,I}, T)$ and $\NeighborVertex^{\to}_{\gamma}(\ell_{\oc,I}, T)$, where $T$ is the (extended) Steiner tree owning $\tau$.
\end{lemma}
\begin{proof}
We consider a candidate $\wtilde{N}^{\inc}_{\gamma}(I)$ of $N^{\inc}_{\gamma}(I)$ defined as
\begin{align*}
\wtilde{N}^{\inc}_{\gamma}(I) = \{y\mid &(x,y,\eta)\in \boundary^{\to}_{\gamma}(\ell, T)\\
&\text{ s.t. } \pos(\ell_{\oc,I},T)<\pos(\pi(x,T))<\pos(r_{\oc,I},T) \text{ and $\eta$ is not affected}\}.
\end{align*}
Note that $\pos(\ell_{\oc,I},T)$ and $\pos(r_{\oc,I},T)$ are stored in $\profile(I)$, $\pos(\pi(x,T))$ is stored in $\NeighborEdge^{\to}_{\gamma}(\ell_{\oc,I},T)$, and we can obtain $\profile(y)$ from $\NeighborVertex^{\to}(\ell_{\oc,I},T)$.

By the construction of $\boundary^{\to}_{\gamma}(\ell, T)$, we have $\wtilde{N}^{\inc}_{\gamma}(I) \subseteq N^{\inc}_{\gamma}(I)$. Thus if $|\wtilde{N}^{\inc}_{\gamma}(I)|>f/\phi$, we detect $|N^{\inc}_{\gamma}(I)|>f/\phi$. From now we assume $|\wtilde{N}^{\inc}_{\gamma}(I)|\leq f/\phi$, and we will show that $N^{\inc}_{\gamma}(I) = \wtilde{N}^{\inc}_{\gamma}(I)$. %

Assume for contradiction that there exists $y$ s.t. $y\in N^{\inc}_{\gamma}(I)$ and $y\notin \wtilde{N}^{\inc}_{\gamma}(I)$. Let $e=\{x,y\}\in E(\wtilde{G}^{\qry}_{\Gamma})$ be an edge with endpoints $x\in V^{\tmn}(I)$ and $y$ and type $\eta$ ($\eta$ must be unaffected since we exclude edges with affected types from $\wtilde{G}^{\qry}_{\Gamma}$). Because $\pi(x,T)$ is to the right of $\ell_{\oc,I}$ and $y\in \gamma$, when constructing $\boundary^{\to}_{\gamma}(\ell_{\oc,I}, T)$, we will add the edge $(x, y, \eta)$ into $E^{\to}_{\ell_{\oc,I}}$. However, we know $(x, y, \eta)\notin\boundary^{\to}_{\gamma}(\ell_{\oc,I}, T)$, because otherwise $y$ will be added in $\wtilde{N}^{\inc}_{\gamma}(I)$. Namely, when we try to add $(x,y,\eta)$ into $\boundary^{\to}_{\gamma}(\ell_{\oc,I}, T)$, at least one of the three conditions is violated. 

Suppose Condition \ref{Cond:boundary2} is violated, i.e. there is an edge $(x',y',\eta')\in \boundary^{\to}_{\gamma}(\ell_{\oc,I},T)$ s.t. $y'=y$ and $\eta'=\eta$. Furthermore, we have $\pos(\ell_{\oc,I})<\pos(\pi(x',T))\leq \pos(\pi(x,T))$ because $(x',y',\eta')$ is added before. This means $y'=y$ will be selected into $\wtilde{N}^{\inc}_{\gamma}(I)$, a contradiction.

Suppose Condition \ref{Cond:boundary3} is violated, i.e. there are at least $h(f+2) + 1$ many edges $(x',y',\eta')\in\boundary^{\to}_{\gamma}(\ell, T)$ with $y'=y$, and all of them satisfy $\pos(\ell_{\oc,I})<\pos(\pi(x',T))\leq \pi(x,T)$. Furthermore, all these edges have different types $\eta'$ by Condition \ref{Cond:boundary2}. Because the number of affected types is at most $h(f+2)$ by \Cref{ob:AffectedNum}, at least one of these edges has an unaffected type $\eta'$. Then $y'=y$ will be selected into $\wtilde{N}^{\inc}_{\gamma}(I)$, a contradiction.

Suppose Condition \ref{Cond:boundary4} is violated, i.e. the number of vertices in $\NeighborVertex^{\to}_{\gamma}(\ell, T)$ reaches $hf\lambda_{\nb} + f/\phi + 1$. We know that at most $hf\lambda_{\nb}$ vertices in $\gamma$ can be incident to affected edges, because (1) the number of affected components is at most $h(f+2)$ by \Cref{ob:AffectedNum}, and (2) for each affected component, the shortcut edges it created will be incident to at most $\lambda_{\nb}$ vertices in $\gamma$ (recall the construction of shortcut edges). Thus, at least $f/\phi + 1$ vertices in $\NeighborVertex^{\to}_{\gamma}(\ell, T)$ will be incident to only unaffected edges. All these vertices will be picked into $\wtilde{N}^{\inc}_{\gamma}(I)$, which means $|\wtilde{N}^{\inc}_{\gamma}(I)|> f/\phi$, a contradiction. %

\end{proof}

\noindent{\underline{Labels $\NeighborEdge^{\larr}_{\gamma}(v_{\oc}, T)$ and $\NeighborVertex^{\larr}_{\gamma}(v_{\oc}, T)$}.} Symmetrically, for each component $\Gamma\in{\cal C}$, each (extended) Steiner tree $T\in\{T_{\gamma^{\ext}}\}\cup\{T_{\gamma'}\mid \gamma'\prec \gamma\}$, and each occurrence $v_{\oc}\in \Euler(T)$ \underline{not} owned by $S$-vertices, we work on the \emph{reversed} $\Euler(T)$, and construct and store labels $\NeighborEdge^{\larr}_{\gamma}(v_{\oc}, T)$ and $\NeighborVertex^{\larr}_{\gamma}(v_{\oc}, T)$ in the same way as $\NeighborEdge^{\to}_{\gamma}(v_{\oc}, T)$ and $\NeighborVertex^{\to}_{\gamma}(v_{\oc}, T)$. \Cref{ob:AccessNeighborVertex} below can be proved using an argument similar to \Cref{ob:AccessInnerTerminals}.

We will use labels $\NeighborEdge^{\larr}_{\gamma}(v_{\oc},T)$ and $\NeighborVertex^{\larr}_{\gamma}(v_{\oc},T)$ in \Cref{sect:CountingEdges}.

\begin{observation}
\label{ob:AccessNeighborVertex}
Given $\id(\Gamma)$ and $\profile(\tau)$ of some $\tau\in{\cal T}^{\ext}_{\Gamma}$, we can access $\NeighborEdge^{\to}_{\gamma}(\ell_{\oc,I}, T)$, $\NeighborVertex^{\to}_{\gamma}(\ell_{\oc,I}, T)$, $\NeighborEdge^{\larr}_{\gamma}(r_{\oc,I}, T)$, and $\NeighborVertex^{\larr}_{\gamma}(r_{\oc,I}, T)$ for all $I\in{\cal I}_{\tau}$, where $T$ is the (extended) Steiner tree owning $\tau$.
\end{observation}

\begin{corollary}
Given an interval $I\in{\cal I}_{\tau}$ of some non-giant subtree $\tau\in{\cal T}^{\ext}_{\Gamma}$, we have $N^{\inc}_{\gamma}(I)\subseteq \NeighborVertex^{\to}_{\gamma}(\ell_{\oc,I}, T)$ and $N^{\inc}_{\gamma}(I)\subseteq \NeighborVertex^{\larr}_{\gamma}(r_{\oc,I}, T)$, where $T$ is the (extended) Steiner tree owning $\tau$.
\label{coro:NeighborVertex}
\end{corollary}
\begin{proof}
This corollary follows the proof of \Cref{lemma:ListNeighbors}. Recall the definition of the candidate $\wtilde{N}^{\inc}_{\gamma}(I)$. When $\tau$ is non-giant, we have $|\wtilde{N}^{\inc}_{\gamma}(I)|\leq f/\phi$ and $N^{\inc}_{\gamma}(I)= \wtilde{N}^{\inc}_{\gamma}(I)$. By definition, $\wtilde{N}^{\inc}_{\gamma}(I)\subseteq \NeighborVertex^{\to}_{\gamma}(\ell_{\oc,I}, T)$, so $N^{\inc}_{\gamma}(I)\subseteq \NeighborVertex^{\to}_{\gamma}(\ell_{\oc,I}, T)$. $N^{\inc}_{\gamma}(I)\subseteq \NeighborVertex^{\larr}_{\gamma}(r_{\oc,I}, T)$ can be proved similarly.
\end{proof}

\subsubsection{Labels for Implementing $\EnumFromGiant(\tau_{y},\Gamma)$} 
\label{sect:CountingEdges}
Before introducing the labels, we first discuss the high level idea. For an undirected (multi-)graph $H$ and two vertex set $X,Y\subseteq V(H)$ ($X,Y$ may intersect), we use $\delta_{H}(X,Y)$ to denote the number of edges with one endpoint in $X$ and the other one in $Y$, but an edge with both endpoints in $X\cap Y$ will be counted twice. Namely, $\delta_{H}(X,Y) = \sum_{x\in X, y\in Y}\delta_{H}(x,y)$. 

Our goal is to compute $\delta_{\wtilde{G}^{\qry}_{\Gamma}\setminus F}(\bigcup_{\giant\ \tau_{x}\in{\cal T}^{\ext}_{\Gamma}} V^{\tmn}(\tau_{x}), V^{\tmn}_{\gamma}(\tau_{y}))$ for a non-giant subtree $\tau_{y}\in{\cal T}_{\gamma^{\ext}}$. We first rewrite it to get rid of the affected shortcut edges. Precisely, we have
\begin{align*}
\delta_{\wtilde{G}^{\qry}_{\Gamma}\setminus F}(\bigcup_{\giant\ \tau_{x}\in{\cal T}^{\ext}_{\Gamma}} V^{\tmn}(\tau_{x}), V^{\tmn}_{\gamma}(\tau_{y})) &= \sum_{\giant\ \tau_{x}\in{\cal T}^{\ext}_{\Gamma}}\delta_{\wtilde{G}_{\Gamma}}(V^{\tmn}(\tau_{x}), V^{\tmn}_{\gamma}(\tau_{y})) \\
&- \sum_{\giant\ \tau_{x}\in{\cal T}^{\ext}_{\Gamma}}\delta^{\aff}_{\wtilde{G}_{\Gamma}}(V^{\tmn}(\tau_{x}), V^{\tmn}_{\gamma}(\tau_{y})),
\end{align*}
where $\delta^{\aff}_{\wtilde{G}_{\Gamma}}(V^{\tmn}(\tau_{x}), V^{\tmn}_{\gamma}(\tau_{y}))$ is the number of affected edges in $\wtilde{\Gamma}$ connecting $V^{\tmn}(\tau_{x})$ and $V^{\tmn}_{\gamma}(\tau_{y})$. This equation holds
because each $\tau\in{\cal T}^{\ext}_{\Gamma}$ has $V^{\tmn}(\tau_{x})\subseteq Q^{\ext}_{\Gamma}\setminus F$ and by definition $\wtilde{G}^{\qry}_{\Gamma} = \wtilde{G}_{\Gamma}[Q^{\ext}_{\Gamma}]\setminus \{e\in E(\wtilde{G}_{\Gamma})\mid \text{$e$ is affected}\}$. For each term $\delta^{\aff}_{\wtilde{G}_{\Gamma}}(V^{\tmn}(\tau_{x}), V^{\tmn}_{\gamma}(\tau_{y}))$, its value is given by \Cref{lemma:CountArtificialEdge}. 

Therefore, it remains to compute $\sum_{\giant\ \tau_{x}\in{\cal T}^{\ext}_{\Gamma}}\delta_{\wtilde{G}_{\Gamma}}(V^{\tmn}(\tau_{x}), V^{\tmn}_{\gamma}(\tau_{y}))$. We further have
\begin{align*}
\sum_{\giant\ \tau_{x}\in{\cal T}^{\ext}_{\Gamma}} \delta_{\wtilde{G}_{\Gamma}}(V^{\tmn}(\tau_{x}), V^{\tmn}_{\gamma}(\tau_{y})) = &\sum_{x\in Q^{\ext}_{\Gamma}}\delta_{\wtilde{G}_{\Gamma}}(x, V^{\tmn}_{\gamma}(\tau_{y})) - \sum_{x\in Q^{\ext}_{\Gamma}\cap F}\delta_{\wtilde{G}_{\Gamma}}(x, V^{\tmn}_{\gamma}(\tau_{y}))\\
- &\sum_{\substack{\text{non-giant}\\\tau_{x}\in {\cal T}^{\ext}_{\Gamma}}}\delta_{\wtilde{G}_{\Gamma}}(V^{\tmn}(\tau_{x}), V^{\tmn}_{\gamma}(\tau_{y})),
\end{align*}
because $\{V^{\tmn}(\tau)\mid \tau\in{\cal T}^{\ext}_{\Gamma}\}$ partitions $Q^{\ext}_{\Gamma}\setminus F$. The second equation basically says that, to compute the number of edges from giant subtrees to $V^{\tmn}_{\gamma}(\tau_{y})$, we can first compute the number of edges from all vertices in $Q^{\ext}_{\Gamma}$ to $V^{\tmn}_{\gamma}(\tau_{y})$, and then subtract those edges starting from failed vertices and non-giant subtrees.
We will compute the three terms on the right hand side using \Cref{lemma:Term1,lemma:Term2,lemma:Term3} respectively.

\medskip

\noindent{\underline{Label $\ArtificialE(\gamma', \wtilde{G}_{\Gamma})$.}} For each component $\Gamma$ and each strict descendant $\Gamma'\prec\Gamma$, we construct a label $\ArtificialE(\gamma',\wtilde{G}_{\Gamma})$ which, for each artificial edges $\{x,y\}$ with type $\gamma'$ in $\wtilde{G}_{\Gamma}$, stores a tuple $(\profile(x), \profile(y), \id(\Gamma'))$. We store $\ArtificialE(\gamma', \wtilde{G}_{\Gamma})$ at each vertex $v\in \Gamma'$.

\begin{lemma}
\label{lemma:CountArtificialEdge}
Given $\id(\Gamma)$ and $\profile(\tau_{x}), \profile(\tau_{y})$ of $\tau_{x}\in{\cal T}^{\ext}_{\Gamma}$ and non-giant $\tau_{y}\in{\cal T}_{\gamma^{\ext}}$, we can access $\ArtificialE(\gamma',\wtilde{G}_{\Gamma})$ for each affected $\gamma'\prec \gamma$, and then compute $\delta^{\aff}_{\wtilde{G}_{\Gamma}}(V^{\tmn}(\tau_{x}), V^{\tmn}_{\gamma}(\tau_{y}))$. 
\end{lemma}
\begin{proof}
First, we can get $\id(\Gamma')$ for all affected component $\Gamma'\prec \Gamma$. For each affected $\Gamma'\prec\Gamma$, we then access $\ArtificialE_{\wtilde{G}_{\Gamma}}(\gamma_{\aff})$ from an arbitrary vertex in $\Gamma'\cap (F\cup\{s,t\})$ (such vertex exists because $\Gamma'$ is affected). In other words, we can obtain $(\profile(u),\profile(v))$ for all affected shortcut edges $\{u,v\}$ in $\wtilde{G}_{\Gamma}$.

Therefore, we just need to scan each affected $\wtilde{G}_{\Gamma}$-edge $\{u,v\}$, and decide the membership of each $u,v$ at each $V^{\tmn}(\tau_{x}),V^{\tmn}_{\gamma}(\tau_{y})$. The membership at $V^{\tmn}_{\gamma}(\tau_{y})$ can be easily decided because we can obtain $V^{\tmn}_{\gamma}(\tau_{y})$ explicitly by $\ListTerminals(\tau_{y},\Gamma)$ (note that $\tau_{y}$ is non-giant). To decide whether $u\in V^{\tmn}(\tau_{x})$ (resp. $v\in V^{\tmn}(\tau_{x})$), we just need to invoke $\IsTerminal(u,\tau_{x})$ (resp. $\IsTerminal(v,\tau_{x})$). 
\end{proof}

\medskip

\noindent{\underline{Label $\Degree(x, \wtilde{G}_{\Gamma}[\gamma^{\ext}])$.}} For each component $\Gamma\in{\cal C}$ and each occurrence $v_{\oc}\in \Euler(T_{\gamma^{\ext}})$ \underline{not} owned by $S$-vertices, we construct the following labels.

For each vertex $x\in \InnerTerminals_{\gamma}^{\to}(v_{\oc}, T_{\gamma^{\ext}})$, we let $\Degree(x,\wtilde{G}_{\Gamma}[\gamma^{\ext}])$ denote the degree of vertex $x$ in graph $\wtilde{G}_{\Gamma}[\gamma^{\ext}]$ (i.e. the subgraph of the sparsified shortcut graph $\wtilde{G}_{\Gamma}$ induced by the extended core $\gamma^{\ext}$). We store $\Degree(x,\wtilde{G}_{\Gamma}[\gamma^{\ext}])$ along with the vertex $x$ in the label $\InnerTerminals_{\gamma}^{\to}(v_{\oc}, T_{\gamma^{\ext}})$.

\medskip

\noindent{\underline{Labels $\Enum(\gamma',y,\wtilde{G}_{\Gamma})$ and $\Enum(x,y,\wtilde{G}_{\Gamma})$.}} For each component $\Gamma\in{\cal C}$ and each $\Gamma'\preceq\Gamma$ s.t. $\gamma'$ is not in $\gamma^{\ext}$, we construct the following labels. Note that $|N_{\gamma}(\Gamma')|\leq\lambda_{\nb}$ because $\gamma'$ is not in $\gamma^{\ext}$. 
\begin{itemize}
\item For each vertex $y\in N_{\gamma}(\Gamma')$, we let $\Enum(\gamma',y,\wtilde{G}_{\Gamma})$ be the total number of $\wtilde{G}_{\Gamma}$-edges connecting some vertex $x\in\gamma'$ and the vertex $\gamma$. We store $\Enum(\gamma',y,\wtilde{G}_{\Gamma})$ at each vertex $v\in \Gamma'$.
\item For each vertex $x\in\gamma'$ and vertex $y\in N_{\gamma}(\Gamma')$, let $\Enum(x,y,\wtilde{G}_{\Gamma})$ be the number of $\wtilde{G}_{\Gamma}$-edges connecting $x$ and $y$ (recall that $\wtilde{G}_{\Gamma}$ is a multigraph, so $\edgenum(x,y,\wtilde{G}_{\Gamma})$ may be larger than 1), and store it at vertex $x$.
\end{itemize}

\begin{lemma}
\label{lemma:Term1}
Given $\id(\Gamma)$ and $\profile(\tau_{y})$ of some non-giant subtree $\tau_{y}\in {\cal T}_{\gamma^{\ext}}$, we can compute $\delta_{\wtilde{G}_{\Gamma}}(Q^{\ext}_{\Gamma},V^{\tmn}_{\gamma}(\tau_{y}))$.
\end{lemma}
\begin{proof}
We further decompose the expression to be
\[
\delta_{\wtilde{G}_{\Gamma}}(Q^{\ext}_{\Gamma}, V^{\tmn}_{\gamma}(\tau_{y})) = \sum_{x\in \gamma^{\ext}}\delta_{\wtilde{G}_{\Gamma}}(x,V^{\tmn}_{\gamma}(\tau_{y})) + \sum_{\substack{\affected\ \gamma'\prec\gamma\\\text{s.t. $\gamma'$ is not in $\gamma^{\ext}$}}}\delta_{\wtilde{G}_{\Gamma}}(\gamma',V^{\tmn}_{\gamma}(\tau_{y})),
\]
because $\{\gamma'\mid \gamma'\prec \gamma,\text{ 
$\gamma'$ is affected and $\gamma'$ is not in $\gamma^{\ext}$}\}\cup\{\gamma^{\ext}\}$ forms a partition of $Q^{\ext}_{\Gamma}$.

For the first term on the right hand side, we can rewrite it as
\[
\sum_{x\in\gamma^{\ext}}\delta_{\wtilde{G}_{\Gamma}}(x, V^{\tmn}_{\gamma}(\tau_{y})) = \sum_{x\in\gamma^{\ext}}\delta_{\wtilde{G}_{\Gamma}[\gamma^{\ext}]}(x,V^{\tmn}_{\gamma}(\tau_{y})) = \sum_{y\in V^{\tmn}_{\gamma}(\tau_{y})}\Degree(y, \wtilde{G}_{\Gamma}[\gamma^{\ext}]).
\]
To compute it, we first obtain (the profiles of) all vertices in $V^{\tmn}_{\gamma}(\tau_{y})$ using $\ListTerminals(\tau_{y}, \Gamma)$ (because $\tau_{y}$ is non-giant). By \Cref{coro:InnerTerminals}, we have 
\[V^{\tmn}_{\gamma}(\tau_{y}) = \bigcup_{I\in{\cal I}_{\tau_{y}}} V^{\tmn}_{\gamma}(I)\subseteq \bigcup_{I\in {\cal I}_{\tau_{y}}}\InnerTerminals(\ell_{\oc,I}, T_{\gamma^{\ext}}).
\]
By \Cref{ob:AccessInnerTerminals}, we can access all these $\InnerTerminals(\ell_{\oc,I}, T_{\gamma^{\ext}})$, so we can further access $\Degree(y,\wtilde{G}_{\Gamma}[\gamma^{\ext}])$ for all $y\in V^{\tmn}_{\gamma}(\tau_{y})$.

Regarding the second term on the right hand side, note that for each vertex $y\in \gamma$ s.t. $y\notin N_{\gamma}(\Gamma')$, there is no $\wtilde{G}_{\Gamma}$-edge connecting $\gamma'$ and $y$. Therefore, we can compute
\[
\delta_{\wtilde{G}_{\Gamma}}(\gamma', V^{\tmn}_{\gamma}(\tau)) = \sum_{y\in V^{\tmn}_{\gamma}(\tau)\cap N_{\gamma}(\Gamma')} \edgenum_{\wtilde{G}_{\Gamma}}(\gamma',y),
\]
because we can access $\edgenum_{\wtilde{G}_{\Gamma}}(\gamma',y)$ for each $y\in N_{\gamma}(\Gamma')$ at any vertex in $\Gamma'\cap (F\cup\{s,t\})$ (note that $\Gamma'\cap (F\cup\{s,t\})$ is not empty because $\gamma'$ is affected).
\end{proof}

\medskip

\noindent{\underline{Labels $\incidentedge(x,\wtilde{G}_{\Gamma}[\gamma^{\ext}])$.}} For each component $\Gamma\in{\cal C}$, we construct the following labels. Recall the construction of the sparsified shortcut graph $\wtilde{G}_{\Gamma}$ in \Cref{sect:DivideAndConquer}, where we use $\wtilde{G}^{\sp}_{\gamma^{\ext}}$ to denote the simple graph corresponding to $\wtilde{G}_{\Gamma}[\gamma^{\ext}]$. Furthermore, it is guaranteed that $\wtilde{G}^{\sp}_{\gamma^{\ext}}$ has arboricity $\lambda_{\arbo}$. Namely, we have an orientation of $\wtilde{G}^{\sp}_{\gamma^{\ext}}$-edges s.t. each vertex $x\in \gamma^{\ext} = V(\wtilde{G}^{\sp}_{\gamma^{\ext}})$ has at most $\lambda_{\arbo}$ incident edges oriented outwards. 

For each vertex $x\in \gamma^{\ext}$, we define a label $\incidentedge(x, \wtilde{G}_{\Gamma}[\gamma^{\ext}])$ which, for all its incident edges $(x,y)\in E(\wtilde{G}^{\sp}_{\gamma^{\ext}})$ with orientation $x\to y$, stores $\id(x),\id(y)$ along with the number of $E(\wtilde{G}_{\Gamma}[\gamma^{\ext}])$-edges connecting $x$ and $y$ (i.e. $\delta_{\wtilde{G}_{\Gamma}[\gamma^{\ext}]}(x,y)$). %

For each vertex $x\in \gamma^{\ext}$, we will store $\incidentedge(x,\wtilde{G}_{\Gamma}[\gamma^{\ext}])$ at $x$. Besides, for each occurrence $v_{\oc}\in \Euler(T_{\gamma^{\ext}})$ not owned by $S$-vertices and each vertex $y\in \InnerTerminals^{\to}_{\gamma}(v_{\oc}, T_{\gamma^{\ext}})\subseteq \gamma^{\ext}$, we store $\incidentedge(y,\wtilde{G}_{\Gamma}[\gamma^{\ext}])$ along with the vertex $y$ in the label $\InnerTerminals^{\to}_{\gamma}(v_{\oc}, T_{\gamma^{\ext}})$.

\begin{observation}
\label{ob:IncidentEdge}
Given $\id(x),\id(y)$ of two vertices $x,y\in\gamma^{\ext}$, if we can access $\incidentedge(x,\wtilde{G}_{\Gamma}[\gamma^{\ext}])$ and $\incidentedge(y,\wtilde{G}_{\Gamma}[\gamma^{\ext}])$, we can compute $\delta_{\wtilde{G}_{\Gamma}[\gamma^{\ext}]}(x,y)$.
\end{observation}

\begin{lemma}
\label{lemma:Term2}
Given $\id(\Gamma),\profile(\tau_{y}),\id(x)$ of some non-giant subtree $\tau_{y}\in {\cal T}_{\gamma^{\ext}}$ and some vertex $x\in \gamma^{\ext}\cap F$, we can compute $\delta_{\wtilde{G}_{\Gamma}}(x, V^{\tmn}_{\gamma}(\tau_{y}))$.
\end{lemma}
\begin{proof}
Because $\tau_{y}$ is non-giant, we can obtain the identifiers of all vertices in $V^{\tmn}_{\gamma}(\tau_{y})$ using $\ListTerminals(\tau_{y}, \gamma)$. We consider the following two cases.

Suppose $x\in \gamma^{\ext}\cap F$. Then
\[
\delta_{\wtilde{G}_{\Gamma}}(x,V^{\tmn}_{\gamma}(\tau_{y})) = \delta_{\wtilde{G}_{\Gamma}[\gamma^{\ext}]}(x, V^{\tmn}_{\gamma}(\tau_{y})) = \sum_{y\in V^{\tmn}_{\gamma}(\tau_{y})} \delta_{\wtilde{G}_{\Gamma}[\gamma^{\ext}]}(x, y).
\]
since $x\in \gamma^{\ext}$ and $V^{\tmn}_{\gamma}(\tau_{y}) \subseteq \gamma^{\ext}$, so it suffices to compute the latter. We enumerate (the identifiers of) vertices $y\in V^{\tmn}_{\gamma}(\tau_{y})$, and then compute $\delta_{\wtilde{G}_{\Gamma}[\gamma^{\ext}]}(x,y)$ using \Cref{ob:IncidentEdge}. Note that we can access $\incidentedge(x,\wtilde{G}_{\Gamma}[\gamma^{\ext}])$ because $x\in F$, and access $\incidentedge(y,\wtilde{G}_{\Gamma}[\gamma^{\ext}])$ because $y\in V^{\tmn}_{\gamma}(\tau_{y})\subseteq \bigcup_{I\in{\cal I}_{\tau_{y}}}\InnerTerminals^{\to}_{\gamma}(\ell_{\oc,I}, T_{\gamma^{\ext}})$ and for all $I\in{\cal I}_{\tau_{y}}$, $\InnerTerminals^{\to}_{\gamma}(\ell_{\oc,I}, T_{\gamma^{\ext}})$ are accessible by \Cref{ob:AccessInnerTerminals}.

Suppose $x\in (Q^{\ext}_{\Gamma}\setminus \gamma^{\ext})\cap F$, which implies $\gamma_{x}$ (the core containing $x$) satisfies $\gamma_{x}\prec\gamma$ and $\gamma_{x}$ is not in $\gamma^{\ext}$. Thus, we can access $\Enum(x,y',\wtilde{G}_{\Gamma})$ for all $y'\in N_{\gamma}(\Gamma')$ at vertex $x$. Again, enumerate (the identifiers of) all vertices $y\in V^{\tmn}_{\gamma}(\tau_{y})$. If $y\in N_{\gamma}(\Gamma_{x})$, then $\delta_{\wtilde{G}_{\Gamma}}(x,y) = \edgenum_{\wtilde{G}_{\Gamma}}(x,y)$ by definition, otherwise $\delta_{\wtilde{G}_{\Gamma}}(x,y) = 0$ (for implementation, just compare $\id(y)$ with the $\id(y')$ in each entry $\Enum(x,y,\wtilde{G}_{\Gamma})$). Finally, we compute
$\delta_{\wtilde{G}_{\Gamma}}(x,V^{\tmn}_{\gamma}(\tau)) = \sum_{y\in V^{\tmn}_{\gamma}(\tau)}\delta_{\wtilde{G}_{\Gamma}}(x,y)$.

\end{proof}

\medskip

\noindent{\underline{Labels $\prefixSum^{\inc}_{\gamma}(v_{\oc}, T, y)$ and $\prefixSum^{\exc}_{\gamma}(v_{\oc}, T, y)$.}} For each component $\Gamma\in{\cal C}$, each (extended) Steiner tree $T\in\{T_{\gamma^{\ext}}\}\cup\{T_{\gamma'}\mid \gamma'\prec \gamma\}$, and each occurrence $v_{\oc}\in \Euler(T_{\gamma})$ \underline{not} owned by $S$-vertices, we construct the following labels. Let $\bar{\gamma}$ be the set of vertices with terminal nodes in $T$ (i.e. $\bar{\gamma} = \gamma^{\ext}$ if $T = T_{\gamma^{\ext}}$, and $\bar{\gamma} = \gamma'$ if $T = T_{\gamma'}$ for some $\gamma'\prec\gamma$).

For each vertex $y\in \NeighborVertex^{\to}_{\gamma}(v_{\oc}, T)$, we define label $\prefixSum^{\inc}_{\gamma}(v_{\oc}, T, y)$ to be the total number of edges in $E(\wtilde{G}_{\Gamma})$ that connect a vertex $x\in \bar{\gamma}$ and the vertex $y$, summing over all $x$ whose principal occurrence $\pi(x,T)$ is to the left of $v_{\oc}$ or exactly $v_{\oc}$. We store $\prefixSum^{\inc}_{\gamma}(v_{\oc}, T,y)$ along with the vertex $y$ in the label $\NeighborVertex^{\to}_{\gamma}(v_{\oc}, T)$.

Similarly, for each vertex $y\in \NeighborVertex^{\larr}_{\gamma}(v_{\oc}, T)$, we define $\prefixSum^{\exc}_{\gamma}(v_{\oc}, T, y)$ to be the total number of edges in $E(\wtilde{G}_{\Gamma})$ that connect a vertex $x\in \bar{\gamma}$ and the vertex $y$, summing over all $x$ whose $\pi(x,T)$ is to the left of $v_{\oc}$ (excluding $v_{\oc}$). We store $\prefixSum^{\exc}_{\gamma}(v_{\oc}, T,y)$ along with the vertex $y$ in the label $\NeighborVertex^{\larr}_{\gamma}(v_{\oc}, T)$.

\begin{lemma}
\label{lemma:Term3}
Given $\id(\Gamma),\profile(\tau_{x}),\profile(\tau_{y})$ of two non-giant subtrees $\tau_{x}\in{\cal T}^{\ext}_{\Gamma}$ and $\tau_{y}\in{\cal T}_{\gamma^{\ext}}$ ($\tau_{x}$ and $\tau_{y}$ can be the same subtree), we can compute $\delta_{\wtilde{G}_{\Gamma}}(V^{\tmn}(\tau_{x}), V^{\tmn}_{\gamma}(\tau_{y}))$.
\end{lemma}

\begin{proof}
Because $\tau_{x},\tau_{y}$ are non-giant subtrees, we can obtain the identifiers of vertices in $V^{\tmn}_{\gamma}(\tau_{y})$ using $\ListTerminals(\tau_{y}, \Gamma)$. Because we have 
\[
\delta_{\wtilde{G}_{\Gamma}}(V^{\tmn}(\tau_{x}), V^{\tmn}_{\gamma}(\tau_{y})) = \sum_{y\in V^{\tmn}_{\gamma}(\tau_{y})} \delta_{\wtilde{G}_{\Gamma}}(V^{\tmn}(\tau_{x}), y),
\]
it suffices to compute $\delta_{\wtilde{G}_{\Gamma}}(V^{\tmn}(\tau_{x}),y)$ for each $y\in V^{\tmn}_{\gamma}(\tau_{x})$. %
We can further rewrite
\[
\delta_{\wtilde{G}_{\Gamma}}(V^{\tmn}(\tau_{x}),y) = \sum_{I\in{\cal I}_{\tau_{x}}}\delta_{\wtilde{G}_{\Gamma}}(V^{\tmn}(I),y)
\]
and compute $\delta_{\wtilde{G}_{\Gamma}}(V^{\tmn}(I),y)$ for each $I\in{\cal I}_{\tau_{x}}$. 

In what follows, we focus on compute $\delta_{\wtilde{G}_{\Gamma}}(V^{\tmn}(\tau_{x}),y)$ for a fixed $I\in{\cal I}_{\tau_{x}}$ and a fixed vertex $y\in V^{\tmn}_{\gamma}(\tau_{y})$. We consider two cases. 

\begin{itemize}
\item If $y\in N^{\inc}_{\gamma}(I)$, we will compute 
\[
\delta_{\wtilde{G}_{\Gamma}}(V^{\tmn}(I),y) = \prefixSum^{\exc}_{\gamma}(r_{\oc,I}, T, y) - \prefixSum^{\inc}_{\gamma}(\ell_{\oc,I}, T, y).
\]
The reason is that, by \Cref{coro:NeighborVertex}, $N^{\inc}_{\gamma}(I)\subseteq \NeighborVertex^{\to}_{\gamma}(\ell_{\oc,I},T)$ and $N^{\inc}_{\gamma}(I)\subseteq \NeighborVertex^{\larr}_{\gamma}(r_{\oc,I},T)$, we can access $\prefixSum^{\inc}_{\gamma}(\ell_{\oc,I}, T, y)$ and $\prefixSum^{\exc}_{\gamma}(r_{\oc,I}, T, y)$ because we have $\id(y)$ and we can access $\NeighborVertex^{\to}_{\gamma}(\ell_{\oc,I}, T)$ and $\NeighborVertex^{\larr}_{\gamma}(r_{\oc,I}, T)$ by \Cref{ob:AccessNeighborVertex}.

\item If $y\notin N^{\inc}_{\gamma}(I)$, we claim that all edges in $\wtilde{G}_{\Gamma}$ connecting $y$ and $V^{\tmn}(\tau_{x})$ are \emph{affected}. To see this, assume for contradiction that there is an \emph{unaffected} edge $e\in E(\wtilde{G}_{\Gamma})$ connecting $y$ and $V^{\tmn}(\tau_{x})$. Recall that $\wtilde{G}^{\qry}_{\Gamma} = \wtilde{G}_{\Gamma}[Q^{\ext}_{\Gamma}]\setminus \hat{E}_{\Gamma,\aff}$ (the subgraph of $\wtilde{G}_{\Gamma}$ induced by vertices $Q^{\ext}_{\Gamma}$ excluding all affected edges). Because $V^{\tmn}(I)\subseteq Q^{\ext}_{\Gamma}$ and $y\in Q^{\ext}_{\Gamma}$, this edge $e$ is in $\wtilde{G}^{\qry}_{\Gamma}$, so $y\in N^{\inc}_{\gamma}(I)$ by its definition (see \Cref{eq:Ninc}), a contradiction.

Therefore, we can use a strategy similar to the proof of \Cref{lemma:CountArtificialEdge}. Concretely, we can obtain $(\profile(u),\profile(v))$ for all affected shortcut edges $\{u,v\}$ in $\wtilde{G}_{\Gamma}$, and for each of them, check if it contributes to $\delta_{\wtilde{G}_{\Gamma}}(V^{\tmn}(I),y)$ using $\IsTerminal$.
\end{itemize}

\end{proof}

\subsubsection{Space Analysis} Finally, we analyse the space of our labeling schemes. First, we bound the size of \underline{one} label for all different types of labels.

\begin{itemize}
\item \underline{$\id(\cdot)$.} It takes $O(\log n)$ bits.
\item \underline{$\pos(v_{\oc})$ of occurrences.} It takes $O(\log n)$ bits because the length of each $\Euler(T)$ is polynomial. 
\item \underline{$\profile(v_{\oc})$ of occurrences.} It takes $O(\log n)$ bits because it stores two $\id(\cdot)$ and one $\pos(v_{\oc})$.
\item \underline{$\occurrences(v,T)$.} It takes $O(h\Delta\log n)$ bits because it stores $O(h\Delta)$ occurrence-profiles. Recall that each vertex $v\notin S$ owns at most $O(h\Delta)$ occurrences in $\Euler(T)$ by \Cref{lemma:ExtLowDegree}.
\item \underline{$\profile(\Gamma)$ of components.} It takes $O(h\log n)$ bits because it stores $h$ $\id(\cdot)$ and $h$ one-bit indicators. Note that the number of ancestors of $\Gamma$ is $h$.
\item \underline{$\profile(v)$ of vertices.} It takes $O(h\log n)$ bits because it stores $h$  occurrence-profiles. Note that the number of (extended) Steiner trees $T\in \{T_{\gamma_{v}}\}\cup\{T_{\gamma^{\ext}}\mid v\in \gamma^{\ext}\}$ is $h$.
\item \underline{$\SuccTerminal(v_{\oc},T)$.} It takes $O(\log n)$ bits because it stores one occurrence-profile.
\item \underline{$\InnerTerminals^{\to}_{\gamma}(v_{\oc}, T_{\gamma^{\ext}})$.} It takes $O(hf\log n/\phi)$ bits because it stores $f/\phi + 1$ vertex-profiles.
\item \underline{$\NeighborEdge^{\to}_{\gamma}(v_{\oc},T)$ and $\NeighborEdge^{\larr}_{\gamma}(v_{\oc},T)$.} It takes 
\[
(h(f+2)+1)(h(f+2)\lambda_{\nb} + f/\phi+1)\cdot O(\log n)\text{ bits},
\]
because the number of edges in $\NeighborEdge^{\to}_{\gamma}(v_{\oc},T)$ is at most $(h(f+2)+1)(h(f+2)\lambda_{\nb} + f/\phi + 1)$, and for each edge we store one occurrence-position and two $\id(\cdot)$.
\item \underline{$\NeighborVertex^{\to}_{\gamma}(v_{\oc},T)$ and $\NeighborVertex^{\larr}_{\gamma}(v_{\oc},T)$.} It takes $O((hf\lambda_{\nb} + f/\phi)h\log n)$ bits because the number of vertices in $\NeighborVertex^{\to}_{\gamma}(v_{\oc},T)$ is at most $h(f+2)\lambda_{\nb} + f/\phi+1$, and for each vertex we store its profile.
\item \underline{$\ArtificialE(\gamma',\wtilde{G}_{\Gamma})$.} It takes $O(h^{2}\lambda^{2}_{\nb}\log n)$ bits by the following reasons. Because $\wtilde{G}_{\Gamma}$ is a subgraph of $\hat{G}_{\Gamma}$, the number of $\gamma'$-type edges in $\wtilde{G}_{\Gamma}$ is at most that in $\hat{G}_{\Gamma}$. Recall the construction of shortcut edges. The $\gamma'$-type edges in $\hat{G}_{\Gamma}$ forms a biclique between $\hat{N}(\Gamma')\cap \Gamma$ and $\hat{N}_{\gamma}(\Gamma')$. Combining $|\hat{N}(\Gamma')\cap \Gamma|\leq |\hat{N}(\Gamma')|\leq h\lambda_{\nb}$ and $|\hat{N}_{\gamma}(\Gamma')|\leq \lambda_{\nb}$, we have the number of $\gamma'$-type edges in $\wtilde{G}_{\Gamma}$ is $O(h\lambda^{2}_{\nb})$. Lastly, for each such edge, we store two vertex-profiles and one $\id(\cdot)$.

\item \underline{$\incidentedge(x,\wtilde{G}_{\Gamma}[\gamma^{\ext}])$.} It takes $O(\lambda_{\arbo}\log n)$ bits because there are $\lambda_{\arbo}$ edges and for each edge we store two $\id(\cdot)$ and one polynomially bounded number.

\item \underline{$\Degree(x,\wtilde{G}_{\Gamma}[\gamma^{\ext}])$,  $\Enum(\gamma',y,\wtilde{G}_{\Gamma})$, $\Enum(x,y,\wtilde{G}_{\Gamma})$, $\prefixSum^{\inc}_{\gamma}(v_{\oc}, T, y)$, $\prefixSum^{\exc}_{\gamma}(v_{\oc}, T, y)$.} Each of them takes $O(\log n)$ bits because they are polynomially bounded number.

\end{itemize}

Next, fixing a vertex $v\in V(G)$, we bound the number of labels at $v$ for each label-type.
\begin{itemize}
\item \underline{$\occurrences(v,T)$.} The number is $O(h)$ because the number of (extended) Steiner trees $T$ s.t. $V(T)$ has nodes corresponding to $v\notin S$ is $O(h)$ (recall that this label requires $v\notin S$). 

\item \underline{$\profile(\Gamma)$ of components.} The number is $h$ because the number of components containing $v$ is at most $h$.

\item \underline{$\profile(v)$ of vertices.} The number is one.

\item \underline{$\SuccTerminal(v_{\oc},T)$.} The number is $O(h^{2}\Delta + h)$ because the number of (extended) Steiner trees $T$ s.t. $V(T)$ has nodes corresponding to $v$ is $O(h)$, and each $\Euler(T)$ has at most $O(h\Delta)$ occurrences owned by $v\notin S$ (recall that this label requires $v_{\oc}$ not owned by $S$-vertices). The vertex $v$ will additionally store $O(h)$ such labels with $v_{\oc} = \sstart(T)\text{ or }\eend(T)$, because $v$ is in at most $h$ components, and each component $\Gamma$ corresponds to two (extended) Steiner trees $T_{\gamma}$ and $T_{\gamma^{\ext}}$.

\item \underline{$\InnerTerminals^{\to}_{\gamma}(v_{\oc}, T_{\gamma^{\ext}})$.} The number is $O(h^{2}\Delta + h)$ because the number of extended Steiner trees $T$ s.t. $V(T)$ has nodes corresponding to $v$ is at most $h$, and each $\Euler(T)$ has at most $O(h\Delta)$ occurrences owned by $v\notin S$ (recall that this label requires $v_{\oc}$ not owned by $S$-vertices). The additional term $O(h)$ is due to the case $v_{\oc} = \sstart(T_{\gamma^{\ext}})\text{ or }\eend(T_{\gamma^{\ext}})$.

\item \underline{$\NeighborEdge^{\to}_{\gamma}(v_{\oc},T)$, $\NeighborEdge^{\larr}_{\gamma}(v_{\oc},T)$, $\NeighborVertex^{\to}_{\gamma}(v_{\oc},T)$ and $\NeighborVertex^{\larr}_{\gamma}(v_{\oc},T)$.} For each of them, the number is $O((h^{2}\Delta+h)h)$ by the following reasons. First, the number of (extended) Steiner trees $T$ s.t. $V(T)$ has nodes corresponding to $v$ is $O(h)$. Second, each $\Euler(T)$ has at most $O(h\Delta)$ occurrences owned by $v\notin S$ (recall that this label requires $v_{\oc}$ not owned by $S$-vertices). The additional term $O(h^{2})$ is due to the case $v_{\oc} = \sstart(T_{\gamma^{\ext}})\text{ or }\eend(T_{\gamma^{\ext}})$. Finally, fixing $v_{\oc}$ and $T$, the number of eligible $\gamma$ is $h$.

\item \underline{$\ArtificialE(\gamma',\wtilde{G}_{\Gamma})$.} The number is $O(h^{2})$, because $v$ is in at most $h$ component $\Gamma'$, and each $\Gamma'$ has $h$ ancestors.

\item \underline{$\incidentedge(x,\wtilde{G}_{\Gamma}[\gamma^{\ext}])$.} The number is $O(h + (h^{2}\Delta + h)(f/\phi + 1))$ by the following reasons. The first term $h$ is because we store $\incidentedge(v,\wtilde{G}_{\Gamma}[\gamma^{\ext}])$ at $v$ for each $\gamma^{\ext}\ni v$, and there are $h$ such $\gamma^{\ext}$. The second term $(h^{2}\Delta + h)(f/\phi + 1)$ is because we store a label $\incidentedge(x,\wtilde{G}_{\Gamma}[\gamma^{\ext}])$ along with each vertex in $\InnerTerminals^{\to}_{\gamma}(v_{\oc}, T_{\gamma^{\ext}})$. We have shown that there are $O(h^{2}\Delta + h)$ $\InnerTerminals^{\to}_{\gamma}(v_{\oc}, T_{\gamma^{\ext}})$ stored at $v$, and each of them has $f/\phi + 1$ vertices.

\item \underline{$\Degree(x,\wtilde{G}_{\Gamma}[\gamma^{\ext}])$.} The number is $O((h^{2}\Delta + h)(f/\phi + 1))$ because each $\Degree(x,\wtilde{G}_{\Gamma})$ is stored along with a vertex in $\InnerTerminals^{\to}_{\gamma}(v_{\oc}, T_{\gamma^{\ext}})$.

\item \underline{$\Enum(\gamma',y,\wtilde{G}_{\Gamma})$.} The number is $O(h^{2}\lambda_{\nb})$ by the following reasons. First, there are at most $h$ components $\Gamma'$ containing $v$. Second, each $\Gamma'$ has at most $h$ ancestor $\Gamma$. Finally, fixing $\Gamma'$ and $\Gamma$, the number of eligible $y$ is $|N_{\gamma}(\Gamma')|\leq \lambda_{\nb}$.

\item 
\underline{$\Enum(x,y,\wtilde{G}_{\Gamma})$.} The number is $O(h\lambda_{\nb})$ by the following reasons. First, $v$ is in exactly one $\gamma'$. Second, the number of component $\Gamma$ s.t. $\Gamma'\preceq \Gamma$ is at most $h$. Finally, fixing $\Gamma'$ and $\Gamma$, the number of eligible $y$ is $|N_{\gamma}(\Gamma')|\leq \lambda_{\nb}$.

\item \underline{$\prefixSum^{\inc}_{\gamma}(v_{\oc}, T, y)$ and $\prefixSum^{\exc}_{\gamma}(v_{\oc}, T, y)$.} For each of them, the number is $O((h^{2}\Delta + h)h\cdot (hf\lambda_{\nb} + f/\phi+1))$ by the following reasons. Take $\prefixSum^{\inc}_{\gamma}(v_{\oc}, T, y)$ as an example. Recall that we store this label along with each vertex $y\in \NeighborVertex^{\to}_{\gamma}(v_{\oc},T)$. We have shown that there are $O((h^{2}\Delta + h)h)$ $\NeighborVertex^{\to}_{\gamma}(v_{\oc},T)$ stored at $v$, and each of them has $hf\lambda_{\nb} + f/\phi+1$ vertices.
\end{itemize}

Regarding the total space at an vertex $v\in V(G)$, observe that the bottleneck is the label $\NeighborEdge^{\to}_{\gamma}(v_{\oc},T)$ (also $\NeighborEdge^{\larr}_{\gamma}(v_{\oc},T)$) and the label $\incidentedge(x,\wtilde{G}_{\Gamma}[\gamma^{\ext}])$. The former takes total bits
\[
(h(f+2)+1)(h(f+2)\lambda_{\nb} + f/\phi+1)\cdot O(\log n)\cdot O((h^{2}\Delta + h)h) = O((h^{5}f^{2}\lambda_{\nb}\Delta + h^{4}f^{2}\Delta/\phi)\log n),
\]
and the latter takes total bits
\[
O(\lambda_{\arbo}\log n)\cdot O(h + (h^{2}\Delta + h)(f/\phi+1)) = O(h^{2}f\lambda_{\arbo}\Delta\log n/\phi).
\]
Plugging in $h = O(\log n),\Delta = O(1/\phi), \lambda_{\nb} = O(f\log n)$ and $\lambda_{\arbo} = O(f^2\log^{3}n)$, the total number of bits is $O(f^{3}(\log^{7}n/\phi + \log^{6}n/\phi^{2}))$.

\subsubsection{The Final Labeling Scheme} Recall that the labeling scheme described above is for one $S_{i}\in{\cal S}$. In fact, our final labeling scheme will be made up of $f+1$ separated (sub-)labeling schemes for the $f+1$ groups $S_{i}$ in ${\cal S}$. For each sub-scheme (corresponding to $S_{i}$) and each vertex $v\in V(G)$, we add the index $i$ to the labels at $v$ belonging to this sub-scheme (view these labels as a whole, so the index $i$ will only be added once), so that we can locate the correct sub-scheme if we know (the index $i$) of a valid $S_{i}$ for a query $\langle s,t,F\rangle$. Note that such indices takes $O((f+1)\log f)$ extra bits at $v$ because each index $i$ takes $O(\log f)$ bits and there are $f+1$ sub-schemes.

Furthermore, we store a label $\ccolor(v)$ at each vertex $v$, where $\ccolor(v)$ is the unique index $i$ s.t. $v\in S_{i}$. For a query $\langle s,t,F\rangle$, to find a valid $S_{i}$ (i.e. $S_{i}$ is disjoint from $F$), it suffices to look at $\ccolor(v)$ of all $v\in F$ and pick an index $i$ different from any of such $\ccolor(v)$. This label $\ccolor(v)$ takes extra $O(\log f)$ bits at a vertex $v$.

In summary, the space of our final labeling scheme is $O(f^{4}(\log^{7}n/\phi + \log^{6}/\phi^{2}))$. 
The query time is $\poly(f,\log n)$.

%% file: 4_edge_rand.tex
\section{Randomized Edge Fault Connectivity Labels}\label{sec:edge-faults}

Dory and Parter~\cite{DoryP21} presented two Monte Carlo labeling schemes for $f$ edge faults.  
The first uses $O(f + \log n)$ bits,
which is optimal for $f\leq \log n$,
while the second is an $O(\log^3 n)$-bit sketch
based on $\ell_0$-samplers, 
following  Ahn, Guha, and McGregor~\cite{AhnGM12} and 
Kapron, King, and Mountjoy~\cite{KapronKM13}.
In \cref{sec:simple-edge-fault-labeling} we present a simpler proof of the $O(f + \log n)$-bit sketch, with a slightly
faster construction time $O(m(1 + f/\log n))$, 
rather than $O(m(f+\log n))$~\cite{DoryP21}, 
and in \cref{sec:smaller-edge-fault-labels} 
we combine the two sketches to yield an $O(\log^2 n\log(f/\log^2 n))$-bit sketch, which improves on~\cite{DoryP21} 
whenever $f=n^{o(1)}$.

\subsection{A Simple Labeling Scheme}\label{sec:simple-edge-fault-labeling}

\begin{theorem}[Cf.~Dory and Parter~\cite{DoryP21}]\label{thm:edge-failure-long-labels}
    Fix any undirected graph $G=(V,E)$ and integer $f\geq 1$.
    There are randomized labeling functions $L_V : V \to \{0,1\}^{O(\log n)}$ and $L_E : E \to \{0,1\}^{f + O(\log n)}$ such that
    given any query $\ang{s,t,F}$, $F\subset E$, $|F|\leq f$, with high probability 
    we can determine if $s$ and $t$ are connected in $G-F$, by inspecting the labels
    $L_V(s),L_V(t)$, and $\{L_E(e) \mid e\in F\}$.
    The labeling can be constructed in $O(m(1 + f/\log n))$ time.
\end{theorem}

    Let $T^*$ be any spanning tree of $G$, 
    rooted at an arbitrary vertex $\root(T^*)$,
    and let $T^*_v$ be the set of vertices in the subtree rooted at $v$.
    Let $L_V(v) = (\min_{u\in T_v^*} \DFS(u), \max_{u\in T_v^*} \DFS(u))$ contain the first and last DFS-numbers 
    in the subtree rooted at $v$.  Given $L_V(u),L_V(v)$, we can determine whether $u,v$ 
    have an ancestor/descendant relationship.
    Let $\sk^0 : E \to \{0,1\}^{c\log n + f}$ be a uniformly random labeling of the edges.  This notation is overloaded for vertices and vertex-sets as follows.
    \begin{align*}
        \sk^0_{\hat{E}}(v) &= \bigoplus_{\substack{e\in \hat{E}-E(T^*)\\ \text{s.t. $v\in e$}}} \sk^0(e) &\mbox{bitwise XOR of $\hat{E}-E(T)$ edges incident to $v$}\\
        \sk^0_{\hat{E}}(S) &= \bigoplus_{v\in S} \sk^0_{\hat{E}}(v) & \mbox{for $S\subset V$.}
    \end{align*} 
    \begin{definition}[Edge Fault Tolerant Labels for \cref{thm:edge-failure-long-labels}]
        Fix any edge $e = \{u,v\}\in E$.  
        The label $L_E(e)$ contains
        \begin{itemize}
            \item $L_V(u), L_V(v)$, and a bit indicating 
            whether $e \in E(T^*)$.
            \item Either $\sk^0(e)$, if $e\not\in T^*$, 
                 or $\sk^0_E(T^*_v)$,
                 if $e\in E(T^*)$ with $v$ being the child of $u$ in $T^*$.
        \end{itemize}
    \end{definition}

    Note that since $\sk^0(e)\oplus \sk^0(e)=\mathbf{0}$,
    this last component of $L_E(\{u,v\})$ is the XOR of all 
    $\sk^0$-labels of edges crossing the cut $(T^*_v,V-T^*_v)$.
    As a special case, $\sk^0_E(T^*_{\root(T^*)}) = \sk^0_E(V) = \mathbf{0}$.\footnote{Here $\mathbf{0}$ refers to a zero-vector of the appropriate length.}

\begin{observation}[Homomorphism from Sets to Sketches]\label{obs:symdiff}
    We also let $\oplus$ be the symmetric difference of sets, i.e., $A\oplus B = (A-B)\cup (B-A)$. 
    If $A,B\subset V$, 
    $\sk^0(A\oplus B) = \sk^0(A)\oplus \sk^0(B)$.
\end{observation}

\begin{proof}[Proof of \cref{thm:edge-failure-long-labels}]
    To answer a query $\ang{s,t,F}$ we first 
    identify those tree edges $F\cap E(T^*) = \{e_1,\ldots,e_{f_0}\}$,
    and let $T^*_0,\ldots,T^*_{f_0}$ 
    be the connected components of $T^*-F$.
    We then compute $\sk^0_E(T^*_i)$ for all $i\in [0,f_0]$ as follows.
    Suppose the deleted tree edges incident to $T^*_i$ 
    are $F_i = \{\{u_1,v_1\},\ldots,\{u_t,v_t\}\}$,
    with $v_j$ the child of $u_j$.
    We claim that $\sk^0_E(T^*_i)$ is
    \[
    \sk^0_E(T^*_i) = \bigoplus_{j\in [t]} \sk^0_E(T^*_{v_j}),
    \]
    which can be calculated from the labels of $F$.
    If $T^*_i$ is rooted at $v_1$, then
    \[
    T^*_i = T^*_{v_1} - (T^*_{v_2}\cup \cdots \cup T^*_{v_t}) = \bigoplus_{j\in[t]} T^*_{v_j},
    \]
    and if $T^*_i$ is rooted at $\root(T^*)$, then 
    $T^*_i = \overline{T^*_{v_1}\cup \cdots \cup T^*_{v_t}} = \overline{\bigoplus_{j\in [t]} T^*_{v_j}}$.
    Correctness follows from \cref{obs:symdiff}, and the 
    fact that $\sk^0(S) = \sk^0(S)\oplus \mathbf{0} = \sk^0(S)\oplus \sk^0(V) = \sk^0(\overline{S})$, for any $S\subset V$.  See \cref{fig:XOR-trees}.

    \begin{figure}[h]
        \centering\includegraphics[scale=0.4]{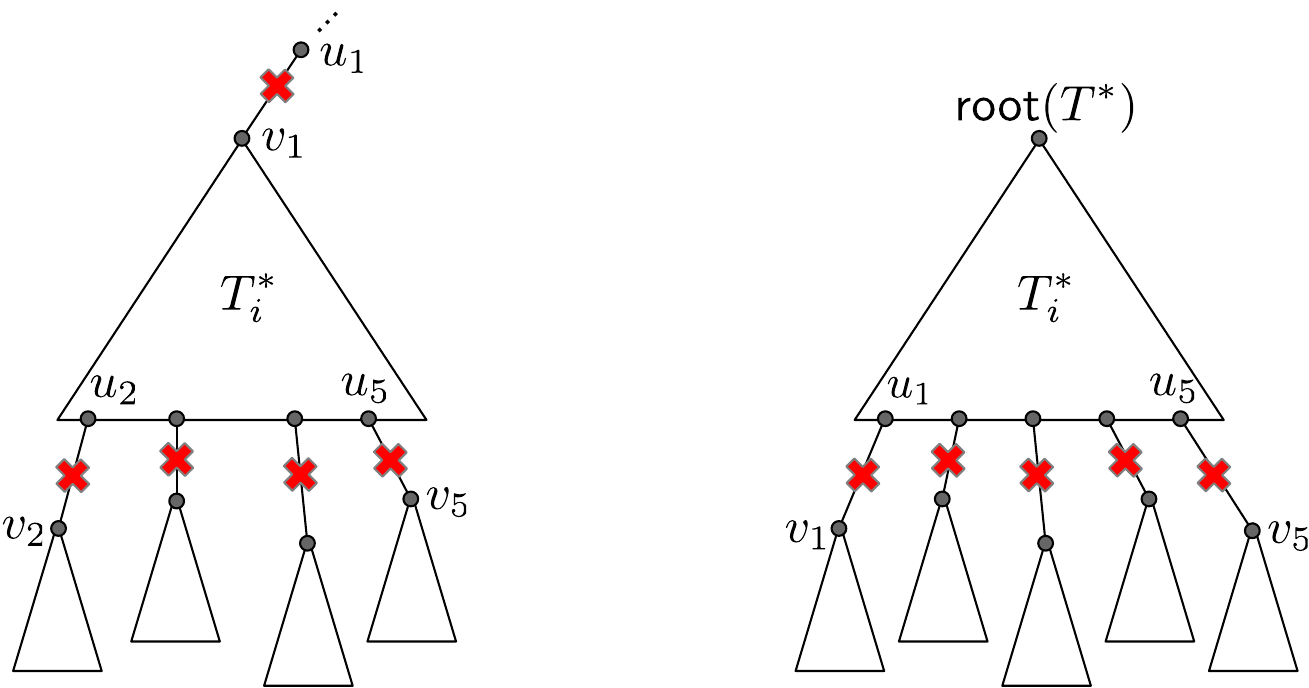}
        \caption{\label{fig:XOR-trees}}
    \end{figure}
    
    We can then delete the contribution of edges from $F-E(T^*)$
    by setting
    \[
    \sk^0_{E-F}(T^*_i) = \sk^0_E(T^*_i) \oplus \bigoplus_{\substack{\{u,v\}\in F - E(T^*)\\
    u\in T^*_i, v\not\in T^*_i}} \sk^0(\{u,v\}).
    \]
    Consider a set $S\subset V$, which is the union of some strict 
    subset of $\{T^*_0,\ldots,T^*_{f_0}\}$.
    Then $\sk^0_{E-F}(S) = \bigoplus_{T^*_i \subseteq S} \sk^0_{E-F}(T^*_i)$ is the XOR of the $\sk^0$-labels of edges crossing the cut $(S,\overline{S})$.
    Thus, if $S$ is the union of some connected components in 
    $G-F$, then $\sk^0_{E-F}(S) = \mathbf{0}$.  The converse is
    true with high probability, since the expected number of 
    \emph{false positive} zeros (over all non-trivial partitions of 
    $\{T^*_0,\ldots,T^*_{f_0}\}$) 
    is $(2^{f_0}-1)2^{-(c\log n + f)} < n^{-c}$.

    Using Gaussian elimination, 
    we find a subset $I\subset [f_0+1]$ for which 
    $S = \bigcup_{i\in I} T^*_i$ and $\sk_{E-F}(S)=\mathbf{0}$,
    then recursively look for more such subsets in $I$ and $[f_0+1]-I$.
    The leaves of this recursion tree enumerate all connected components of $G-F$, assuming no false positives.  
    We can then answer connectivity queries w.r.t.~$G-F$
    in $O(\min\{\frac{\log\log n}{\log\log\log n}, \frac{\log f}{\log\log n}\})$ time using predecessor search~\cite{PatrascuT06,PatrascuT14} over
    the set of $\DFS$-numbers of endpoints of edges in $F\cap E(T^*)$.

    Dory and Parter's~\cite{DoryP21} preprocessing algorithm
    takes linear time $O(m)$ to generate each bit of the labeling, or $O(m(f+\log n))$ time in total.  Assuming a machine with $(\log n)$-bit words, the random edge 
    labels $\{\sk^0(e) \mid e\in E\}$ can be generated in $O(m(1+f/\log n))$ time.  
    It takes $O(m(1+f/\log n))$ to form $\{\sk^0(v) \mid v\in V\}$,
    then another $O(n(1+f/\log n))$ to generate 
    $\{\sk^0(T_v^*) \mid v\in V\}$ with a postorder 
    traversal of $T^*$.
\end{proof}

\subsection{A Smaller Labeling Scheme}\label{sec:smaller-edge-fault-labels}

\begin{theorem}\label{thm:rand-edge-labeling}
    Fix any undirected graph $G=(V,E)$ and 
    integer $f\geq 2\log^2 n$.
    There are randomized labeling functions $L_V : V \to \{0,1\}^{O(\log n)}$ and $L_E : E \to \{0,1\}^{O(\log^2 n\log(f/\log^2 n))}$ such that
    given any query $\ang{s,t,F}$, $F\subset E$, $|F|\leq f$, with high probability 
    one can determine
    whether $s$ and $t$ are connected in 
    $G-F$ by inspecting only 
    $L_V(s),L_V(t),\{L_E(e) \mid e\in F\}$.
\end{theorem}

The vertex labeling function $L_V$ of \cref{thm:rand-edge-labeling}
is the same as \cref{thm:edge-failure-long-labels}; 
only the edge-labels will be different.
As in \cite{DoryP21}, the edge-label sketches
contain the names of a set of edges XORed together.  
We need to be able to decide (w.h.p.) 
when this set has cardinality zero, one, or greater than one.
\cref{lem:singleton} improves the seed-length of
the singleton-detection schemes of
Ghaffari and Parter~\cite{GhaffariP16} and 
Gibb, Kapron, King, and Thorn~\cite{GibbKKT15}
from $O(\log^2 n)$ bits to $O(\log n)$ bits.
See \cref{app:singleton} for proof.

\begin{lemma}\label{lem:singleton}
There are functions 
$\uid : \{0,1\}^{O(\log n)} \times E \to \{0,1\}^{O(\log n)}$ and 
$\singleton : \{0,1\}^{O(\log n)}\times \{0,1\}^{O(\log n)} 
\to \{0,1\}^{O(\log n)} \cup\{\perp\}$ with the following properties.
\begin{itemize}
    \item Given $\uid_s(e)$ where $e=\{u,v\}$ and any seed $s\in\{0,1\}^{O(\log n)}$, we can recover $L_V(u),L_V(v)$ with probability 1.
    
    \item For any single-edge set $E' = \{e^*\} \subset E$
    and any seed $s\in\{0,1\}^{O(\log n)}$, 
    $\singleton_s(\bigoplus_{e\in E'} \uid_s(e)) = \uid_s(e^*)$ with probability 1.
    \item When $|E'|>1$, 
    $\Pr[\singleton_s(\bigoplus_{e\in E'} \uid_s(e))= \,\perp] = 1 - 1/\poly(n)$.  With probability $1/\poly(n)$, it may return a \emph{false positive} $\uid(e)$, where $e$ may or may not be in $E'$. (These probabilities are over the choice of the random seed $s$.)
\end{itemize}
\end{lemma}

For brevity the subscript $s$ is always omitted.

We construct an $\ell_0$-sampling sketch as in~\cite{KapronKM13,AhnGM12}
as follows.  Let $B=\Theta(\log(f/\log^2 n))$.
For each $i\in [B]$, $\rank_i : E \to \Z^+$ is a 
random rank assignment such 
that $\Pr(\rank_i(e)=j)=2^{-j}$, independent of other edges.
Define $\sk(e)$ to be a $B\times \log m$ matrix where
\begin{align*}
\sk(e)[i,j] &= \left\{
\begin{array}{ll}
    \uid(e) & \mbox{ if $\rank_i(e)=j$,}\\
    \mathbf{0}     & \mbox{ otherwise.}
\end{array}\right.
\intertext{Overloading the notation to vertices and vertex sets,}
\sk_{\hat E}(v) &=\bigoplus_{\substack{e \in \hat{E} - E(T)\\ \text{s.t.~$v\in e$}}} \sk_{\hat E}(e),\\
\sk_{\hat E}(S) &= \bigoplus_{v\in S} \sk_{\hat E}(v).
\end{align*}
Here $\oplus$ is applied entrywise to the sketch array.

\begin{definition}[Edge Fault Tolerant Labels for \cref{thm:rand-edge-labeling}]
    The label $L_E(e)$, $e=\{u,v\}$ has bit-length $O(\log^2 n\log(f/\log^2 n))$.  It consists of:
\begin{itemize}
    \item The random $O(\log n)$-bit seed $s$.
    
    \item The sketch from \cref{thm:edge-failure-long-labels}, 
    where $\sk^0 : E \to \{0,1\}^{\log^2 n}$ assigns $\log^2 n$-bit labels, independent of $f$.  Specifically, it includes 
    $L_V(u),L_V(v)$, and either $\sk^0(e)$, if $e\not\in T^*$,
    or $\sk^0(T^*(v))$, where $v$ is the child of $u$, if $e\in T^*$.
    \item Either $\sk(e)$, if $e\not\in T^*$, or $\sk_E(T^*_v)$, 
    if $e=\{u,v\}\in T^*$, where $v$ is the child of $u$ in $T^*$.
\end{itemize}
\end{definition}

\cref{thm:rand-edge-labeling} is 
proved in the remainder of this section.

\medskip 

Consider a query $\ang{s,t,F}$.
Removing the faulty tree edges $\{e_1,\ldots,e_{f_0}\} = F\cap E(T^*)$
results in a set of trees $\{T^*_0,\ldots,T^*_{f_0}\}$.  
Suppose the deleted tree edges incident to $T^*_i$ are 
$\{\{u_1,v_1\},\ldots,\{u_t,v_t\}\}$, 
with $v_j$ the child of $u_j$, and let $F_i\subset F$
be the set of non-tree edges with exactly one endpoint in $T^*_i$.
Then
\[
\sk_{E-F}(T^*_i) = \bigoplus_{j\in [t]} \sk_E(T^*_{v_j}) \oplus \bigoplus_{e\in F_i} \sk(e), \quad\qquad\text{and}\qquad\quad
\sk^0_{E-F}(T^*_i) = \bigoplus_{j\in [t]} \sk^0_E(T^*_{v_j}) \oplus \bigoplus_{e\in F_i} \sk^0(e),
\]
which can be computed directly from the labels 
$\{L_E(e) \mid e\in F\}$. 

At this point we run 
$B=O(\log(f/\log^2 n))$ probabilistic \emph{\Boruvka{} steps}.
We begin with the partition 
$\mathcal{P}_0 = \{T^*_0,\ldots,T^*_{f_0}\}$ of $V$ 
and maintain the loop invariant 
that after $i$ \Boruvka{} steps, 
for each part $P\in \mathcal{P}_{i}$, 
we have the sketches $\sk_{E-F}(P)$ and $\sk^0_{E-F}(P)$.
Here $\mathcal{P}_i$ is a coarsening of $\mathcal{P}_{i-1}$. 

In the $(i+1)$th \Boruvka{} step we attempt, 
for each $P\in \mathcal{P}_{i}$,  
to extract from $\sk_{E-F}(P)$ the $\uid(e)$ 
of an edge $e=\{u,u'\}$, with $u\in P$ and $u'\in P'\neq P$, 
then unify $P$ and $P'$ in $\mathcal{P}_{i+1}$.

\begin{lemma}[Cut Sketch~\cite{AhnGM12,KapronKM13}]
For any $i$ and $P\in\mathcal{P}_i$, 
with constant probability there exists a $j$
such that 
$\singleton(\sk_{E-F}(P)[i,j]) = \uid(e)$.
Conditioned on $\singleton$ returning a $\uid(e)$,
$e\in E-F$ is an edge crossing the cut $(P,\overline{P})$,
with probability $1-1/\poly(n)$.
\end{lemma}

\begin{proof}
    Suppose the number of edges in $E-F$
    crossing the cut $(P,\overline{P})$ 
    is in the range $[2^{j-1},2^j)$, 
    then with constant probability there is exactly 
    one such $e$ with $\rank_i(e) = j$, 
    in which case $\singleton(\sk_{E-F}(P)[i,j]) = \uid(e)$.
    By \cref{lem:singleton}, the probability that $\singleton(\sk_{E-F}(P)[i,j'])$ returns
    a false positive, for any $j'\in [\log m]$, 
    is $\log m / \poly(n) = 1/\poly(n)$.
\end{proof}

Let $X_P\in\{0,1\}$ be an indicator for the event 
that $\sk_{E-F}(P)$ reports
a valid edge in the $(i+1)$th 
\Boruvka{} step. 
If the $\{X_P\}_{P\in \mathcal{P}_i}$ were independent
then Chernoff-Hoeffding bounds would imply
that $\sum_P X_P$ is concentrated around its expectation,
meaning the number of non-isolated parts in the partitions 
would drop by a constant factor, w.h.p., 
so long as there are $\Omega(\log n)$ non-isolated parts.
However, the $\{X_P\}$ are \emph{not independent}, 
so we require a more careful analysis.

\begin{lemma}\label{lem:Boruvka-tail-bound}
Let $\mathcal{P}^*_i\subset \mathcal{P}_i$ be the parts
that are not already connected components of $G-F$.
With probability $1-\exp(-\Omega(|\mathcal{P}^*_i|))$,
$|\mathcal{P}^*_{i+1}| < 0.94 |\mathcal{P}^*_i|$.
\end{lemma}

Before proving \cref{lem:Boruvka-tail-bound} 
let us briefly explain how queries are handled.  
We use the $\sk$-sketches to implement 
$B = O(\log(f/\log^2 n))$ \Boruvka{} steps.
The success of these steps are independent, as step $i$ only 
uses $\sk_{E-F}(P)[i,\cdot]$, 
which depends only on $\rank_i$.
\cref{lem:Boruvka-tail-bound} guarantees that
the number of non-isolated components drops by 
a constant factor in each step, 
hence after $B$ \Boruvka{} steps the number of 
non-isolated parts is at most $f' = (\log^2 n)/2$ 
with probability $1-\exp(-\Omega(\log^2 n))$.

We declare $P$ \emph{isolated} if $\sk^0_{E-F}(P) = \mathbf{0}$,
then determine the connected components of 
the remaining components using 
Gaussian elimination on
the $\sk^0$-sketches 
$\{\sk^0_{E-F}(P) \mid P\in\mathcal{P}_B, P \mbox{ non-isolated}\}$, 
exactly as in \cref{thm:rand-edge-labeling}.
This last step succeeds with probability 
$1-\exp(-\Omega(\log^2 n))$
as $\sk^0$ assigns edge labels with 
$\log^2 n = f' + (\log^2 n)/2$ bits.

The proof of \cref{lem:Boruvka-tail-bound} uses the following martingale concentration inequality.

\begin{theorem}[See \cite{DubhashiPanconesi09} or \cite{ChungL06}]\label{thm:martingale-concentration}
Let $f=f(\mathbf{Y}_m)$ be some function of independent 
random variables $\mathbf{Y}_m = (Y_1,\ldots,Y_m)$.
Define 
$\Delta_i = \E(f \mid \mathbf{Y}_i) - \E(f \mid \mathbf{Y}_{i-1})$, 
$v_i = \lim\sup_{\mathbf{Y}_{i-1}} \Var(\Delta_i \mid \mathbf{Y}_{i-1})$,
$V=\sum_i v_i$, and $M$ be such that $\forall i. \Delta_i \leq M$. 
Then for any $\lambda > 0$,
\[
\Pr(\E(f)-f \geq \lambda), \Pr(f-\E(f) \geq \lambda) \leq \exp\left(-\frac{\lambda^2}{2(V + \lambda M/3)}\right).
\]
\end{theorem}

\begin{proof}[Proof of \cref{lem:Boruvka-tail-bound}]
We order the edges in $E-F$ 
with endpoints in distinct components 
arbitrarily as $e_1,\ldots,e_m$ and let 
$Y_j = \rank_i(e_j)$.
Define $f(\mathbf{Y}_m) = \sum_{P\in \mathcal{P}_i} X_P$, 
where $X_P$ is an indicator for the event that
\emph{exactly} one edge $e$ incident to $P$ has
$\rank_i(e) \geq \floor{\log\deg(P)}+1$, where 
$\deg(P)$ is the number of $E-F$ 
edges with exactly one endpoint in $P$.
Clearly $f$
is an \emph{underestimate} 
for the number of components
$P$ that isolate a single incident edge in 
$\sk_{E-F}(P)[i,\cdot]$.

For $p=\Pr(\rank_i(e) \geq \floor{\log\deg(P)}+1) = 2^{-\floor{\log\deg(P)}}$, 
we have
$\E(X_P)=\deg(P)p(1-p)^{\deg(P)-1} 
\geq \deg(P)pe^{-\deg(P)p}$.  
In the interval $[1,2)$ this is minimized
when $\deg(P)p \to 2$, so $\E(X_P) \geq 2e^{-2} > 0.27$.

Suppose $e_j$ joins parts $P,P'\in\mathcal{P}_i$.
Note that revealing $Y_j=\rank_i(e_j)$
can only change the conditional probability of 
$X_P$ and $X_{P'}$.  In particular, $\Delta_j \leq M \bydef 2$
and
\begin{align}
\Var(\Delta_j \mid \mathbf{Y}_{j-1}) &=
 \E((X_P - \E(X_P \mid \mathbf{Y}_{j-1}))^2 \mid \mathbf{Y}_{j-1})
 +
 \E((X_{P'} - \E(X_{P'} \mid \mathbf{Y}_{j-1}))^2 \mid \mathbf{Y}_{j-1})\nonumber\\
 &\qquad +\, 
 2\E((X_P - \E(X_P \mid \mathbf{Y}_{j-1}))(X_{P'} - \E(X_{P'} \mid \mathbf{Y}_{j-1})) \mid \mathbf{Y}_{j-1}),\label{eqn:variance}
\end{align}
where the expectations are over choice of $Y_j$.
Suppose that $\mathbf{Y}_{j-1}$ reveals the levels of all but $g$ edges incident to $P$, and among those revealed, $b$ are at level at least $\floor{\log\deg(P)}+1$.

\paragraph{Case $b\geq 2$.} Then it is already known that $X_P=0$.

\paragraph{Case $b=1$.} Then $X_P=1$ iff the remaining $g$ 
edges choose levels at most $\floor{\log\deg(P)}$.
\begin{align*}
    &\E((X_P - \E(X_P \mid \mathbf{Y}_{j-1}))^2 \mid \mathbf{Y}_{j-1})\\ 
    &= p\cdot (0 - (1-p)^{g})^2 
    + (1-p)\cdot ((1-p)^{g-1} - (1-p)^{g})^2 \\
    &\leq p(1-p)^2 + (1-p)p^2 = p(1-p) \tag{maximized at $g=1$.}
\end{align*}

\paragraph{Case $b=0$.} Then $X_P=1$ iff exactly one of the remaining $g$ edges chooses a level at least $\floor{\log\deg(P)}+1$.
\begin{align*}
    &\E((X_P - \E(X_P \mid \mathbf{Y}_{j-1}))^2 \mid \mathbf{Y}_{j-1})\\
    &= p\cdot ((1-p)^{g-1} - gp(1-p)^{g-1})^2
    + (1-p)\cdot ((g-1)p(1-p)^{g-2} - gp(1-p)^{g-1})^2\\
    &\leq p\cdot (1-p)^2 + (1-p)p^2 = p(1-p) \tag{maximized at $g=1$.}
\end{align*}

Let $p = 2^{-\floor{\log\deg(P)}}$ and $q = 2^{-\floor{\log\deg(P')}}$ where $q \leq p$.
From the calculations above, which are maximized at $g=1$,
we can upper bound the last term of \cref{eqn:variance} as
follows.  Note: $Y_j \in [1,\floor{\log\deg(P)}]$, $(\floor{\log\deg(P)},\floor{\log\deg(P')}]$, and $(\floor{\log\deg(P')},\infty)$ 
with probability $1-p,p-q,$ and $q$, respectively.
\begin{align*}
    &\E((X_P - \E(X_P \mid \mathbf{Y}_{j-1}))(X_{P'} - \E(X_{P'} \mid \mathbf{Y}_{j-1})) \mid \mathbf{Y}_{j-1})\\
    &= q(1-p)(1-q) 
    + (p-q)(1-p)q
    + (1-p)pq \\
    &= (1-p)q(1-2q+2p)
\intertext{and therefore}
\Var(\Delta_j \mid \mathbf{Y}_{j-1}) 
&\leq p(1-p) + q(1-q) + 2(1-p)q(1-2q+2p)\\
&< 4p.
\end{align*}
In other words, $\Var(\Delta_j \mid \mathbf{Y}_{j-1}) < 4/2^{\floor{\min\{\log\deg(P),\log\deg(P')\}}} < 8/\min\{\deg(P),\deg(P')\}$, 
and therefore 
$V = \sum_j \Var(\Delta_j \mid \mathbf{Y}_{j-1}) < 8|\mathcal{P}^*_i|$.
By \cref{thm:martingale-concentration},
\begin{align*}
\Pr(f < \E(f) - \lambda) < \exp\left(-\frac{\lambda^2}{2(V + \lambda M/3)}\right) = \exp\left(-\frac{\lambda^2}{16|\mathcal{P}^*_i| + 4\lambda/3}\right).
\end{align*}
Setting $\lambda=\E(f)/2 > 0.13|\mathcal{P}^*_i|$, 
we conclude that
$\Pr(f \geq \E(f)/2 > 0.13|\mathcal{P}^*_i|) > 1 - \exp(-\Omega(|\mathcal{P}^*_i|))$.
This implies the number of \emph{distinct} edges reported 
in the $i$th \Boruvka{} step is at least $0.13|\mathcal{P}^*_i|/2 = 0.065|\mathcal{P}^*_i|$, as each edge can be reported twice, 
by either endpoint.  Hence $|\mathcal{P}^*_{i+1}|\leq 0.935|\mathcal{P}^*_i|$.
\end{proof}

%% file: 5_vertex_rand.tex
\section{Randomized Vertex Fault Connectivity Labels}

We improve the size of labeling scheme under $f$ vertex faults from $\tilde{O}(f^3)$ \cite{ParterPP24} to $\tilde{O}(f^2)$ bits.

\begin{theorem}\label{thm:rand-vertex-labeling}
    Fix any undirected graph $G=(V,E)$ and 
    integer $f\geq 1$.
    There are randomized labeling functions $L_V : V \to \{0,1\}^{O(f^2 \log^6 n)}$ such that
    given any query $\ang{s,t,F}$, $F\subset V$, $|F|\leq f$, with high probability 
    one can determine
    whether $s$ and $t$ are connected in 
    $G-F$ by inspecting only 
    $L_V(s),L_V(t),\{L_V(v) \mid v\in F\}$.
\end{theorem}

Our scheme is a small modification to that of 
Parter, Petruschka, and Pettie~\cite{ParterPP24}. 
Fix a collection of vertex sets $N_{1},N_{2},\dots,N_{k}\subset V$ where $k\le n$.\footnote{In \cite{ParterPP24}, each $N_{j}=N(\Gamma_{j})$ is a neighborhood for some component $\Gamma_{j}$ in the low-degree hierarchy, which is deterministically constructed and fixed. Here, we can think of each $N_{j}$ as some arbitrary vertex set.} We say that a \emph{neighborhood hitter }$S\subset V$ is \emph{good} for a fault set $F\subset V$ of size $|F|\le f$ if
\[
S\cap F=\emptyset\text{ and }S\cap N_{j}\neq\emptyset\text{ for all }N_{j}\text{ where }|N_{j}|\ge cf\log n,
\]
where $c$ is a constant.
A collection of neighborhood hitters
${\cal S} = \{S_{1},\dots,S_{s}\}$ 
is good for $F$ if there exists 
$S_{i} \in \cal{S}$ that is good 
for $F$.
Parter et al.~\cite{ParterPP24} reduces the labeling
problem to constructing a collection of neighborhood 
hitters as follows. 
\begin{lem}
[Section 5.2 of \cite{ParterPP24}]\label{lem:reduction}Suppose we can construct a collection of neighborhood hitters ${\cal S}=\{S_{1},\dots,S_{s}\}$ such that, for each fault set $F$ of size at most $|F|\le f$, ${\cal S}$ is good for $F$ with high probability, then there exists a randomized vertex labeling of size $O(sf^{2}\log^{5}n)$ satisfying the guarantee in the setting of \Cref{thm:rand-vertex-labeling}.
\end{lem}

We construct ${\cal S}=\{S_{1},\dots,S_{s}\}$ as follows. Set $s=O(c\log n)$. For each $i$, sample each vertex into $S_{i}$ with probability $1/f$. 
Observe that $S_{i}$ is good for $F$ with constant probability. 
Indeed, 
\[
\Pr[S_{i}\cap F=\emptyset]=(1-1/f)^{|F|} \ge \Omega(1).
\]
Also, for each $N_{j}$ where $|N_{j}|\ge cf\log n$, 
\[
\Pr[S_{i}\cap N_{j}=\emptyset]=(1-1/f)^{|N_{j}|} \le 1/n^{\Omega(c)}.
\]
Thus, by a union bound, $\Pr[\exists j\text{ s.t. }S_{i}\cap N_{j}=\emptyset] \le n^{-\Omega(c)}$. Hence, for any fixed $F$, ${\cal S}$ is not good for $F$ with probability at most $(1-\Omega(1))^{s} \le n^{-\Omega(c)}$.
By plugging ${\cal S}$ into \Cref{lem:reduction}, we obtain \Cref{thm:rand-vertex-labeling}.\footnote{Parter et al.~\cite{ParterPP24} gave the same construction of ${\cal S}$,
but set $s=f+1$ so that 
the criterion that
there exists $S_i\in {\cal S}$
with 
$S_i\cap F=\emptyset$ holds 
with probability 1, which is necessary for a \emph{deterministic} 
labeling scheme. Our observation is simply that in a Monte Carlo labeling scheme, we only need this property to hold with high probability.}

%% file: 6_lower_bound.tex
\section{Lower Bound for Global Connectivity under Vertex Faults}
\label{sec:lowerbound}

In this section, we show that any vertex labeling scheme that supports \emph{global} connectivity queries under $f$ vertex faults requires $\Omega(n^{1-1/f}/f)$ bits.
This improves on an 
$\Omega(\min\{n/f,4^{f}/f^{3/2}\})$ lower bound of Parter et al.~\cite{ParterPP24},
and gives a negative answer to the open problem by \cite{ParterPP24}, which asks 
for an $\tilde{O}(1)$-size labeling scheme 
for global connectivity queries when $f=O(1)$. In contrast, under edge faults it is easy to answer global connectivity queries with previous schemes~\cite{DoryP21,IzumiEWM23} 
or \Cref{thm:shorter-deterministic-labels,thm:rand-edge-labeling}.

The lower bound is stated below. 
The proof closely follows
\cite[Theorem 9.2]{ParterPP24} 
but with different parameters.

\begin{theorem}\label{thm:lower bound}
Let $L:V\rightarrow\{0,1\}^{b}$ be a $b$-bit vertex labeling scheme such that, given $\{L(v)\mid v\in F\}$ where $|F|\le f$, reports whether $G-F$ is still connected. Then $b =\Omega(n^{1-1/f}/f)$.
\end{theorem}

\begin{proof}
First, construct a ``base'' bipartite graph $G_{0}=(L\cup R,E_{0})$ as follows. Set $L=\{v^{*}\}\cup\{v_{1},\dots,v_{n}\}$ and $R=\{u_{1},\dots,u_{r}\}$ where $r=\left\lceil fn^{1/f}\right\rceil $. Connect $v^{*}$ to all vertices in $R$. For each $i\in[n]$, $F_{i}\subseteq R$ is a neighbor set of $v_{i}$ where $|F_{i}|=f$. We make sure that the $\{F_{i}\}$ are all distinct, i.e., $F_{i}\neq F_{j}$ for all $i\neq j\in[n]$. This is possible because $\binom{r}{f}\ge(r/f)^{f}\ge n$ by the choice of $r$. Finally, we create a family ${\cal G}$ of $2^{n}$ graphs from a fixed graph $G_{0}$ as follows: For each $v_{i}$, we can choose to add new edges into $G_{0}$ so that $v_{i}$ is connected to all vertices in $R$. 

Suppose an unknown graph $G$ is promised to be from ${\cal G}$. Observe that $G-F_{i}$ is disconnected if and only if we did not add new edges incident to $v_{i}$ into $G_{0}$. So, by reading the labels on $F_{i}$, we can check what choice was made for $v_{i}$. Thus, the labels of all vertices in $R$ can determine $G\in{\cal G}$. Thus, 
$2^{rb}\ge|{\cal G}|$, implying that  
$b\ge n/\left\lceil fn^{1/f}\right\rceil $. 
\end{proof}

%% file: conclusion.tex
\section{Conclusion and Open Problems }

In this paper we gave improved constructions of 
\emph{expander hierarchies} (w.r.t.~both vertex and edge cuts)
and developed shorter \emph{labeling schemes} for 
$f$-fault connectivity queries, 
in all four quadrants 
of \{vertex faults, edge faults\} $\times$ \{randomized, deterministic\}.

\medskip 

Our deterministic labeling scheme for edge faults has size $\tilde{O}(\sqrt{f})$, but it is not clear that there 
must be a 
polynomial dependence on $f$.

\begin{question}
Is there a deterministic, 
$f$-edge-fault connectivity labeling scheme 
with $\tilde{O}(1)$-bit labels?
\end{question}

In the case of randomized edge-labeling schemes, 
we improved Dory and Parter~\cite{DoryP21} 
from $O(\log^3 n)$ to $O(\log^2 n\log f)$.
An interesting problem is to prove a non-trivial
lower bound on edge labels, randomized or not.  
There are natural targets around $f=\omega(\log n)$ 
and $f=\poly(n)$.

\begin{question}
    Concerning the $f$-edge fault connectivity problem:
    \begin{itemize}
    \item When $f=\omega(\log n)$, are $\omega(\log n)$-bit edge labels necessary? (See~\cite{DoryP21} and \cref{thm:edge-failure-long-labels}.)
    \item When $f=\poly(n)$, are $\Theta(\log^3 n)$-bit edge labels optimal?  (See $\Omega(\log^3 n)$-bit lower bounds of Nelson and Yu~\cite{NelsonY19} and Yu~\cite{Yu21} for similar problems.)
    \end{itemize}
\end{question}

We actually know \emph{how} 
to improve the randomized 
label length to $O(\log n\log^2(f\log n))$ 
\underline{\emph{if}} the following conjecture holds.

\newcommand{\cutsize}{\mathsf{cutsize}}

\begin{conjecture}\label{conj:distribution-spanning-trees}
Given a graph $G$ and spanning tree $T$, 
define $\cutsize_T(e)$ to be 0 if $e\not\in T$ and 
the number of (non-tree) edges crossing the cut 
defined by $T-e$ otherwise.
For any $c_0>0$, there exists a $c_1>0$,
such that for any graph $G$ and integer $f\gg \log n$,
there is a distribution $\mathcal{T}$ of its 
spanning trees such that for any 
$F=\{e_1,\ldots,e_f\}\subset E$,
\[
\Pr_{T\sim \mathcal{T}}(|\{i\in [f] : \cutsize_T(e_i) > (f\log n)^{c_1}\}| \geq c_1\log n) < n^{-c_0}.
\]
\end{conjecture}

If \cref{conj:distribution-spanning-trees} were true,
we would pick the spanning tree $T\sim \mathcal{T}$ used in 
the labeling scheme of \cref{thm:rand-edge-labeling}.
Since, with probability $1-n^{c_0}$, 
the cut sizes for all but $O(\log n)$ trees
of $T-F$ would be bounded by $f(f\log n)^{c_1}$, we
would only need to store $O(\log(f\log n))$ rows
in the sketch-matrix $\sk$, rather than $O(\log n)$,
in order to implement \Boruvka{} steps.
The remaining $c_1\log n$ trees with large 
cut-sizes would be handled using the $\sk^0$ sketch,
just as we handle the residual trees in 
\cref{thm:rand-edge-labeling}.
Using \emph{spanning} trees~\cite{AbrahamN19} in 
R\"acke's tree distribution $\mathcal{R}$~\cite{Racke08} 
guarantees that for any $e\in E(G)$, 
$\Pr_{T\sim\mathcal{R}}(\cutsize_T(e) > f\poly(\log n)) < 1/f$.
\cref{conj:distribution-spanning-trees} 
can be viewed as asserting that there is a distribution
where these events are sufficiently independent,
for any $F\subset E(G)$ with $|F|=f$,
so that we can get a Chernoff-like tail bound on the event
$|\{i\in [f] : \cutsize_T(e_i) > (f\log n)^{c_1}\}| \geq c_1\log n$.

\medskip 

The state-of-the-art for connectivity labels under vertex faults are now $\tilde{O}(f^{2})$ for randomized schemes, by \Cref{thm:rand-vertex-labeling}, 
and $\tilde{O}(f^{4})$ for deterministic schemes, 
by \Cref{thm:vertex label}, 
while there is a simple lower bound of $\Omega(f+\log n)$~\cite{ParterPP24}. 
We believe the correct exponent is likely $2$, but any non-trivial lower bound would be welcome.

\begin{question}
Is there an $\Omega(f^{1.1})$ 
lower bound for $f$-vertex fault connectivity labels?
What is the optimal exponent?
\end{question}

\Cref{thm:lower bound} strongly separated pairwise connectivity and global connectivity for vertex-labeling schemes under vertex faults. Can we match the lower bound for global connectivity?
\begin{question}
Show a labeling scheme for global connectivity under vertex faults whose size matches the lower bound by \Cref{thm:lower bound}. 
\end{question}

The (edge) expander hierarchy of \cref{thm:edge exp hie} 
improves the expansion of \Patrascu{} and Thorup's~\cite{PatrascuT07} by a $\Theta(\log n)$-factor.  The 
$O(f\log^2 n\log\log n)$ query time of their 
$f$-edge-failure connectivity oracle contains some
$O(\log^2 n)$ terms unrelated to the expander hierarchy parameters,
so it may be worth revisiting the complexity of this 
problem in the deterministic setting.  
See~\cite{DuanP20,GibbKKT15} for smaller and faster 
randomized data structures.

%% file: a_low_deg_tree.tex
\section{Low-Degree Steiner Trees Spanning Tough Sets}
\label{app:low-deg steiner tree}

In this section, we prove that there exists a $O(1/\phi)$-degree Steiner tree spanning any $\phi$-vertex-expanding set.

\steinertree*

In fact, we will show that the statement holds even for \emph{$\phi$-tough sets}, which we define now. For any graph $G$ and vertex set $X\subseteq V(G)$, let $c_{G}(X)$ count the number of connected components $C$ in $G$ containing some vertex of $X$. In particular, if $G$ is connected, then $c_{G}(X)=1$ for every $X$. We say that $X$ is \emph{$\phi$-tough} in $G$ if, for every vertex set $S\subseteq V(G)$, we have 
\[
|S|\ge\phi\cdot c_{G-S}(X)
\]
 whenever $c_{G-S}(X)>1$, i.e., $X$ is not connected in $G-S$. 

We say that $G$ has \emph{toughness $\tau(G)=\phi$} if $V(G)$ is $\phi$-tough. The toughness of graphs is a well-studied measure of the robustness of graphs with many connections to other graph properties (see, e.g., \cite{chvatal1973tough,enomoto1985toughness,win1989connection,bauer2006toughness,gu2021proof}). Most literature considers the toughness of $V(G)$, but here, we will focus on the toughness of an arbitrary vertex subset $X$. 

Observe that any vertex-expanding set is tough with the same parameter to up a constant. 
\begin{fact}
If $X$ is $3\phi$-vertex-expanding in $G$, then $X$ is $\phi$-tough. 
\end{fact}

\begin{proof}
Suppose that $X$ is not $\phi$-tough, i.e. there exists $S$ where $|S|<\phi\cdot c_{G-S}(X)$. Let $c=c_{G-S}(X)$ and $C_{1},\dots,C_{c}$ be different connected components in $G-S$ where $|C_{i}\cap X|>0$ for all $i$. Let $C'$ be the union of other connected components in $G-S$ disjoint from $X$. Let $L=C_{1}\cup\dots\cup C_{\left\lceil c/2\right\rceil }$ and $R=C_{\left\lceil c/2\right\rceil +1}\cup\dots\cup C_{c}\cup C'$. We have that $|X\cap L|\ge\ceil{c/2}\ge c/3,|X\cap R|\ge c-\ceil{c/2}c/3$. So 
\[
|S|<\phi c\le3\phi\min\{|X\cap(L\cup S)|,|X\cap(R\cup S)|\},
\]
meaning that $X$ is not $3\phi$-vertex-expanding. 
\end{proof}
The following theorem stating that there exists a $O(1/\phi)$-degree Steiner tree spanning any $\phi$-tough set immediately implies \Cref{lemma:LowDegreeSteinerTree}. 
\begin{theorem}
\label{thm:steiner from tough}Given a graph $G$ such that a set $X\subseteq V(G)$ is $\phi$-tough in G, there is an algorithm that computes a $(2/\phi+3)$-degree Steiner tree spanning $X$. The running time is $O(mn\log n)$. 
\end{theorem}

The weaker statement of this theorem was shown by Win \cite{win1989connection}, who gave  a non-algorithmic version of this theorem when $A=V$, i.e., a \emph{spanning} tree case. 

To prove \Cref{thm:steiner from tough}, we apply the additive-1 approximation algorithm by \cite{FurerR94} for finding a minimum degree Steiner trees. 
The structural guarantees of their algorithm can be summarized as follows. 
\begin{lem}
[\cite{FurerR94}]There is an algorithm that, given a graph $G$ with $n$ vertices and $m$ edges and a vertex $X$, in $O(mn\log n)$ time returns a tree $T$ in $G$ with maximum degree $\Delta$ and a vertex set $B\subseteq V(T)$ such that
\begin{enumerate}
\item Every leaf of $T$ is a vertex in $X$,
\item Each vertex $v\in B$ has degree $\deg_{T}(v)\ge\Delta-1$, and
\item \label{enu:bad set preserve conn}For any two vertices $s,t\in X-B$, $s$ and $t$ are connected in $G-B$ if and only if they are connected in $T-B$. 
\end{enumerate}
\end{lem}

The last property says that the connectivity between vertices of $X$ in $G-B$ and $T-B$ are preserved exactly. 
\begin{proof}
[Proof of \Cref{thm:steiner from tough}]We claim that $|B|<2c_{T-B}(X)/(\Delta-3)$. Since $c_{G-B}(X)=c_{T-B}(X)$ because of Property \ref{enu:bad set preserve conn} and $X$ is $\phi$-tough, it must be that $\frac{2}{(\Delta-3)}>\phi$, meaning that $\Delta<2/\phi+3$. That is, $T$ is a $(2/\phi+3)$-degree Steiner tree spanning $X$ as desired and $T$ can be computed in $O(mn\log n)$ time. 

Now we prove the claim. Just for analysis, consider the set $E_{B}$ of all edges incident to $B$. If we delete \emph{edges} in $E_{B}$, observe that $T-E_{B}$ contains exactly $|E_{B}|+1$ connected components. We classify these connected components in $T-E_{B}$ into three types: 
\begin{enumerate}
\item \textbf{(Trivial components):} Components that contain of a single vertex $v\in B$,
\item \textbf{(Internal components):} Components that contains no vertex in $B$ or $X$,
\item \textbf{(Leaf components):} Components that contains no vertex in $B$, but contains a vertex in $X$.
\end{enumerate}
The number trivial components is clearly $|B|$. The number of internal components is at most $|E_{B}|/2$. This is because each deleted edge in $E_{B}$ has at most one endpoint in non-trivial components (internal or leaf components.) But, crucially, each internal component must be incident to at least $2$ deleted edges (otherwise, it will contain a leaf, which is a vertex in $X$). The number of leaf components is precisely $c_{T-B}(X)$. Since the total number of components is $|E_{B}|+1$, we have that
\[
|E_{B}|+1\le|B|+|E_{B}|/2+c_{T-B}(X).
\]
So, 
\[
c_{T-B}(X)\ge|E_{B}|/2-|B|+1\ge |B|(\Delta-1)/2-|B|+1>|B|(\Delta-3)/2
\]
where the second inequality is because $B$ has minimum degree $\Delta-1$.
\end{proof}

%% file: b_singleton_detection.tex
\section{Improved Singleton-Detection Scheme: Proof of \cref{lem:singleton}}\label{app:singleton}

\newcommand{\Sample}{\textsf{Sample}}
\newcommand{\Sig}{\textsf{Sig}}

Let $a$ be a uniformly random odd $w$-bit integer and
$t$ a uniformly random $w$-bit integer.  
Thorup~\cite{thorup2018sample} proved that the function $\Sample : [2^w] \to \{0,1\}$, $\Sample_{a,t}(x) = \ind{ax \mod 2^w < t}$ 
(written \texttt{a*x < t} in C++ notation) 
is a \emph{distinguisher} with probability 1/8.
In other words, for any non-empty set $S\subset [2^w]$,
\[
\Pr\left(\sum_{x\in S}\Sample_{a,t}(x) \equiv 1 \pmod 2\right) \geq 1/8.
\]
We identify the edge-set $E$ with $\{0,1\}^{2\log n}$.
Suppose $\uid : \{0,1\}^{2\log n}\to \{0,1\}^{2\log n} \times \{0,1\}^{c\log n}$ is defined so that
$\uid(x) = (x,\Sig(x))$, where
the signature \
\[
\Sig(x) = \left(\Sample_{a_1,t_1}(x),\ldots,\Sample_{a_{c\log n},t_{c\log n}(x)}\right)
\]
consists of $c\log n$ independent invocations of the distinguisher.
In this case $\uid$ would satisfy
the singleton-detection properties of \cref{lem:singleton},
but with a $O(\log^2 n)$-bit seed.  To see why, 
consider an arbitrary set $S\subset E$ of edges.
Let $e^* = \bigoplus_{e\in S} e$ be the XOR of the names 
of all edges in $S$.  If $|S|>2$, then the probability that
we mistakenly believe $S$ to be the 
singleton set $\{e^*\}$ is
\begin{align*}
\Pr\left(\bigoplus_{e\in S} \uid(e) = (e^*, \Sig(e^*))\right)
= 
\Pr\left(\bigoplus_{e\in (S\oplus \{e^*\})} \uid(e) = (\mathbf{0},\mathbf{0})\right)
\leq (7/8)^{c\log n} = 1/\poly(n).
\end{align*}

Like~\cite{GibbKKT15,GhaffariP16}, 
this scheme uses a $O(\log^2 n)$-bit seed, but it is
simpler than both of \cite{GibbKKT15,Ghaffari16}.
The seed-length can be reduced to $O(\log n)$ bits 
by taking a random walk on an expander of length $O(\log n)$.

\begin{theorem}[{Hitting Property of Random Walks; see Vadhan~\cite[Theorem 4.17]{Vadhan12}}]
If $G$ is a regular digraph with spectral 
expansion $1-\lambda$, then for any $B \subset V(G)$
of density $\mu$, the probability that a random 
walk $(v_1,\ldots,v_t)$ of $t-1$ steps in $G$ starting 
at a uniformly random vertex $v_1$ always 
remains in $B$ is
\[
\Pr\left(\bigwedge_{i\in [t]} (v_i\in B)\right)
\leq (\mu+\lambda(1-\mu))^t.
\]
\end{theorem}

In our case $V(G)$ corresponds to the set of possible seeds
$(a,t)$ for $\Sample$, so $|V(G)|=\poly(n)$,
and $\mu=7/8$.
Whenever $\lambda<1$, $t=O(\log n)$ suffices to 
reduce the error probability to $1/\poly(n)$.
If $G$ is $d$-regular, the cost of encoding the random walk
is $O(\log n) + t\log d = O(\log n)$ bits.

%% file: main.bbl
\newcommand{\etalchar}[1]{$^{#1}$}

%% file: main.bbl
\begin{thebibliography}{AALOG18}

\bibitem[AAK{\etalchar{+}}06]{AbiteboulAKMR06}
Serge Abiteboul, Stephen Alstrup, Haim Kaplan, Tova Milo, and Theis Rauhe.
\newblock Compact labeling scheme for ancestor queries.
\newblock {\em {SIAM} J. Comput.}, 35(6):1295--1309, 2006.

\bibitem[AALOG18]{alev2018graph}
Vedat~Levi Alev, Nima Anari, Lap~Chi Lau, and Shayan Oveis~Gharan.
\newblock Graph clustering using effective resistance.
\newblock In {\em 9th Innovations in Theoretical Computer Science Conference (ITCS 2018)}. Schloss-Dagstuhl-Leibniz Zentrum f{\"u}r Informatik, 2018.

\bibitem[ABR05]{AlstrupBR05}
Stephen Alstrup, Philip Bille, and Theis Rauhe.
\newblock Labeling schemes for small distances in trees.
\newblock {\em {SIAM} J. Discret. Math.}, 19(2):448--462, 2005.

\bibitem[ACGP16]{AbrahamCGP16}
Ittai Abraham, Shiri Chechik, Cyril Gavoille, and David Peleg.
\newblock Forbidden-set distance labels for graphs of bounded doubling dimension.
\newblock {\em {ACM} Trans. Algorithms}, 12(2):22:1--22:17, 2016.

\bibitem[ADK17]{AlstrupDK17}
Stephen Alstrup, S{\o}ren Dahlgaard, and Mathias B{\ae}k~Tejs Knudsen.
\newblock Optimal induced universal graphs and adjacency labeling for trees.
\newblock {\em J. {ACM}}, 64(4):27:1--27:22, 2017.

\bibitem[AG11]{AbrahamG11}
Ittai Abraham and Cyril Gavoille.
\newblock On approximate distance labels and routing schemes with affine stretch.
\newblock In {\em Proceedings 25th International Symposium on Distributed Computing ({DISC})}, pages 404--415, 2011.

\bibitem[AGHP16a]{AlstrupGHP16b}
Stephen Alstrup, Cyril Gavoille, Esben~Bistrup Halvorsen, and Holger Petersen.
\newblock Simpler, faster and shorter labels for distances in graphs.
\newblock In {\em Proceedings of the Twenty-Seventh Annual {ACM-SIAM} Symposium on Discrete Algorithms ({SODA})}, pages 338--350, 2016.

\bibitem[AGHP16b]{AlstrupGHP16a}
Stephen Alstrup, Inge~Li G{\o}rtz, Esben~Bistrup Halvorsen, and Ely Porat.
\newblock Distance labeling schemes for trees.
\newblock In {\em Proceedings 43rd International Colloquium on Automata, Languages, and Programming (ICALP)}, volume~55 of {\em LIPIcs}, pages 132:1--132:16. Schloss Dagstuhl - Leibniz-Zentrum f{\"{u}}r Informatik, 2016.

\bibitem[AGM12]{AhnGM12}
Kook~J. Ahn, Supdipto Guha, and Andrew McGregor.
\newblock Analyzing graph structure via linear measurements.
\newblock In {\em Proceedings of the 23rd Annual ACM-SIAM Symposium on Discrete Algorithms (SODA)}, pages 459--467, 2012.

\bibitem[AHL14]{AlstrupHL14}
Stephen Alstrup, Esben~Bistrup Halvorsen, and Kasper~Green Larsen.
\newblock Near-optimal labeling schemes for nearest common ancestors.
\newblock In {\em Proceedings of the Twenty-Fifth Annual {ACM-SIAM} Symposium on Discrete Algorithms ({SODA})}, pages 972--982, 2014.

\bibitem[AKTZ19]{AlstrupKTZ19}
Stephen Alstrup, Haim Kaplan, Mikkel Thorup, and Uri Zwick.
\newblock Adjacency labeling schemes and induced-universal graphs.
\newblock {\em {SIAM} J. Discret. Math.}, 33(1):116--137, 2019.

\bibitem[AN19]{AbrahamN19}
Ittai Abraham and Ofer Neiman.
\newblock Using petal-decompositions to build a low stretch spanning tree.
\newblock {\em {SIAM} J. Comput.}, 48(2):227--248, 2019.

\bibitem[ARV09]{AroraRV09}
Sanjeev Arora, Satish Rao, and Umesh~V. Vazirani.
\newblock Expander flows, geometric embeddings and graph partitioning.
\newblock {\em J.~ACM}, 56(2), 2009.

\bibitem[BBS06]{bauer2006toughness}
Douglas Bauer, Hajo Broersma, and Edward Schmeichel.
\newblock Toughness in graphs--a survey.
\newblock {\em Graphs and Combinatorics}, 22:1--35, 2006.

\bibitem[BCG{\etalchar{+}}22]{Bar-NatanCGMW22}
Aviv Bar{-}Natan, Panagiotis Charalampopoulos, Pawel Gawrychowski, Shay Mozes, and Oren Weimann.
\newblock Fault-tolerant distance labeling for planar graphs.
\newblock {\em Theor. Comput. Sci.}, 918:48--59, 2022.

\bibitem[BCHR20]{BaswanaCHR20}
Surender Baswana, Keerti Choudhary, Moazzam Hussain, and Liam Roditty.
\newblock Approximate single-source fault tolerant shortest path.
\newblock {\em {ACM} Trans. Algorithms}, 16(4):44:1--44:22, 2020.

\bibitem[BF67]{BreuerF67}
Melvin~A Breuer and Jon Folkman.
\newblock An unexpected result in coding the vertices of a graph.
\newblock {\em Journal of Mathematical Analysis and Applications}, 20(3):583--600, 1967.

\bibitem[BGP22]{BonamyGP22}
Marthe Bonamy, Cyril Gavoille, and Michal Pilipczuk.
\newblock Shorter labeling schemes for planar graphs.
\newblock {\em {SIAM} J. Discret. Math.}, 36(3):2082--2099, 2022.

\bibitem[Bre66]{Breuer66}
Melvyl Breuer.
\newblock Coding the vertexes of a graph.
\newblock {\em IEEE Transactions on Information Theory}, 12(2):148--153, 1966.

\bibitem[CGKT08]{CourcelleGKT08}
Bruno Courcelle, Cyril Gavoille, Mamadou~Moustapha Kant{\'{e}}, and Andrew Twigg.
\newblock Connectivity check in 3-connected planar graphs with obstacles.
\newblock {\em Electron. Notes Discret. Math.}, 31:151--155, 2008.

\bibitem[Cho16]{choudhary2016optimal}
Keerti Choudhary.
\newblock An optimal dual fault tolerant reachability oracle.
\newblock In {\em 43rd International Colloquium on Automata, Languages, and Programming (ICALP 2016)}. Schloss-Dagstuhl-Leibniz Zentrum f{\"u}r Informatik, 2016.

\bibitem[Chv73]{chvatal1973tough}
Vasek Chv{\'{a}}tal.
\newblock Tough graphs and hamiltonian circuits.
\newblock {\em Discret. Math.}, 5(3):215--228, 1973.

\bibitem[CL06]{ChungL06}
Fan R.~K. Chung and Lincoln Lu.
\newblock Survey: Concentration inequalities and martingale inequalities: {A} survey.
\newblock {\em Internet Math.}, 3(1):79--127, 2006.

\bibitem[CMW23]{chechik2023optimal}
Shiri Chechik, Shay Mozes, and Oren Weimann.
\newblock $\tilde{O}$ptimal fault-tolerant reachability labeling in planar graphs.
\newblock {\em arXiv preprint arXiv:2307.07222}, 2023.

\bibitem[CPR11]{ChanPR11}
Timothy~M. Chan, Mihai P\v{a}tra\c{s}cu, and Liam Roditty.
\newblock Dynamic connectivity: Connecting to networks and geometry.
\newblock {\em SIAM J.~Comput.}, 40(2):333--349, 2011.

\bibitem[CT07]{CourcelleT07}
Bruno Courcelle and Andrew Twigg.
\newblock Compact forbidden-set routing.
\newblock In {\em Proceedings 24th Annual Symposium on Theoretical Aspects of Computer Science (STACS)}, volume 4393 of {\em Lecture Notes in Computer Science}, pages 37--48. Springer, 2007.

\bibitem[DP09]{DubhashiPanconesi09}
Devdatt~P. Dubhashi and Alessandro Panconesi.
\newblock {\em Concentration of Measure for the Analysis of Randomized Algorithms}.
\newblock Cambridge University Press, 2009.

\bibitem[DP20]{DuanP20}
Ran Duan and Seth Pettie.
\newblock Connectivity oracles for graphs subject to vertex failures.
\newblock {\em {SIAM} J. Comput.}, 49(6):1363--1396, 2020.

\bibitem[DP21]{DoryP21}
Michal Dory and Merav Parter.
\newblock Fault-tolerant labeling and compact routing schemes.
\newblock In {\em Proceedings of the 40th {ACM} Symposium on Principles of Distributed Computing (PODC)}, pages 445--455, 2021.

\bibitem[EJKS85]{enomoto1985toughness}
Hikoe Enomoto, Bill Jackson, Panagiotis Katerinis, and Akira Saito.
\newblock Toughness and the existence of k-factors.
\newblock {\em Journal of Graph Theory}, 9(1):87--95, 1985.

\bibitem[FHL05]{feige2005improved}
Uriel Feige, MohammadTaghi Hajiaghayi, and James~R Lee.
\newblock Improved approximation algorithms for minimum-weight vertex separators.
\newblock In {\em Proceedings of the thirty-seventh annual ACM symposium on Theory of computing}, pages 563--572, 2005.

\bibitem[FR94]{FurerR94}
Martin F{\"{u}}rer and Balaji Raghavachari.
\newblock Approximating the minimum-degree {S}teiner tree to within one of optimal.
\newblock {\em J.~Algor.}, 17(3):409--423, 1994.

\bibitem[Gha16]{Ghaffari16}
Mohsen Ghaffari.
\newblock An improved distributed algorithm for maximal independent set.
\newblock In {\em Proceedings 27th Annual {ACM-SIAM} Symposium on Discrete Algorithms (SODA)}, pages 270--277, 2016.

\bibitem[GKKT15]{GibbKKT15}
David Gibb, Bruce~M. Kapron, Valerie King, and Nolan Thorn.
\newblock Dynamic graph connectivity with improved worst case update time and sublinear space.
\newblock {\em CoRR}, abs/1509.06464, 2015.

\bibitem[GP16]{GhaffariP16}
Mohsen Ghaffari and Merav Parter.
\newblock {MST} in log-star rounds of congested clique.
\newblock In {\em Proceedings of the 35th {ACM} Symposium on Principles of Distributed Computing ({PODC})}, pages 19--28, 2016.

\bibitem[GPPR04]{GavoillePPR04}
Cyril Gavoille, David Peleg, St{\'{e}}phane P{\'{e}}rennes, and Ran Raz.
\newblock Distance labeling in graphs.
\newblock {\em J. Algorithms}, 53(1):85--112, 2004.

\bibitem[GRST21]{goranci2021expander}
Gramoz Goranci, Harald R{\"a}cke, Thatchaphol Saranurak, and Zihan Tan.
\newblock The expander hierarchy and its applications to dynamic graph algorithms.
\newblock In {\em Proceedings of the 2021 ACM-SIAM Symposium on Discrete Algorithms (SODA)}, pages 2212--2228. SIAM, 2021.

\bibitem[Gu21]{gu2021proof}
Xiaofeng Gu.
\newblock A proof of brouwer's toughness conjecture.
\newblock {\em SIAM Journal on Discrete Mathematics}, 35(2):948--952, 2021.

\bibitem[GU23]{GawrychowskiU23}
Pawel Gawrychowski and Przemyslaw Uznanski.
\newblock Better distance labeling for unweighted planar graphs.
\newblock {\em Algorithmica}, 85(6):1805--1823, 2023.

\bibitem[HKNS15]{HenzingerKNS15}
Monika Henzinger, Sebastian Krinninger, Danupon Nanongkai, and Thatchaphol Saranurak.
\newblock Unifying and strengthening hardness for dynamic problems via the online matrix-vector multiplication conjecture.
\newblock In {\em Proceedings of the 47th Annual {ACM} Symposium on Theory of Computing (STOC)}, pages 21--30, 2015.

\bibitem[HL09]{HsuL09}
Tai{-}Hsin Hsu and Hsueh{-}I Lu.
\newblock An optimal labeling for node connectivity.
\newblock In {\em Proceedings of 20th International Symposium on Algorithms and Computation (ISAAC)}, volume 5878 of {\em Lecture Notes in Computer Science}, pages 303--310. Springer, 2009.

\bibitem[IEWM23]{IzumiEWM23}
Taisuke Izumi, Yuval Emek, Tadashi Wadayama, and Toshimitsu Masuzawa.
\newblock Deterministic fault-tolerant connectivity labeling scheme with adaptive query processing time.
\newblock In {\em Proceedings of the 42nd ACM Symposium on Principles of Distributed Computing (PODC)}, 2023.

\bibitem[IN12]{IzsakN12}
Rani Izsak and Zeev Nutov.
\newblock A note on labeling schemes for graph connectivity.
\newblock {\em Inf. Process. Lett.}, 112(1-2):39--43, 2012.

\bibitem[KKKP04]{KatzKKP04}
Michal Katz, Nir~A. Katz, Amos Korman, and David Peleg.
\newblock Labeling schemes for flow and connectivity.
\newblock {\em {SIAM} J. Comput.}, 34(1):23--40, 2004.

\bibitem[KKM13]{KapronKM13}
Bruce~M. Kapron, Valerie King, and Ben Mountjoy.
\newblock Dynamic graph connectivity in polylogarithmic worst case time.
\newblock In {\em Proceedings of the 24th Annual ACM-SIAM Symposium on Discrete Algorithms (SODA)}, pages 1131--1142, 2013.

\bibitem[KNR92]{KannanNR92}
Sampath Kannan, Moni Naor, and Steven Rudich.
\newblock Implicit representation of graphs.
\newblock {\em {SIAM} J. Discret. Math.}, 5(4):596--603, 1992.

\bibitem[Kos23]{kosinas2023connectivity}
Evangelos Kosinas.
\newblock Connectivity queries under vertex failures: Not optimal, but practical.
\newblock In {\em 31st Annual European Symposium on Algorithms (ESA 2023)}. Schloss-Dagstuhl-Leibniz Zentrum f{\"u}r Informatik, 2023.

\bibitem[KP21]{KarthikP21}
{Karthik {C. S.}} and Merav Parter.
\newblock Deterministic replacement path covering.
\newblock In {\em Proceedings of the 32nd {ACM-SIAM} Symposium on Discrete Algorithms (SODA)}, pages 704--723, 2021.

\bibitem[KPP16]{KopelowitzPP16}
Tsvi Kopelowitz, Seth Pettie, and Ely Porat.
\newblock Higher lower bounds from the {3SUM} conjecture.
\newblock In {\em Proceedings of the 27th Annual {ACM-SIAM} Symposium on Discrete Algorithms ({SODA})}, pages 1272--1287, 2016.

\bibitem[LS22]{LongS22}
Yaowei Long and Thatchaphol Saranurak.
\newblock Near-optimal deterministic vertex-failure connectivity oracles.
\newblock In {\em Proceedings 63rd Annual {IEEE} Symposium on Foundations of Computer Science ({FOCS})}, pages 1002--1010, 2022.

\bibitem[LW24]{long2024better}
Yaowei Long and Yunfan Wang.
\newblock {Better Decremental and Fully Dynamic Sensitivity Oracles for Subgraph Connectivity}.
\newblock In {\em Proceedings 51st International Colloquium on Automata, Languages, and Programming (ICALP 2024)}, volume 297 of {\em Leibniz International Proceedings in Informatics (LIPIcs)}, pages 109:1--109:20, Dagstuhl, Germany, 2024. Schloss Dagstuhl -- Leibniz-Zentrum f{\"u}r Informatik.

\bibitem[MS18]{moshkovitz2018decomposing}
Guy Moshkovitz and Asaf Shapira.
\newblock Decomposing a graph into expanding subgraphs.
\newblock {\em Random Structures \& Algorithms}, 52(1):158--178, 2018.

\bibitem[NI92]{NagamochiI92}
Hiroshi Nagamochi and Toshihide Ibaraki.
\newblock A linear-time algorithm for finding a sparse $k$-connected spanning subgraph of a $k$-connected graph.
\newblock {\em Algorithmica}, 7(5{\&}6):583--596, 1992.

\bibitem[NSWN17]{NanongkaiSW17}
Danupon Nanongkai, Thatchaphol Saranurak, and Christian Wulff-Nilsen.
\newblock Dynamic minimum spanning forest with subpolynomial worst-case update time.
\newblock In {\em Proceedings 58th Annual IEEE Symposium on Foundations of Computer Science (FOCS)}, pages 950--961, 2017.

\bibitem[NY19]{NelsonY19}
Jelani Nelson and Huacheng Yu.
\newblock Optimal lower bounds for distributed and streaming spanning forest computation.
\newblock In {\em Proceedings of the 30th Annual {ACM-SIAM} Symposium on Discrete Algorithms ({SODA})}, pages 1844--1860, 2019.

\bibitem[PP22]{ParterP22a}
Merav Parter and Asaf Petruschka.
\newblock {\~{O}}ptimal dual vertex failure connectivity labels.
\newblock In {\em Proceedings of the 36th International Symposium on Distributed Computing ({DISC})}, volume 246 of {\em LIPIcs}, pages 32:1--32:19. Schloss Dagstuhl - Leibniz-Zentrum f{\"{u}}r Informatik, 2022.

\bibitem[PPP24]{ParterPP24}
Merav Parter, Asaf Petruschka, and Seth Pettie.
\newblock Connectivity labeling and routing with multiple vertex failures.
\newblock In {\em Proceedings of the 56th Annual {ACM} Symposium on Theory of Computing (STOC)}, pages 823--834, 2024.

\bibitem[PSS{\etalchar{+}}22]{PilipczukSSTV22}
Michal Pilipczuk, Nicole Schirrmacher, Sebastian Siebertz, Szymon Torunczyk, and Alexandre Vigny.
\newblock Algorithms and data structures for first-order logic with connectivity under vertex failures.
\newblock In {\em Proceedings of the 49th International Colloquium on Automata, Languages, and Programming ({ICALP})}, volume 229 of {\em LIPIcs}, pages 102:1--102:18. Schloss Dagstuhl - Leibniz-Zentrum f{\"{u}}r Informatik, 2022.

\bibitem[PSY22]{PettieSY22}
Seth Pettie, Thatchaphol Saranurak, and Longhui Yin.
\newblock Optimal vertex connectivity oracles.
\newblock In {\em Proceedings of the 54th Annual ACM Symposium on Theory of Computing (STOC)}, pages 151--161, 2022.

\bibitem[PT06]{PatrascuT06}
Mihai P{\v a}tra{\c s}cu and Mikkel Thorup.
\newblock Time-space trade-offs for predecessor search.
\newblock In {\em Proceedings of the 38th ACM Symposium on Theory of Computing (STOC)}, pages 232--240, 2006.

\bibitem[PT07]{PatrascuT07}
Mihai P{\v a}tra{\c s}cu and Mikkel Thorup.
\newblock Planning for fast connectivity updates.
\newblock In {\em Proceedings of the 48th IEEE Symposium on Foundations of Computer Science (FOCS)}, pages 263--271, 2007.

\bibitem[PT14]{PatrascuT14}
Mihai P\v{a}tra\c{s}cu and Mikkel Thorup.
\newblock Dynamic integer sets with optimal rank, select, and predecessor search.
\newblock In {\em Proceedings 55th Annual {IEEE} Symposium on Foundations of Computer Science ({FOCS})}, pages 166--175, 2014.

\bibitem[R{\"{a}}c08]{Racke08}
Harald R{\"{a}}cke.
\newblock Optimal hierarchical decompositions for congestion minimization in networks.
\newblock In {\em Proceedings 40th Annual {ACM} Symposium on Theory of Computing (STOC)}, pages 255--264, 2008.

\bibitem[RST14]{racke2014computing}
Harald R{\"a}cke, Chintan Shah, and Hanjo T{\"a}ubig.
\newblock Computing cut-based hierarchical decompositions in almost linear time.
\newblock In {\em Proceedings of the twenty-fifth annual ACM-SIAM symposium on Discrete algorithms}, pages 227--238. SIAM, 2014.

\bibitem[Tho18]{thorup2018sample}
Mikkel Thorup.
\newblock $\mathtt{Sample(x)= (a* x <= t)}$ is a distinguisher with probability 1/8.
\newblock {\em SIAM Journal on Computing}, 47(6):2510--2526, 2018.

\bibitem[TZ05]{TZ05}
Mikkel Thorup and Uri Zwick.
\newblock Approximate distance oracles.
\newblock {\em J.~ACM}, 52(1):1--24, 2005.

\bibitem[Vad12]{Vadhan12}
Salil~P. Vadhan.
\newblock Pseudorandomness.
\newblock {\em Foundations and Trends in Theoretical Computer Science}, 7(1–3):1--336, 2012.

\bibitem[vdBS19]{BrandS19}
Jan van~den Brand and Thatchaphol Saranurak.
\newblock Sensitive distance and reachability oracles for large batch updates.
\newblock In {\em Proceedings of the 60th Annual {IEEE} Symposium on Foundations of Computer Science ({FOCS})}, pages 424--435, 2019.

\bibitem[Win89]{win1989connection}
Sein Win.
\newblock On a connection between the existence of k-trees and the toughness of a graph.
\newblock {\em Graphs and Combinatorics}, 5(1):201--205, 1989.

\bibitem[Yu21]{Yu21}
Huacheng Yu.
\newblock Tight distributed sketching lower bound for connectivity.
\newblock In {\em Proceedings of the 32nd {ACM-SIAM} Symposium on Discrete Algorithms ({SODA})}, pages 1856--1873, 2021.

\end{thebibliography}
